\documentclass[11pt]{article}

\usepackage[normalem]{ulem} 
\usepackage{preamble}       

\usepackage[style=ieee-alphabetic]{biblatex}
\title{The Proof Analysis Problem\footnote{An extended abstract appeared in the 66th Annual Symposium on Foundations of Computer Science (FOCS 2025) \cite{AAdRK25-FOCS}.}}

\author{    Noel Arteche\footnote{Lund University and University of Copenhagen, \email{noel.arteche@cs.lth.se}} 
        \and Albert Atserias\footnote{Universitat Politècnica de Catalunya and Centre de Recerca Matemàtica, \email{atserias@cs.upc.edu}}
        \and Susanna F. de Rezende\footnote{Lund University, \email{susanna.rezende@cs.lth.se}}
        \and Erfan Khaniki\footnote{University of Oxford, \email{erfan.khaniki@cs.ox.ac.uk}}
    }

\date{}

\begin{document}
\maketitle

\thispagestyle{empty}

\begin{abstract}
Atserias and Müller (\emph{JACM}, 2020) proved that for every unsatisfiable CNF formula~$\varphi$, the formula $\Reff(\varphi)$---stating that \say{$\varphi$ has small Resolution refutations}---does not have subexponential-size Resolution refutations. Conversely, when $\varphi$ is satisfiable, Pudlák (\emph{TCS}, 2003) showed how to construct a polynomial-size Resolution refutation of~$\Reff(\varphi)$ given a satisfying assignment of $\varphi$. A question that had remained open is: \emph{do all short Resolution refutations of~$\Reff(\varphi)$ explicitly leak a satisfying assignment of~$\varphi$?}

We answer this question affirmatively by providing a polynomial-time algorithm that extracts a satisfying assignment for $\varphi$ given any short Resolution refutation of $\Reff(\varphi)$. The algorithm follows from a new feasibly constructive proof of the Atserias--Müller lower bound, formalizable in Cook's theory $\PVO$ of bounded arithmetic. This implies that Extended Frege can efficiently prove (a suitable formalization of the statement) that automating Resolution is $\NP$-hard.

Motivated by this algorithm, we introduce a new meta-computational problem concerning Resolution lower bounds: the
\emph{Proof Analysis Problem} ($\PAP$). For a fixed proof system~$Q$, the Proof Analysis Problem for $Q$ asks, given a CNF formula $\varphi$ and a $Q$-proof of a Resolution lower bound for $\varphi$, encoded as $\neg \Reff(\varphi)$, whether $\varphi$ is satisfiable. In contrast to the Proof Analysis Problem for Resolution, which is in $\P$, we prove that $\PAP$ for Extended Frege ($\EF$) is $\NP$-complete. In particular,~$\EF$ can prove Resolution lower bounds on satisfiable formulas without necessarily revealing a satisfying assignment.

Our results yield new insights into proof search and the meta-mathematics of Resolution lower bounds: (i) for every proof system that simulates $\EF$ as well as for Resolution, the system is (weakly) automatable if and only if it can be (weakly) automated exclusively on formulas stating Resolution lower bounds; (ii) we provide explicit $\Reff$ formulas that are exponentially hard for bounded-depth Frege systems; and (iii) for every strong enough theory of arithmetic $T$ we construct explicit unsatisfiable CNF formulas that are exponentially hard for Resolution but for which $T$ cannot prove even a quadratic Resolution lower bound. This latter result applies to arbitrarily strong theories like $\mathsf{PA}$ or $\mathsf{ZFC}$, and does not require any complexity-theoretic assumptions.
\end{abstract}

\newpage
\thispagestyle{empty}
\tableofcontents

\newpage
\setcounter{page}{1}

\section{Introduction}
\label{sec:intro}

The most natural computational problem arising in proof complexity is that of \emph{proof search}: what is the complexity of finding proofs? In the late 90s, the notion of \emph{automatability}, defined by Bonet, Pitassi, and Raz \cite{BPR00}, emerged as a central concept in the theory of propositional proof complexity. A proof system $Q$ is \emph{automatable} if there is a deterministic algorithm that finds a $Q$-proof of a formula $\varphi$ in time polynomial in the shortest one available. Except for Tree-like Resolution, which is automatable in quasi-polynomial time \cite{BP96}, no other non-trivial proof system is known to be automatable in polynomial or quasi-polynomial time.

Krajíček and Pudlák \cite{KP98} and Bonet, Pitassi, and Raz \cite{BPR00} proved that under standard worst-case number-theoretic assumptions in cryptography, strong proof systems like $\TC^0$-Frege and Extended Frege are not automatable. These results can be transferred to $\ComplexityFont{AC}^0$-Frege under slightly stronger hardness assumptions \cite{BDGMP04}, but it seems hard to push them further. Their proof techniques require some amount of basic number theory to be formalized in the system, something that is likely unworkable for Resolution. Since then, efforts focused on showing the hardness of automating Resolution and related weak systems \autocite{Pudlak03, AR08, GL10, AM11, PH11, Atserias13, BPT14, MPW19}, culminating in the final answer by Atserias and Müller \cite{AM20}, who proved that Resolution is not automatable unless~$\P = \NP$. This is the optimal hardness assumption since $\P = \NP$ implies the automatability of any proof system.

The technique used in~\cite{AM20} relies on the insight that Resolution \emph{cannot reason about its own lower bounds}. To every CNF formula $\varphi$, they associate a new formula $\Reff_{s}(\varphi)$ that encodes the statement \say{there is a size-$s$ Resolution refutation of $\varphi$}. As a tautology, $\neg \Reff_s(\varphi)$ is a natural propositional encoding of a Resolution lower bound. (We postpone to the preliminaries the details of the encoding of the $\Reff$ formula we use, where we also discuss previously studied variations.)

Pudlák~\cite{Pudlak03} had shown already in 2003 that whenever $\varphi$ is satisfiable, the formula~$\Reff_s(\varphi)$ is easily refutable by Resolution. On the other hand, Atserias and Müller~\cite{AM20} proved that whenever $\varphi$ is unsatisfiable, Resolution will require exponential size to refute $\Reff_s(\varphi)$, for $s$ being some fixed polynomial in the number of variables of $\varphi$, which we omit in the subscript for the rest of this introduction for the sake of clarity.

As a consequence, an automating algorithm running on formulas of the form $\Reff(\varphi)$ can be used to decide $\SAT$ in polynomial time: if~$\varphi \in \SAT$, then the algorithm must find a short refutation of $\Reff(\varphi)$ that Pudlák guarantees must exist; on the other hand, if $\varphi \not\in\SAT$, then there are no short refutations of $\Reff(\varphi)$, so we can stop the automating algorithm after a polynomial number of steps and be certain that $\varphi$ is unsatisfiable.

The proof strategy behind the Resolution lower bound on $\Reff$ formulas was soon adapted to a variety of weak proof systems (those where unconditional size lower bounds are known) \autocite{B20, BY24, G20, GKSMP20, dRGNPRS21, IR22, P23}, although the $\Reff$-like formulas used in these spin-off results are no longer natural lower-bound statements for these systems. In general, as pointed out by Pudlák, the question of whether a proof system can prove any of its own lower bounds \say{is widely open, except for Resolution, and we consider it more important than automatability} \cite[3]{Pudlak20}. It is currently open, for example, whether systems like constant-depth Frege have polynomial-size proofs of any of their own lower bounds.

The feat of the Resolution lower bound on $\Reff$ formulas, combined with the upper bound for satisfiable formulas, implies that Resolution can \emph{only} reason about \say{trivial} Resolution lower bounds (i.e., lower bounds on satisfiable formulas, which do not have refutations of any size). This highlights the upper bound construction as something even more remarkable, given that Resolution cannot efficiently prove its own soundness \cite{AB04}. Intriguingly, the known upper bound for $\Reff(\varphi)$ for satisfiable $\varphi$ crucially relies on Resolution guessing a satisfying assignment and using it as the backbone of the refutation. It is then natural to ask whether this is necessary:

\begin{enumerate}
	\item[($Q_1$)] \label{hello} \emph{Is it the case that whenever there is a short Resolution refutation of $\Reff(\varphi)$, the proof must \say{leak} a satisfying assignment?}
\end{enumerate}

By \say{leaking} we mean that a satisfying assignment is always readable in polynomial time from the given refutation. It is important to note that given a refutation $\pi$ of $\Reff(\varphi)$, one cannot simply restrict~$\pi$ in a way that corresponds to $\Reff(\varphi_{\restriction x_1 = 0})$ and $\Reff(\varphi_{\restriction x_1 = 1})$ to extract a satisfying assignment. This is because the variables of $\varphi$ are \emph{not} variables of $\Reff(\varphi)$. In principle, such a self-reducibility trick seems to require access to an automating algorithm, so that one could successively look for refutations of $\Reff(\varphi_{\restriction x_1 = 0})$ or $\Reff(\varphi_{\restriction x_1 = 1})$, then $\Reff(\varphi_{\restriction x_1 = b_1,x_2=0})$ or $\Reff(\varphi_{\restriction x_1 = b_1,x_2=1})$, and so on for all variables. Without access to an automating algorithm, it is not at all clear whether satisfying assignments can be extracted efficiently.

Yet another way of phrasing the lower bound on $\Reff$ formulas is to see it as the correctness proof of a \emph{lower bound analysis algorithm}. Namely, the result proves that there is an algorithm that given a Resolution refutation $\pi$ of~$\Reff(\varphi)$ decides whether $\varphi$ is satisfiable. The algorithm consists simply of checking whether $\pi$ is correct and short enough. The correctness of this procedure requires the proof of the lower bound, and this framing naturally leads to the following second natural question regarding $\Reff$ formulas:

\begin{enumerate}
    	\item[($Q_2$)] \emph{Is there an algorithm that given an Extended Frege proof $\pi$ of a Resolution lower bound $\neg \Reff(\varphi)$ decides in polynomial time whether $\varphi$ is satisfiable?}
\end{enumerate}

If the answer were affirmative, this would settle the long-standing open problem of the $\NP$-hardness of automating Extended Frege: given a CNF formula $\varphi$, construct the formula $\Reff(\varphi)$ and run the automating algorithm to find a short Extended Frege refutation. If an algorithm as the one asked for in ($Q_2$) existed, then we could apply it on this refutation to \emph{analyze} whether $\varphi$ is satisfiable. This distills the main idea in~\cite{AM20}, and the framing of the question in terms of algorithm design suggests that such an algorithm might well be possible \emph{without the need for unconditional Extended Frege lower bounds}.

Overall, the two questions ($Q_1$) and ($Q_2$) above hint at the central role of meta-mathematical lower bound statements in the theory of efficient proof search. We believe this calls for a deeper structural understanding that could lead to much-needed conceptual insights in automatability.

\subsection{Contributions}
\label{subsec:contributions}
Motivated by questions $(Q_1)$ and $(Q_2)$ above, we introduce a new meta-computational problem relating proofs and computation: the \emph{Proof Analysis Problem}.

For every propositional proof system $Q$, the \emph{Proof Analysis Problem for $Q$} ($\PAP_Q$) consists in analyzing Resolution lower bounds proven by $Q$. More formally, given a CNF formula $\varphi$ and a $Q$-proof of the Resolution lower bound encoded by the formula $\neg\Reff(\varphi)$, the task is to decide whether $\varphi$ is satisfiable.

\problemStatement{The Proof Analysis Problem for $Q$ ($\PAP_Q$)}{A CNF formula $\varphi$, a size parameter $s$ in unary and a $Q$-proof $\pi$ of the formula $\neg \Reff_{s}(\varphi)$.}{Is $\varphi$ satisfiable?}

The problem can be seen as the computational task of distinguishing \say{true} Resolution lower bounds (those where $\varphi$ is actually unsatisfiable) from \say{trivial} ones (those where the lower bound trivially holds because $\varphi$ is satisfiable and there is therefore no Resolution refutation, of any size). For those proof systems for which $\PAP_Q \in \P$, we say that $Q$ is \emph{analyzable}. We remark that the $\Reff$ formula in the definition of $\PAP_Q$ is always referring to Resolution refutations, while the proof system $Q$ where $\Reff(\varphi)$ is being derived can be arbitrarily strong.

For the case of Resolution itself, the problem $\PAP_\Res$ is easy to compute thanks to the lower bound on $\Reff$ formulas \cite{AM20}: if $\pi$ is a correct refutation of $\Reff(\varphi)$ and it is small, then $\varphi$ must be satisfiable. Until now, however, to the best of our knowledge this was the extent of what could be said about $\PAP$-like problems. In particular, we are not aware of any other upper or lower bounds on this problem for proofs systems other than Resolution.

In the language of $\PAP$, questions ($Q_1$) and ($Q_2$) can be neatly rephrased as follows:
\begin{enumerate} \itemsep=0pt
    \item[($Q_1$)] Does $\PAP_\Res$ admit a search-to-decision reduction?
    \item[($Q_2$)] Is Extended Frege analyzable? Namely, is $\PAP_\EF$ in $\P$?
\end{enumerate}

In this work we kick-start the systematic study of these Proof Analysis Problems and settle questions ($Q_1$) and ($Q_2$) above. This in turn yields a series of interesting consequences for the meta-mathematics of proof complexity lower bounds as well as proof search. We outline our results next.

\subsubsection{An algorithm for assignment extraction}
On the topic of question ($Q_1$), our main result is that the search version of $\PAP_\Res$ can be solved deterministically in polynomial time.

\begin{theorem}[Assignment extraction algorithm, informal]
    \label{thm:algo-informal}
    The search version of the Proof Analysis Problem for Resolution can be solved in deterministic polynomial time whenever the size parameter $s$ is at least $n^3$. That is, there is an algorithm that, given a CNF formula $\varphi$ over $n$ variables and $\poly(n)$ clauses and a Resolution refutation $\pi$ of $\Reff_{s}(\varphi)$ with $s \geq n^3$, extracts a satisfying assignment for $\varphi$ in time polynomial in $n$, $s$ and the size $|\pi|$ of $\pi$, whenever $\varphi$ is satisfiable.
\end{theorem}

The question can be stated more formally in terms of \emph{Levin reductions}. A Levin reduction between search problems $R_1$ and $R_2$ is a Karp-style many-one reduction that maps instances of $R_1$ to instances of $R_2$, with the additional property that it also maps solutions of $R_1$ to solutions of $R_2$, and back. The reduction $\varphi \mapsto \Reff(\varphi)$ showing that SAT reduces to the Proof Size Problem for Resolution with an exponential gap is clearly Levin in one direction: given a satisfying assignment of $\varphi$, Pudlák's construction can craft a refutation of $\Reff(\varphi)$. However, it had remained open whether 
this Levin reduction could be made two-way: given a refutation $\pi$ of $\Reff(\varphi)$, can one always extract a satisfying assignment to $\varphi$ in polynomial time?

For most if not all natural $\NP$-complete languages, the corresponding search problems tend to be complete under Levin reductions. However, the same decision problem could admit different search problems associated to it, and it is known that if $\P \neq \NP \cap \coNP$, then there are $\NP$ search problems that do not reduce to each other under Levin reductions, while their decision versions are $\NP$-complete (and hence do reduce to each other) under Karp reductions (see, for example, \cite{KM00, FFNR03}). To the best of our knowledge, until now the only natural examples of candidates to be $\NP$-hard search problems without Levin reductions were precisely certain problems arising in the context of meta-complexity. One is the Minimum Circuit Size Problem ($\MCSP$), for which Mazor and Pass \cite{MP24} recently proved that a certain gap version is \emph{not} $\NP$-complete under Levin reductions, assuming the existence of indistinguishability obfuscation (iO). The other candidate was precisely the reduction from SAT to the Proof Size Problem for Resolution. \Cref{thm:algo-informal} settles this, giving a two-way Levin reduction. 

The existence of the extraction algorithm answers question ($Q_1$) in the affirmative: Resolution refutations of $\Reff_{s}(\varphi)$ \emph{must leak a satisfying assignment}. This has a certain information-theoretic flavor: the fact that satisfying assignments can \emph{always} be efficiently extracted implies that the most succinct description of a refutation of $\Reff_{s}(\varphi)$ must include the description of a satisfying assignment for $\varphi$. We can make this precise in the language of Kolmogorov complexity using the framework of \emph{information efficiency} of \citeauthor{Krajicek22} \cite{Krajicek22}, who studied the minimum time-bounded Kolmogorov complexity ($\Kt$) of propositional proofs.

\begin{theorem}[Assignment extraction as information efficiency, informal]
    \label{thm:krajicek-informal}
    For every satisfiable CNF formula~$\varphi$ over $n$ variables and $\poly(n)$ clauses,
    \[ \info_\Res(\neg\Reff(\varphi)) \approx \min  \{ \Kt(\alpha \mid \varphi) \mid \varphi(\alpha)=1 \},\]
    where $\info_Q(\psi) \coloneq \min \{ \Kt(\pi \mid \psi) \mid \pi : Q \vdash \psi \}$ is Krajíček's information efficiency function. 
\end{theorem}

To the best of our knowledge, this is one of the first applications of Krajíček's framework.

\subsubsection{The Proof Analysis Problem for strong proof systems}
Motivated by ($Q_2$), we ask whether $\PAP$ is in $\P$ for strong proof systems. We conclude that the answer is likely negative by proving optimal conditional lower bounds in the form of $\NP$-hardness for every proof system that p-simulates Extended Frege ($\EF$).

\begin{theorem}[$\NP$-hardness of $\PAP_\EF$, informal]
    \label{thm:np-hardness-informal}
    For every propositional proof system $S$ that p-simulates Extended Frege, the Proof Analysis Problem for $S$ is $\NP$-complete.
\end{theorem}

This means that, unlike Resolution, strong proof systems are seemingly able to prove \say{trivial} Resolution lower bounds on satisfiable formulas without having to first prove that the underlying formula is satisfiable. In particular, this means Extended Frege is strong enough to \emph{obfuscate} the satisfying assignments. As a consequence, for strong proof systems like Extended Frege, one cannot hope to prove they are $\NP$-hard to automate following a strategy similar to that of~\cite{AM20}.

The idea behind the hardness reduction highlights that the question of whether a proof system is analyzable boils down to whether it can efficiently prove any non-trivial Resolution lower bounds. In fact, the situation is rather dramatic: as soon as a strong enough system can prove \emph{at least one} non-trivial Resolution lower bound, then the system immediately becomes hard to analyze.

\begin{theorem}[Equivalence of analyzability and $\Reff$ lower bounds, informal]
    \label{thm:equivalence-PAP-AM-informal}
    Let $S$ be a reasonably strong propositional proof system. Assuming that there is some $\delta >0 $ such that $\TSAT \not\in \ComplexityFont{io}\SIZE\left[2^{\delta n}\right]$, the following are equivalent:
    \begin{enumerate}[label=(\roman*)] \itemsep=0pt
        \item \label{item:PAP}  the proof system $S$ is analyzable;
        \item \label{item:AM} for every unsatisfiable 3-CNF formula $\varphi$ over $n$ variables, every $S$-proof of the formula $\neg \Reff_{\poly(n)}(\varphi)$ requires size at least $2^{\Omega(n)}$.
    \end{enumerate}
\end{theorem}

\subsubsection{Formalization of the Atserias--Müller lower bound in $\PVO$}
The inspiration for why the extraction algorithm in \Cref{thm:algo-informal} might 
exist in the first place comes from \emph{witnessing theorems} in bounded arithmetic. We work here with Cook's theory $\PVO$ and Buss's $\SOT$, which are first-order theories of arithmetic formalizing polynomial-time reasoning. In these theories, if a statement of the form $\forall x \exists y \varphi(x, y)$ with a low-complexity $\varphi(x,y)$ is provable in the theory, then there exists a polynomial-time algorithm that \emph{witnesses}~$y$ given~$x$. This implies, in particular, that if a problem is proven $\NP$-hard in one of these theories, then the reduction will be a Levin reduction.

The key observation for us is that the statement of the lower bound is itself of this form, a $\forall\Sigma^b_1$ sentence:
\begin{center}
    \say{\emph{for every} formula $\varphi$ and \emph{every} Resolution refutation $\pi$ of $\Reff(\varphi)$, \\
    \emph{there exists} a satisfying assignment for $\varphi$, or else $\pi$ is large.} 
\end{center}
Thus, if the previous statement were provable in $\PVO$, we would get a polynomial-time function extracting satisfying assignments given $\varphi$ and $\pi$.

While the extraction algorithm presented in \Cref{thm:algo-informal} is given directly in natural language, it is still worth formalizing the lower bound in bounded arithmetic to obtain a variety of applications.

\begin{theorem}[Atserias--Müller lower bound \cite{AM20} in $\PVO$, informal]
    \label{thm:AM-in-PV-informal}
    The theory $\PVO$ proves the statement that for every CNF formula $\varphi$ over $n$ variables and every size parameter $s\in \bbN$, if $\varphi$ is unsatisfiable and $\pi$ is a correct Resolution refutation of $\Reff_{s}(\varphi)$, then $|\pi| \geq 2^{\Omega(s/n^2)}$.
\end{theorem}

Formalizations in bounded arithmetic tend to be particularly interesting when they lead to new proofs of known statements. This has been the case, for example, with Razborov's formalizations of circuit lower bounds leading to the now-famous proof of Håstad's switching lemma via a simpler counting argument \cite{Razborov95}. Remarkably, the method introduced by Razborov to formalize the switching lemma is recognized for enabling proofs to at least two major conjectures in combinatorics \cite{AlweissLovettWuZhang2021, ParkPham2024}.
Another example is a recent new proof of the Schwartz--Zippel lemma \cite{AT24}, proven via a hybrid argument formalizable in $\SOT$. Our formalization of the lower bound on $\Reff$ formulas also relies on a new proof. We elaborate on this in the technical overview.

We remark that our bound of the form $2^{\Omega(s/n^2)}$ is slightly worse than the original one, which we state here for convenience.

\begin{theorem}[Atserias--Müller lower bound \cite{AM20}]
    \label{thm:AM-original}
    For every CNF formula $\varphi$ over $n$ variables and every size parameter $s \in \bbN$, if $\varphi$ is unsatisfiable and $\pi$ is a correct Resolution refutation of $\Reff_s(\varphi)$, then $|\pi| \geq 2^{\Omega(s/n)}$.
\end{theorem}

The difference in the bound means that we can only show that $\PVO$ proves hardness of $\Reff_{s}(\varphi)$ for $s \geq n^3$. We leave it open whether $\PVO$ can achieve the original $2^{\Omega(s/n)}$ bound via a different argument. In any case, this is not particularly important for our applications. We comment on this further in the technical overview.

\subsubsection{Formalization of Pudlák's upper bound in Resolution}

We complement the formalization of the lower bound with a formalization of the upper bound \cite{Pudlak03}, showing that there are short refutations of $\Reff(\varphi)$ whenever $\varphi$ is satisfiable. This construction can be carried out by a constant-depth circuit and could be formalized in $\SOT$, but certainly also in much weaker theories. We prove the somewhat surprising fact that the construction can be proven correct in Resolution itself.

\begin{theorem}[Pudlák's upper bound \cite{Pudlak03} in Resolution, informal] \label{thm:ub-in-res-informal}
    There is a polynomial-size depth-2 Boolean circuit~$P(\alpha, \varphi, s)$ of fan-in 2 that given a CNF formula $\varphi$, a satisfying assignment $\alpha$, and $s \in \bbN$, outputs a Resolution refutation $\pi$ of $\Reff_s(\varphi)$. Furthermore, the correctness of this circuit $P$ has polynomial-size proofs in Resolution.
\end{theorem}

That is, not only Resolution has short refutations of $\Reff(\varphi)$ when $\varphi$ is satisfiable: Resolution can show that the circuits generating these refutations from satisfying assignments are correct. This, again, is in striking contrast with the fact that Resolution does not have small proofs of its own soundness \cite{AB04}.

\subsubsection{Propositional fragments of Atserias--Müller: automatability in terms of $\Reff$ formulas}
The main consequence of the extraction algorithm together with its formalization in bounded arithmetic is the following precise characterization theorem relating the provability of a formula $\neg \varphi$ to the provability of  the formula $\neg \Reff(\Reff(\varphi))$. (For the sake of clarity, we ignore for now the exact size parameters of the $\Reff$ formulas, which are always some fixed polynomials; in general, when we write $S \vdash_{\mathrm{poly}} \varphi$ we mean that $S$ has polynomial-size proofs of $\varphi$, and by \say{reasonable proof system} we mean essentially that the system is closed under \emph{modus ponens}.)

\begin{theorem}[Propositional fragments of Atserias--Müller, informal]
\label{thm:characterization-informal}
    Let $S$ be a reasonable propositional proof system that simulates Extended Frege. Then, for every sequence $\{ \varphi_n \}_{n\in\bbN}$ of unsatisfiable CNF formulas,
    \[ S \vdash_{\mathrm{poly}} \neg \varphi_n \qquad \text{ if and only if } \qquad S\vdash_{\mathrm{poly}}  \neg \Reff(\Reff(\varphi_n)).\]
\end{theorem}

The lower bound on $\Reff$ formulas says that for every unsatisfiable $\varphi$, the corresponding $\Reff(\varphi)$ is hard for Resolution, making $\Reff(\Reff(\varphi))$ unsatisfiable. The latter encodes the statement \say{$\Reff(\varphi)$ is hard for Resolution}, and our theorem shows that when restricted to the reasoning power of a specific proof system $S$, such a lower bound has small proofs if, and only if, $S$ has short proofs of the unsatisfiability of $\varphi$ in the first place. That is, the \emph{fragment} of the Atserias--Müller lower bound that has short proofs in $S$ is precisely the one corresponding to the formulas that $S$ can prove unsatisfiable with short proofs.

This characterization is surprisingly tight and has consequences for automatability and proof search. Since we can relate the proof size of $\varphi$ in $S$ to the proof size of $\Reff(\Reff(\varphi))$, this means that looking for proofs of $\Reff(\Reff(\varphi))$ can be a proxy for searching for proofs of $\varphi$.

\begin{theorem}[Automatability in terms of $\Reff$ formulas, informal]
    \label{thm:aut-for-Ref-informal}
    For every reasonable proof system $S$ that simulates Extended Frege as well as for Resolution itself,
    \begin{enumerate}[label=(\roman*)] \itemsep=0pt
        \item $S$ is automatable if, and only if, $S$ is automatable exclusively on $\Reff$ formulas;
        \item $S$ is weakly automatable if, and only if, $S$ is weakly automatable exclusively on $\Reff$ formulas.
    \end{enumerate}
\end{theorem}

We remark again that these $\Reff$ formulas are always talking about Resolution, not about $S$. That is, for every strong enough proof system, \emph{efficient proof search over all tautologies is equivalent to efficient proof search over Resolution lower bounds}.

Until now no such general structural result was known that related proof search generally to proof search for a particular class of formulas. This goes in line with a question of \citeauthor{PS22} \cite{PS22}, who asked whether automating a proof system on truth-table tautologies (i.e., formulas stating circuit lower bounds) implies the automatability of the system on all tautologies. We have proved that this is the case for the class of formulas stating Resolution lower bounds in place of truth-table tautologies.

\subsubsection{Unprovability of Resolution lower bounds}
\Cref{thm:ub-in-res-informal}, together with \Cref{thm:AM-original}, further imply that true Resolution lower bounds can be essentially arbitrarily hard to prove. Namely, if $S$ is a propositional proof system where $\{ \varphi_n\}_{n\in\bbN}$ is a sequence of formulas that $S$ cannot refute in polynomial size, then $S$ cannot refute $\{ \Reff(\Reff(\varphi_n))\}_{n\in\bbN}$ either.

\begin{theorem}[Propositional unprovability of Resolution lower bounds, informal]
    \label{thm:prop-unprovable-informal}
    Let $Q$ be a reasonable propositional proof system that simulates Resolution. If $\{ \varphi_n\}_{n\in\bbN}$ is a sequence of hard unsatisfiable CNF formulas for $Q$, where $\varphi_n$ has $n$ variables and size $|\varphi_n|=\poly(n)$, then
    \begin{enumerate}[label=(\roman*)] \itemsep=0pt
        \item the formulas $\Reff_{n^2}(\varphi_n)$ over $N = \poly(n)$ variables are all unsatisfiable and require size $2^{N^{\Omega(1)}}$ to be refuted in Resolution;
        \item yet, $Q$ does not have polynomial-size refutations of the formulas $\Reff_{N^2}(\Reff_{n^2}(\varphi_n))$ stating quadratic lower bounds on $\Reff_{n^2}(\varphi_n)$.
    \end{enumerate}
\end{theorem}

There is nothing special about quadratic lower bounds being unprovable---one can get arbitrarily small polynomial lower bounds by tweaking the encoding. See the discussion after \Cref{thm:fo-unprovable-informal}.

We note that Iwama showed in 1997 that the Proof Size Problem for Resolution is $\NP$-complete \cite{Iwama97}. This means, in particular, that its complement in $\coNP$-complete and hence, unless $\NP = \coNP$, no propositional proof system can efficiently derive \emph{all} tautological $\Reff$ formulas (i.e., all true Resolution lower bounds), or else there would be a polynomially bounded proof system. While this has a similar flavor to our result, our theorem is different in at least two aspects. First, from an explicit family of hard tautologies we obtain an explicit family of hard $\Reff$ formulas for the system, in a generic way. Second, the parameters are essentially optimal: we identify a sequence of unsatisfiable formulas for which an exponential (and hence maximal) Resolution lower bound holds ---while the system $Q$ in question cannot even prove a quadratic lower bound.

As a corollary of \Cref{thm:prop-unprovable-informal}, for example, we get the first explicit lower bounds for $\Reff$ formulas in bounded-depth Frege systems.

\begin{corollary}[Hard $\Reff$ formulas for bounded-depth Frege, informal]
    \label{cor:PHP-Ref-ACO-informal}
    For every $d\leq O(\log n / \log \log n)$, the formulas $\Reff(\Reff(\PHP_n))$ are all unsatisfiable but require size $\exp({\Omega(n^{1/(2d+1)})})$ to be refuted in depth-$d$ Frege systems.
\end{corollary}

We build here on the best-known PHP lower bounds by \citeauthor{Hastad23} \cite{Hastad23}, but even better bounds could be obtained by applying the same argument to Tseitin formulas using the results of \citeauthor{HR22} \cite{HR22}. This corollary contrasts again with the open question of Pudlák \cite{Pudlak20} about whether constant-depth Frege proves any of its own lower bounds.

Finally, using similar ideas, we obtain unconditional independence results for first-order theories of arithmetic.

\begin{theorem}[First-order unprovability of Resolution lower bounds, informal]
    \label{thm:fo-unprovable-informal}
        Let $T$ be a consistent first-order theory extending Robinson Arithmetic by a set of polynomial-time recognizable axioms. Then, there exists a sequence of unsatisfiable propositional formulas $\{ \psi_{N} \}_{{N} \in \bbN}$ described uniformly by a polynomial-time algorithm, where $\psi_N$ has $N$ variables, such that
    \begin{enumerate}[label=(\roman*)] \itemsep=0pt
        \item[(i)] Resolution refutations of the formula $\psi_{N}$ require size $2^{N^{\Omega(1)}}$;
        \item[(ii)] there exists $c > 0$ such that the theory $T$ cannot prove $\Omega(N^c)$ lower bounds on the Resolution size of
        these refutations; that is, there is $N_0 \in \bbN$ such that the lower bound expressed by the first-order sentence
        $\forall N\forall \pi \left(N > N_0 \land \operatorname{Ref}_{\Res}(\psi_N,\pi) \to |\pi| > N^c\right)$ is unprovable in $T$.
    \end{enumerate}
\end{theorem}

We remark that the theory $T$ in this theorem can be arbitrary strong. This implies that, unconditionally, theories like Peano Arithmetic ($\PA$) cannot prove all true Resolution lower bounds. The same ideas apply to Zermelo--Fraenkel Set Theory ($\ZFC$) and similarly powerful formal systems.

We also note that the constant $c >0$ in the exponent of the unprovable lower bound depends on the definition of $\psi_N$. In general, one can alter $\psi_N$ to get an unprovable lower bound of the form $\Omega(N^{c})$ for any fixed constant $c >0$.

\subsection{Technical overview}

Next we provide a technical overview of the main proof ideas and how these are combined to yield our main results and corollaries.

\subsubsection{Assignment extraction: techniques}
We obtain the extraction algorithm in \Cref{thm:algo-informal} by derandomizing the proof of the Resolution lower bound for the $\Reff$ formulas. The original proof revolves around the concept of \emph{block-width} (called \emph{index-width} in \cite{AM20}) in Resolution refutations of $\Reff_s(\varphi)$. The variables of the formula are arranged into $s$ \emph{blocks}, each encoding a clause in the purported refutation of size $s$. The block-width of a refutation $\pi$ is then the largest number of blocks mentioned in a clause of the refutation $\pi$. The proof proceeded in two steps:
\begin{enumerate} \itemsep=0pt
    \item derive a \emph{block-width lower bound}, showing that if $\varphi \not\in \SAT$, then the block-width of any refutation of $\Reff_s(\varphi)$ must be large;
    \item by a \emph{random restriction argument}, argue that if the refutation $\pi$ is small, there exists a restriction that makes the block-width of the restricted refutation small, contradicting the previous point.
\end{enumerate}

Our algorithm works by following these steps in reverse. First, given a refutation $\pi$ from which we want to extract a satisfying assignment, instead of sampling a restriction at random from a specific distribution,  we construct a restriction \emph{deterministically} in a greedy fashion, tailored to the specifics of $\pi$. This is reminiscent of the kind of greedy deterministic restrictions used by \citeauthor{CP90} \cite{CP90} to formalize Haken's lower bound for the pigeonhole principle in bounded arithmetic, and more broadly in the style of \citeauthor{BP96} \cite{BP96}, \citeauthor{CEI96} \cite{CEI96} and \citeauthor{BW01} \cite{BW01}. Our algorithm runs in deterministic polynomial time and always succeeds in finding a restriction that reduces the block-width to $O\big(\sqrt{s\log |\pi|}\big)$.

In the second step, we look at the proof of the block-width lower bound and interpret it as a Prover-Delayer game in the style of \citeauthor{AD08} \cite{AD08}. The Prover issues queries about the values of the variables of $\Reff_s(\varphi)$, or forgets previously recorded such values, and the Delayer replies following a concrete strategy that allows them to keep playing until a large number of blocks appear queried. Our algorithm traverses the Resolution refutation guided by the Delayer's strategy in the Prover-Delayer game. We can then prove that this Delayer's strategy will, after a polynomial number of steps, reach either 
\begin{enumerate}[label=(\alph*)] \itemsep=0pt
    \item a clause of \emph{high block-width}, or
    \item a clause encoding a \emph{satisfying assignment} to $\varphi$.
\end{enumerate}
Since the greedy deterministic restriction in the first step made sure the block-width is small, the Delayer will be guaranteed to find a satisfying assignment.

We note that our deterministic restriction only achieves a reduction of block-width to $O\big(\sqrt{s\log |\pi|}\big)$, while using a random restriction one could achieve up to $O(\log |\pi|)$. In fact, if one allows randomness in the extraction algorithm, then an argument similar to the random restriction of~\cite{AM20} yields a zero-error probabilistic polynomial-time extraction algorithm that works even when the refutation $\pi$ being analyzed is for the formula $\Reff_{n^2}(\varphi)$. In contrast, the price to pay for determinism is that the size parameter $s$ should be at least $n^3$.

\subsubsection{$\NP$-hardness of the Proof Analysis Problem for Extended Frege}
\label{subsubsec:NP-hardness-proof-MCSP}

The idea behind the hardness proof in \Cref{thm:np-hardness-informal} boils down to the following crucial insight: as soon as a proof system $S$ proves \emph{at lest one} true Resolution lower bounds for some specific formula $\psi$, then $S$ can also prove lower bounds on formulas of the form $\varphi \lor \psi$ for any $\varphi$ by completely ignoring $\varphi$. However, since $\varphi$ might be satisfiable, $\varphi \lor \psi$ might be overall satisfiable, yet $S$ will prove a lower bound \say{agnostically}, without revealing a satisfying assignment.

In more detail, suppose the proof system $S$ (in our case, Extended Frege) can uniformly prove Haken's lower bound for the pigeonhole principle formulas \cite{Haken85}. That is, there is a polynomial-time algorithm that given a number $n$ and $t$ outputs a $S$-proof $\pi$ such that
\[\pi: S \vdash \neg\Reff_{t}(\PHP_n).\]
Now, let $\varphi$ be a propositional formula over $n$ variables different from those of $\PHP_n$. If the system $S$ is strong enough, it will be able to uniformly construct a proof $\pi'$ such that
\[\pi': S \vdash \neg\Reff_{t}(\varphi \lor \PHP_n).\]

The proof $\pi'$ goes roughly as follows: if $\varphi$ is satisfiable then so is $\varphi \lor \PHP_n$, but $S$ being strong enough to prove the soundness of Resolution, it can conclude that there will be no refutation of any size, let alone of size $t$; if, on the other hand, $\varphi$ is unsatisfiable and there is a size-$t$ refutation of $\varphi \lor \PHP_n$, then any restriction to the variables of $\varphi$ necessarily reduces $\varphi \lor \PHP_n$ to just $\PHP_n$. Assuming $S$ is strong enough to prove that Resolution is closed under restrictions, $S$ would conclude that there is a size-$t$ refutation of $\PHP_n$ ---contradicting Haken's lower bound, which we assumed can itself be proven in $S$.

Observe how the previous argument performs a proof by cases that works regardless of the satisfiability of $\varphi$. In particular, even in the case that a particular $\varphi$ is satisfiable, an explicit satisfying assignment never appears in the proof ---we only reason about a hypothetical one. For the argument to go through for $S = \EF$, everything we need is for $\EF$ to prove Haken's lower bound, which was already formalized in this system by Cook and Pitassi \cite{CP90}.

\subsubsection{Formalization of the upper and lower bounds}

A large part of our technical contribution consists in formalizing the proofs leading to the $\NP$-hardness of automating Resolution. We summarize below some of the challenges encountered, and the solutions devised.

\paragraph{The Atserias--Müller lower bound in $\PVO$.}
Different proofs of the lower bound exist in the literature, but none of them seem directly formalizable in $\PVO$. The original proof has the caveat of the random restriction argument, which might be formalizable in \citeauthor{JerabekStack}'s theory $\APC$, but likely not in $\PVO$. In addition, the block-width lower bound is proven by relating small refutations to the canonical exponential-size tree-like refutation of any formula, which could be hard to reason about in bounded theories.

In the work of de Rezende, Göös, Nordström, Pitassi, Robere, and Sokolov \cite{dRGNPRS21} two alternative proofs were presented. The first proof consists of a random restriction followed by a block-width lower bound proven via a reduction to the retraction weak pigeonhole principle. The random restriction presents the same formalization issues as the original proof, and the block-width reduction is equally problematic: the decision tree reduction they use has low depth, which is necessary to transfer the lower bounds on (block-)width, but the size of the decision tree itself seems to be super-polynomial, and hence cannot be reasoned about in $\PVO$. The second proof of \citeauthor{dRGNPRS21} uses this same block-width lower bound followed by the size-width trade-offs of \citeauthor{BW01} \cite{BW01}. The block-width lower bound is still problematic, of course, but in addition to this, the statement of \citeauthor{BW01}---\say{for every small Resolution refutation, there exists another Resolution refutation in small width}---is itself impossible to formalize in bounded arithmetic. The reason for this is that, as demonstrated by Thapen \cite{Thapen16}, in general these narrow proofs can require super-polynomial size and therefore the statement of \citeauthor{BW01} cannot possibly be a bounded formula.

Finally, \citeauthor{Garlik19} \cite{Garlik19} has proven lower bounds on the $\Reff$ formulas for the so-called non-relativized encoding. Unfortunately for us, his proofs encounter the same barrier: they rely on random restriction arguments, which are in any case more involved than the original ones.

We resolve these issues by coming up with new proof, inspired by the extraction algorithm, that modifies both ingredients in the original proof, yielding \Cref{thm:AM-in-PV-informal}. The random restriction argument is replaced by a greedy deterministic restriction just like the one used in the extraction algorithm. For the block-width lower bound, we show that the argument can be completely described by a Prover-Delayer game without referring to the exponential-size canonical tree-like refutation, making the entire proof formalizable in $\PVO$.

We remark that moving from the random restriction to the deterministic one comes at the cost of a slightly worse lower bound. The original size bound on $\Reff$ formulas is of the form $2^{\Omega(s/n)}$ and hence yields $2^{\Omega(n)}$ Resolution size lower bounds for all $\Reff_s$ formulas with $s \geq n^2$. Our deterministic restriction, in line with the parameters of the extraction algorithm, achieves a lower bound of $2^{\Omega(s/n^2)}$ which is exponential in $\Omega(n)$ for $s \geq n^3$. It seems reasonable that the original proof with the random restriction can be formalized in $\APC$, but we have not carried out this formalization.

\paragraph{Pudlák's upper bound in Resolution.}
For the upper bound in Resolution (\Cref{thm:ub-in-res-informal}), our proof is based on a careful analysis of the construction that makes it possible to describe the construction by a low-depth circuit. Carrying out the proof of the correctness of this circuit in Resolution is tedious, but ultimately clear once the right description of the circuit is provided.

\subsubsection{Consequences}

\paragraph{Characterization of the propositional fragments of Atserias--Müller.}
Our characterization theorem (\Cref{thm:characterization-informal}) is a consequence of the formalization of the lower bound in $\PVO$ (\Cref{thm:AM-in-PV-informal}) and the upper bound in Resolution itself (\Cref{thm:ub-in-res-informal}). Those results, in the propositional setting, imply that 
\begin{enumerate} \itemsep=0pt
    \item for the extraction algorithm $E$, we have $\EF \vdash \Reff(\Reff(\varphi), \pi) \to \Satf(\varphi, E(\varphi, \pi))$; and
    \item for Pudlák's algorithm $P$, we have $\Res \vdash \Satf(\varphi, \alpha) \to \Reff(\Reff(\varphi), P(\varphi, \alpha))$.
\end{enumerate}

Here, we highlight the variables of the $\Reff$ formula encoding a refutation $\pi$ as the second argument of $\Reff$. In particular, the Resolution proof of the correctness of $P$ is also possible in $\EF$. Then, simple use of contraposition allows us to go from $\neg\varphi$ to $\neg \Reff(\Reff(\varphi))$, and vice versa. That is, if $\EF$ can prove $\neg \varphi$, then it can also prove it in the encoding $\neg \Satf(\varphi, \alpha)$, and when substituting $E(\varphi, \pi)$ for $\alpha$, where $\pi$ are the only free variables, contraposition on item (1) gives us that $\EF$ derives $\neg \Reff(\Reff(\varphi), \pi)$. The other direction is analogous.

\paragraph{Automatability in terms of $\Reff$ formulas.}
For the characterization of automatability in \Cref{thm:aut-for-Ref-informal} to go through we build on \Cref{thm:characterization-informal} and additionally show that the characterization given there is not only in terms of proof size, but it is actually constructive. Given a proof of $\neg \varphi$ in $S$ we can efficiently construct a proof of $\neg \Reff(\Reff(\varphi))$, and vice versa. In this way, searching for proofs of $\Reff(\Reff(\varphi))$ is a proxy for the proofs of $\varphi$.

Remarkably, our proof techniques fail for proof systems strictly between Extended Frege and Resolution. The upper bound in Resolution does imply that from a refutation of $\Reff(\Reff(\varphi))$ we can obtain a refutation of $\varphi$. Unfortunately, it is our extraction algorithm (\Cref{thm:algo-informal}) what guaranteed that if $\varphi$ has a refutation of size $t$, then $\Reff(\Reff(\varphi))$ has a refutation of size $\poly(t)$. In Extended Frege this is true thanks to the extraction algorithm, but it seems conceivable that weaker systems might be able to easily prove $\neg \varphi$ without being able to prove $\neg \Reff(\Reff(\varphi))$ efficiently. (For Resolution itself this result does go through, for the more ad-hoc reason that $\Reff$ formulas talk about Resolution itself).

\paragraph{Unprovability of Resolution lower bounds.}
For \Cref{thm:prop-unprovable-informal} we exploit \Cref{thm:ub-in-res-informal}: if there is a short refutation of $\Reff(\Reff(\varphi))$, then there is a short refutation of $\varphi$. Since we formalized the upper bound construction in Resolution, the result applies to any proof system that contains Resolution (and behaves naturally in the sense that it is closed under \emph{modus ponens}). Then, if $\varphi$ is a hard formula for $Q$ and $Q$ simulates Resolution, we have that $\Reff(\varphi)$ is unsatisfiable. By the lower bound on $\Reff$ formulas (\Cref{thm:AM-original}), this formula is exponentially hard for Resolution, making $\Reff(\Reff(\varphi))$ unsatisfiable as well---but hard to refute for $Q$.

In the first-order setting, \Cref{thm:fo-unprovable-informal} relies again on the formalization of the upper bound on $\Reff$ formulas. This time, instead of starting from a sequence of hard propositional formulas, we can leverage Gödel's second incompleteness theorem to start from a sentence (the consistency of $T$) that is unconditionally unprovable in $T$. From this follows that $T$ cannot prove the soundness of a certain propositional system based on $T$ (the so-called \emph{strong proof system of $T$} \cite{Pudlak20}). We then consider the $\Reff(\cdot)$ formula around these soundness statements. We conclude that if $T$ could derive the lower bound on the $\Reff(\cdot)$ formulas in question, it would also be able to prove the soundness of the strong proof system of $T$ and, as a consequence, $T$ would derive its own consistency. Since Gödel's incompleteness theorem gives us sentences that are unconditionally independent of $T$, the corresponding Resolution lower bounds are also unconditionally unprovable in $T$. This works essentially for any theory of arithmetic subject to Gödel's incompleteness phenomenon, and does not rely on any complexity-theoretic assumptions.

\subsection{Related work}
Our work fits into a trend in complexity theory concerned with the meta-mathematics of computational complexity, which has gained remarkable momentum in recent years. Most of this work has been primarily concerned with the formalization of cornerstone results of computational complexity in bounded arithmetic and establishing unprovability and logical independence as barrier results. The literature is too vast to review here, so refer the reader to the recent survey of Oliveira \cite{Oliveira24}. In parallel, there has been a growing body of work deploying tools and ideas from mathematical logic to prove complexity-theoretic statements (see, e.g., \cite{azza1,MoritcForcing,khainc,khanw,PS22,Mykyta22,PS23, KraInc,azza2,Kha24,ACG24, AKPS24,KraENS,Mykyta24,azza3}). Our work continues in this direction.

Two recent works conceptually related to our investigations on $\Reff$ formulas merit further discussion. \citeauthor{ST21} \cite{ST21} initiated a general study of $\Reff$ formulas for arbitrarily strong proof systems. In particular, they studied \emph{iterations} of these formulas, which are reminiscent of the nested $\Reff(\Reff(\varphi))$ formulas that feature in our work. Their $\Reff$ formulas are not limited to Resolution, and they consider the iterated version of $\Reff^Q$ when $\Reff^Q$ talks about an arbitrarily strong proof system $Q$. While we are unable to connect our work on analyzability to their results, our characterization of proof size in terms of $\Reff$ formulas (\Cref{thm:characterization-informal}) has conceptual ties to their Iterated Lower Bounds Hypothesis.

The other relevant work is the research of \citeauthor{LLR24} \cite{LLR24}, who studied the provability of Resolution lower bounds in relativized theories of bounded arithmetic in the context of $\TFNP$. Until their work, the only formalization of proof complexity lower bounds that we are aware of is that of \citeauthor{CP90} \cite{CP90}. \citeauthor{LLR24} studied so-called \emph{refuter problems} in proof complexity: given a purported Resolution refutation of, say, $\PHP_n$, which is smaller than the known lower bounds, find a mistake in the proof (which must certainly exist, due to these very lower bounds). They connect the provability of lower bounds to the complexity of solving these refuter problems in subclasses of $\TFNP$. While their results yield formalizations of some proof complexity lower bounds, our results are essentially incomparable. First, their provability results are for \emph{relativized} theories of bounded arithmetic, where the given Resolution refutation is accessed through an oracle, while our proofs are in the non-relativized theories, where we can quantify over the objects in question. Second, they consider the provability of lower bounds for explicit families of tautologies like the pigeonhole principle or the Tseitin formulas. In contrast, the lower bound we are concerned is a sort of \emph{meta lower bound}: it tells us that the $\Reff(\varphi)$ formulas are hard whenever $\varphi$ is unsatisfiable. We believe, however, that the $\TFNP$ perspective on analyzability might shed light on some of our open questions.

\subsection{Open problems}

\paragraph{Analyzability of constant-depth Frege and other weak proof systems.} Similar techniques to those of~\cite{AM20} have been employed to prove the $\NP$-hardness of automating other weak proof systems like Regular and Ordered Resolution \cite{B20, BY24}, $k$-DNF Resolution \cite{G20}, Cutting Planes \cite{GKSMP20}, Nullstellensatz and Polynomial Calculus \cite{dRGNPRS21}, the OBDD proof system \cite{IR22} and, more recently, even $\ComplexityFont{AC}^0$-Frege \cite{P23}. All proof systems weaker than Resolution are analyzable just because Resolution is (i.e, their corresponding $\PAP$ problems are in $\P$), since analyzability is downwards closed under p-simulations. For the stronger systems, the question remains open. Are these systems analyzable? What about their search versions?

We highlight the analyzability of constant-depth Frege as a particularly interesting problem. While we have proven some unconditional lower bounds on $\Reff$ formulas here, it is open whether $\ComplexityFont{AC}^0$-Frege can prove any true Resolution lower bounds at all. It has been conjectured in the past that the PHP lower bound might be formalizable in these systems, at least in quasi-polynomial size. If this was possible, the $\NP$-hardness of $\PAP_\EF$ in \Cref{thm:np-hardness-informal} could be improved all the way to these systems.

\paragraph{$\FP$-completeness of the search version of $\PAP_\Res$.}
While we have shown that assignment extraction can be performed in polynomial time, our algorithm does not seem to be possible anywhere below $\P$. The algorithm seems hard to parallelize, which raises the question of whether the search version of $\PAP_\Res$ is in $\NC$ or even below. This is related to the question of whether the formalization of the lower bound on $\Reff$ formulas is provable in theories weaker than $\PVO$. If the statement was provable in, say, $\mathsf{VNC}^1$, witnessing theorems would give us an extraction algorithm in $\ComplexityFont{FNC}^1$. We conjecture that this improvement is impossible, and that the search problem of $\PAP_\Res$ is complete for $\FP$, but we are unable to prove it. The reduction, if true, likely requires some new technical idea. This would imply, among other things, that $\VZ$ does not prove the lower bound on $\Reff$ formulas, unconditionally.

\paragraph{On the weak automatability of Resolution.} Recall that 
a proof system is \emph{weakly automatable} if there exists a proof system that p-simulates it and is automatable. By our \Cref{thm:aut-for-Ref-informal}, the weak automatability of Resolution is equivalent to a proof system $Q$ simulating Resolution and being automatable on $\Reff$ formulas. If $Q \geq \EF$, then our theorem would imply that $Q$ itself would be automatable on all formulas, hitting cryptographic hardness results \cite{KP98, BPR00, BDGMP04, ACG24}. However, if $Q$ is strictly weaker than $\EF$, our statement does not apply and the automatability of $Q$ on $\Reff$ formulas does not imply automatability on all formulas. This does not seem to contradict any hardness assumptions. Of course, no such $Q$ is known to be efficiently automatable on Resolution lower bound statements, but this raises again the question of whether some non-trivial algorithm weakly automating Resolution might be plausible.

\subsection{Structure of the paper}
The paper is structured as follows. After the preliminaries in \Cref{sec:preliminaries}, we dedicate \Cref{sec:PAP-def} to formally defining the Proof Analysis Problem and stating some basic facts about it. \Cref{sec:the-algorithm} proves \Cref{thm:algo-informal}, describing the algorithm for the search version of $\PAP_\Res$ and \Cref{thm:krajicek-informal}. \Cref{sec:PAP-EF-NP-complete} proves \Cref{thm:np-hardness-informal}, giving $\NP$-hardness of $\PAP_\EF$ and stronger systems. In \Cref{sec:AM20-in-PV} and \Cref{sec:Pudlak-UB} we formalize, respectively, the Resolution lower bound and upper bound on $\Reff$ formulas that yield \Cref{thm:AM-in-PV-informal} and \Cref{thm:ub-in-res-informal}. \Cref{sec:AM-in-EF} translates this to the propositional setting to show that Extended Frege has polynomial-size proofs of the $\NP$-hardness of automating Resolution. Finally, \Cref{sec:Ref-formulas-applications} gives the proof of the characterization in \Cref{thm:characterization-informal} and proves \Cref{thm:aut-for-Ref-informal} on the equivalence of automatability of $\Reff$ formulas, while \Cref{subsec:unprovability} proves \Cref{thm:prop-unprovable-informal} and \Cref{thm:fo-unprovable-informal} on the unprovability of Resolution lower bounds. 

\section{Preliminaries}
\label{sec:preliminaries}
We assume the reader to be familiar with the central concepts of computational complexity theory. Below, we review the essential definitions and facts involving proof complexity and bounded arithmetic that feature in the paper. For a more comprehensive treatment of proof complexity, we refer to Krajíček \cite{Krajicek19}. For bounded arithmetic, the recent survey of \citeauthor{Oliveira24} \cite{Oliveira24} covers all the necessary material in the style of the meta-mathematics of computational complexity, which aligns with the style of our work. Other classical texts in logic and bounded arithmetic also cover these  contents (see, e.g., \cite{HP93, Krajicek95, Buss97, Buss98}).

\subsection{Levin reductions}
\label{subsec:prelim-Levin-reductions}

\newcommand{\Karp}{\leq^{\operatorname{p}}_{\operatorname{m}}}
\newcommand{\Levin}{\leq^{\operatorname{p}}_{\operatorname{Levin}}}

For decision problems $A, B\subseteq \{0, 1\}^*$, we use the notation $A \Karp B$ to express that $A$ \emph{many-one reduces} (or \emph{Karp reduces}) to $B$, meaning that there is a polynomial-time computable function $f$ such that for all $x \in \{ 0,1\}^*$ we have $x\in A$ if and only if $f(x) \in B$. We will be concerned with a strengthening of Karp reductions for \emph{search problems}. A search problem for us is a relation $R \subseteq \{ 0,1\}^* \times \{ 0,1\}^*$. A search problem $R$ is in $\FP$ if there exists a polynomial-time function that on input $x$, outputs some $y$ such that $(x, y) \in R$, if such a $y$ exists.

We will say that $R_1$ \emph{Levin reduces} to $R_2$ and write $R_1 \Levin R_2$ if there exists a triple of polynomial-time computable functions $(f, g, h)$ such that for all $x, w\in \{ 0,1\}^*$, it holds that (i) if $(x, w) \in R_1$, then $(f(x), g(x, w)) \in R_2$, and (ii) if $(f(x), w) \in R_2$, then $(x, h(x, w)) \in R_1$. Whenever membership in $R_1$ and $R_2$ is checkable in polynomial time, it holds that the sets $\dom R_1$  and $\dom R_2$ are in $\NP$ and $\dom R_1 \Karp \dom R_2$.

\subsection{Proof complexity}
\label{subsec:proof-complexity}
Following the classical definition of \citeauthor{CR79} \cite{CR79}, a \emph{propositional proof system} $S$ for the set~$\TAUT$ of propositional tautologies is a polynomial-time function $S : \{ 0,1\}^* \to \TAUT$ whose range is exactly $\TAUT$. We think of $S$ at the polynomial-time verifier mapping proofs to the statements they prove; i.e., if $S(\pi) = \varphi$, then we say $\pi$ is an $S$-proof of $\varphi$. It is often convenient to think of a proof system as establishing \emph{unsatisfiability}; thus, if $\varphi$ is an unsatisfiable formula and $\pi$ is an~$S$-proof of the tautology $\neg\varphi$, then we say that $\pi$ is an \emph{$S$-refutation} of~$\varphi$, or an~$S$-proof of the unsatisfiability of~$\varphi$. Since we deal exclusively with classical logic, here and below we tacitly gloss over the distinction between the formulas $\neg\neg\varphi$ and $\varphi$; this is particularly useful for literals $\ell$, where $\neg \ell$ is sometimes used to denote the complementary literal.

For a tautology $\varphi$ and a proof system $S$, we denote by $\size{S}{\varphi} \coloneq \min_{\pi : S(\pi) = \varphi} |\pi| $ the \emph{size} of its smallest $S$-proof. A proof system $S$ is \emph{polynomially bounded} if there exists a constant~$c \in \bbN$ such that for all $\varphi \in \TAUT$ we have $\size{S}{\varphi} \leq |\varphi|^c$. 
\newcommand{\effproves}{\vdash_{\operatorname{poly}}}
For a sequence $\varphi = \{\varphi_n\}_{n\in\bbN}$ of tautologies, we write $S \effproves \varphi$ or simply $S \effproves \varphi_n$ to express that $\size{S}{\varphi_n} = |\varphi|^{O(1)}$ as $n$ grows.
When we want to emphasize that it is via a specific proof $\pi$ that $S$ proves $\varphi_n$, we write $\pi : S\vdash\varphi_n$. More generally, for $s \in \bbN$, we write $S\vdash_{s} \varphi_n$ to express that there exists an~$S$-proof $\pi$ of size $|\pi| \leq s$ such that~$\pi : S \vdash \varphi_n$.

We say that a proof system $S$ \emph{simulates} another system $Q$, written $S \geq Q$, if there exists a constant~$c \in \bbN$ such that for every $\varphi \in \TAUT$ we have $\size{S}{\varphi} \leq \size{Q}{\varphi}^c$. We additionally say that $S$ \emph{p-simulates} $Q$ and write $S \psim Q$ if there exists a polynomial-time computable function sending $Q$-proofs to $S$-proofs of the same formula; i.e., there exists a polynomial-time computable function $f$ such that for every $\varphi\in \TAUT$ and every $\pi : Q\vdash \varphi$, we have $f(\pi) : S \vdash \varphi$. We say that two proof systems $S$ and $Q$ are \emph{polynomially equivalent} if $S \psim Q$ and $Q \psim S$. A proof system $S$ is \emph{optimal} if $S \geq Q$ for every propositional proof system $Q$, and respectively \emph{p-optimal} if $S \psim Q$ for every propositional proof system $Q$.

A \emph{literal} is a propositional variable or its negation. Given a formula $\varphi(x_1, \dots, x_n)$, a \emph{literal substitution} is a mapping of the form $\rho : \{ x_1, \dots, x_n\} \to \{x_1, \dots, x_n, \neg x_1, \dots, \neg x_n, 0, 1\}$ that replaces variables by other literals or substitutes constants in their place. We denote by $\varphi_{\restriction \rho}$ the substituted formula $\varphi(\rho(x_1), \dots, \rho(x_n))$, with the convention that every resulting occurrence of $\neg\neg x_i$ is replaced by $x_i$. A \emph{restriction} is a particular case of a variable substitution, where all variables are mapped to either $0$, $1$, or themselves. We say that a proof system~$S$ is \emph{closed under substitutions} (respectively, \emph{closed under restrictions}) if there exists a constant $d\in \bbN$ such that for every tautology $\varphi$ and every literal substitution (respectively, restriction) $\rho$, it holds that $\size{S}{\varphi_{\restriction \rho}} \leq \size{S}{\varphi}^d$. All the explicit proof systems dealt with in this work (i.e., the ones described below, like Resolution or Frege or Extended Frege systems) are closed under literal substitutions. In these cases, a proof of the substituted formula can be obtained directly by applying the substitution line by line to every formula appearing in a proof $\pi$ of $\varphi$, and we hence denote by $\pi_{\restriction\rho}$ the corresponding substituted proof of~$\varphi_{\restriction \rho}$. 

\subsubsection{Resolution}
\label{subsec:prelim-resolution}
A central proof system in this work is Resolution ($\Res$). We usually see this as a refutation system for CNF formulas. Accordingly, we sometimes write $\pi : \Res\ \vdash \neg\varphi$ for a CNF formula $\varphi$, to mean that $\pi$ is a Resolution refutation of $\varphi$, hence a proof of the tautology $\neg\varphi$. In this way, $\Res$ is a Cook--Reckhow proof system for the fragment of $\TAUT$ made of the formula of the form $\neg\varphi$, where $\varphi$ is an unsatisfiable CNF formula. Through the standard Tseitin transformation of an arbitrary propositional formula into equisatisfiable CNF form, $\Res$ can also be seen as a Cook--Reckhow proof system for $\TAUT$ itself; we do not need the details of this in this paper.

A \emph{literal} is a propositional atom or its negation, a \emph{clause} is a disjunction of literals, and a \emph{CNF formula} is a conjunction of clauses. We see clauses as sets of literals, and write simply $C \subseteq D$ to express that $C$ is a subclause of $D$.  A \emph{Resolution refutation} of an unsatisfiable CNF formula $\varphi = C_1 \land \dots \land C_m$ over variables~$x_1, \dots, x_n$ is a sequence $D_1, \dots, D_s$ of clauses over $x_1, \dots, x_n$ such that $D_s = \bot$, denoting the empty clause, and for every $i \in [s-1]$, the clause $D_i$ either (a) is one of the clauses $C_1, \dots, C_m$ of $\varphi$, or (b) is a \emph{weakening} of a previous clause, meaning that $D_i \supseteq D_j$ for some $1\leq j < i$, or (c) has been obtained from two previous clauses $D_j = A \lor x$ and $D_k = B \lor \neg x$, for~$j, k < i$, by an application of the \emph{Resolution rule}:
    \begin{prooftree}
        \AxiomC{$A \lor x$}
        \AxiomC{$B \lor \neg x$}
        \RightLabel{\footnotesize{\;\;(Res)}}
        \BinaryInfC{$A \lor B$}
    \end{prooftree}
We say that $A \lor B$ is obtained by \emph{resolving} over $x$.
The \emph{length} of $\pi$, denoted by $\prooflength{\pi}$, is $s$. We often write $\varphi \vdash_\Res^s \bot$ to indicate that there exists a Resolution refutation of $\varphi$ in length $s$.

To every Resolution refutation $\pi$ we can associate a directed acyclic graph in a natural way, and we often do so implicitly. We denote by $\depth(\pi)$ the length of the longest path in the dag, starting from the root labeled by the empty clause $\bot$. The number of vertices in this graph is precisely $\prooflength{\pi}$.

Resolution is \emph{implicationally complete}, meaning that for every pair of CNF formulas $\varphi$ and $\psi$, if $\varphi \models \psi$, there exists a derivation of each clause of $\psi$ from the clauses of $\varphi$. We write $\varphi \vdash_\Res \psi$ to express that such a collection of derivations exist.

We will also deal with a mild extension of the Resolution system, known as \emph{$k$-DNF Resolution} \cite{Krajicek01}, denoted $\Res(k)$ for $k \geq 1$. The system $\Res(k)$ is also a refutational system, but lines are \emph{$k$-DNF} formulas, which are unbounded fan-in disjunctions of \emph{$k$-terms}, conjunctions of up to $k$ literals. A clause is a $1$-DNF. The system consists of a weakening and an introduction rule,
\[
\hbox{
  \AxiomC{$A$}
   \RightLabel{\footnotesize{\;\;(Weak)}}
  \UnaryInfC{$A \lor B$}
  \DisplayProof
  \qquad
  \AxiomC{$A \lor \ell_1$}
  \AxiomC{$B \lor (\ell_2 \land \dots \land \ell_s)$}
   \RightLabel{\footnotesize{\;\;($\land$-Intro)}}
  \BinaryInfC{$A \lor B \lor (\ell_1 \land \dots \land \ell_s)$}
  \DisplayProof
}
\]
together with a Cut rule that generalizes the Resolution rule,
\begin{prooftree}
    \AxiomC{$A \lor (\ell_1 \land \dots \land \ell_s)$}
  \AxiomC{$B \lor \neg \ell_1 \lor \dots \lor \neg \ell_s$}
   \RightLabel{\footnotesize{\;\;(Cut)}}
  \BinaryInfC{$A \lor B$}
\end{prooftree}
where $A$ and $B$ are $k$-DNF formulas and $s \leq k$.

It is easy to see that Resolution ($\Res$) corresponds to $\Res(1)$.

\subsubsection{Frege systems}
\label{subsec:prelim-Frege}
Throughout this work we reason about Resolution refutations within much stronger systems for propositional logic. A \emph{Frege system} \cite{CR79} consists of a finite set of axiom schemas and inference rules that are sound and implicationally complete for the language of propositional tautologies built from the Boolean connectives negation ($\neg$), conjunction ($\land$), and disjunction ($\lor$). A \emph{Frege proof} is then a sequence of formulas where each formula is obtained by either substitution of an axiom schema or by application of an inference rule on previously derived formulas. The specific choice of rules does not affect proof size up to polynomial factors, as long as there are only finitely many rules and these are sound and implicationally complete \cite{CR79}. We refer to Cook and Reckhow \cite{CR79} or Krajíček \cite[§2.1]{Krajicek19} for specific examples of choices for these rules and axioms. One can alternatively define Frege systems in the formalism of Natural Deduction or the Sequent Calculus for classical propositional logic, but we will not be concerned with these syntactic details.

Of central importance for us is the \emph{Extended Frege} ($\EF$) system \cite{CR79}, in which proof lines can be succinctly written as Boolean circuits rather than formulas \cite{jerabek:phd-thesis}. In general, for a circuit class $\ComplexityFont{C}$, one can consider the proof system $\ComplexityFont{C}$-Frege, in which lines are restricted to be Boolean circuits of that type. We are particularly interested in the $\ACZF{d}$ systems, in which lines are restricted to be Boolean circuits of unbounded fan-in and constant depth $d$. We also consider more generally \emph{bounded-depth Frege systems}, where the depth $d$ is bounded, but not necessarily a constant.

For bounded-depth Frege systems, we have strong lower bounds available. The most famous such lower bound is the one for the \emph{Pigeonhole Principle} ($\PHP$). For every $m \in \bbN$ and $n\in \bbN$ such that $m > n$, the formula $\PHP^m_n$ stands for the CNF formula over variables $p_{i,j}$ for $i \in [m]$ and $j \in [n]$ consisting of the clauses
{\allowdisplaybreaks
\begin{axiomdef}{2}{PHP}
        & \bigvee_{j\in[n]} p_{i,j} &&\quad\text{ for all }i\in[m], \label{php-1} \\
        &\neg p_{i,j} \lor \neg p_{i',j} &&\quad\text{ for all }i, i' \in [m], i \neq i' \text{ and } j\in [n] \label{php-2}
    \end{axiomdef}}
We sometimes denote by $\PHP_n$ the formula $\PHP^{n+1}_n$.

Strong lower bounds are known on the proof complexity of the pigeonhole principle for bounded-depth Frege systems \cite{Ajtai94, PBI93, KPW95, UF96}. Here we state only a simplified version of the best such lower bound, proven by Håstad \cite{Hastad23}.

\begin{theorem}[\cite{Hastad23}]
    \label{thm:Hastad-PHP}
    For every $d \leq O(\log n / \log \log n)$, depth-$d$ Frege systems require size at least $\exp({\Omega(n^{1/(4d-2)})})$ to prove $\neg \PHP_n$.    
\end{theorem}

Finally, we often consider extensions of Extended Frege by sets of additional axioms. For a set $A \subseteq \TAUT$ of tautologies that is recognizable in polynomial time, the system $\EF + A$ refers to Extended Frege extended with the axiom schemas that allow (formula) substitution instances of any formula in $A$.

\subsubsection{Automatability and proof search}
\label{subsec:prelim-automatability}
The notion of of automatability, introduced by \citeauthor{BPR00} \cite{BPR00}, is a natural formalization of efficient proof search in propositional proof system. We say that a proof system $S$ is \emph{automatable} if there exists a constant $c\in\bbN$ and an algorithm that given a propositional tautology $\varphi$, outputs an $S$-proof of $\varphi$ in time $(|\varphi| + \size{S}{\varphi})^c$, meaning that the proof search algorithm succeeds in finding a proof of size polynomial in the size of the shortest one.

Even when a system might not be automatable, it seems natural to ask whether there exists a system $Q$ that p-simulates $S$ and is itself automatable. In this case, we say that $S$ is \emph{weakly automatable} \cite{AB04}. Weak automatability is equivalent to the existence of an automating algorithm where the output proof belongs to a system $Q \geq_{\operatorname{p}}S$ rather than $S$ itself. In particular, weak automatability is closed downwards under p-simulation.

A more restrictive notion of proof search is given by the \emph{Proof Size Problem}. Associated to any propositional proof system $S$ we can define the \emph{Proof Size Problem for $S$} ($\PSP_S$), defined as the language
\[ \PSP_S \coloneq \{ (\varphi, 1^s) \mid \text{there is an $S$-proof of $\varphi$ in size $s$} \} .\]
Automating $S$ entails approximating minimum proof-size to a polynomial, in polynomial time.

\subsection{Bounded arithmetic}
\label{subsec:prelim-BA}

We heavily rely on the connection between propositional proof complexity and (weak) theories of arithmetic. We assume familiarity with basic knowledge of first-order logic and introduce the main theories we are concerned with. 

\subsubsection{The theories $\PVO$ and $\SOT$}
\label{subsec:prelim-PVO-SOT}
Theories of bounded arithmetic capture various forms of feasible reasoning and act as a uniform counterpart of propositional proof systems. The main tool to capture feasibility in mathematical reasoning is to bound the complexity of formulas over which one can apply induction.

\paragraph{Cook's $\PVO$.} Cook's theory $\PVO$ \cite{Cook75, KPT91} is an attempt to make precise the idea of polynomial-time reasoning. It is a universal theory whose vocabulary $\calL_\PV$ consists of a function symbol for each polynomial-time function, and the axioms are precisely the recursive definitions of these functions via composition and limited recursion on notation, in the style of Cobham's functional definition of $\FP$ \cite{Cobham65}. The theory further admits induction on quantifier-free formulas, which define precisely polynomial-time predicates.

The formal definition of $\PVO$ is rather technical and the details are not particularly relevant to our proofs, so we refer the reader to Krajíček's textbook \cite[Definition 5.3]{Krajicek95} for the details. The reason we rarely care about the technicalities of $\PVO$ is that we often work instead in the theory $\SOT$ of Buss, which happens to be conservative over $\PVO$ for the classes of formulas we are interested in. We discuss this next.

\paragraph{Buss's $\SOT$.}
We see $\SOT$ as a theory sitting in between Robinson's Arithmetic $\Q$ and Peano Arithmetic $\PA$. Let $\calL_\PA$ denote the language of {Peano Arithmetic}, $\calL_\PA \coloneq \{ 0, 1, +, \cdot, <\} $. The axioms of $\PA$ consist first of the axioms of Robinson's arithmetic $\mathsf{Q}$, which define the basic behavior of the symbols of $\calL_{\PA}$ (see, for example, \cite[§7.4.3]{Krajicek19} for a definition), together the \emph{Induction Schema}
\begin{equation}
    (\varphi(0) \land \forall x(\varphi(x) \to \varphi(x+1)) \to \forall x \varphi(x),
    \tag{$\textsc{Ind}_\varphi$}
\end{equation}
available for every formula $\varphi$.

The language of $\SOT$ is the first-order language of bounded arithmetic, $\mathcal{L}_{\text{BA}} \coloneqq \{ 0, 1, +, \cdot, <, |\cdot|, \lfloor \cdot/2 \rfloor, \#\}$. This extends  the language of Peano Arithmetic $\calL_\PA$ above by the symbols $|x|$, $\lfloor x/2 \rfloor$ and $x \# y$. The standard interpretation of $\lfloor x/2 \rfloor$ is clear. The notation $|x|$ denotes the length of the binary encoding of the number $x$, $\lceil \log (x + 1) \rceil$, while the \emph{smash symbol} $x\#y$ stands for $2^{|x|\cdot|y|}$.

For a term $t$ in the language of bounded arithmetic and a variable $x$ that does not appear in $t$, a formula of the form  $\forall x (x < t \to \varphi(x))$ or $\exists x (x < t \land \varphi(x))$ is called a \emph{bounded formula}. The quantifiers guarded by the bounds on $x$ are called \emph{bounded quantifiers} and we simply write $\forall x < t(\varphi(x))$ and $\exists x < t (\varphi(x))$. If the bounded quantifiers are of the form $\forall x < |s|$ of $\exists x < |s|$ for some term $s$, then they are called \emph{sharply bounded} quantifiers. The \emph{hierarchy of bounded formulas} consists of the classes $\Sigma_n^b$ (and $\Pi_n^b$), for $n \geq 1$, which are defined by counting the alternations of bounded quantifiers ignoring the sharply bounded ones, starting with an existential (respectively, universal) one. The class $\Delta_n^b$ consists of all formulas that admit an equivalent definition in both $\Sigma_n^b$ and $\Pi_n^b$. In particular, the class $\Delta_0^b$ stands for all formulas with sharply bounded quantifiers only.
 
The theory $\SOT$ of Buss \cite{Buss86} extends Robinson's arithmetic $\mathsf{Q}$ by a set $\operatorname{BASIC}$ of simple axioms for the new function symbols (see, e.g., \cite[Definition 5.2.1]{Krajicek95} for the complete list). On top of this, the theory has the \emph{Polynomial Induction} schema (\textsc{PInd}) for $\Sigma_1^b$-formulas: for every $\varphi \in \Sigma_1^b$, the theory contains the axiom
\begin{equation}
    \varphi(0) \land \forall x (\varphi (\lfloor x/2 \rfloor) \to \varphi(x)) \to \forall x \varphi(x).
    \tag{$\textsc{PInd}_\varphi$}
\end{equation}

When working over $\SOT$, we often invoke instead the schema for \emph{Length Induction},
\begin{equation}
    \varphi(0) \land \forall x (\varphi (x) \to \varphi(x+1)) \to \forall x \varphi(|x|),
    \tag{$\textsc{LInd}_\varphi$}
\end{equation}
made available for all $\Sigma^b_1$-formulas. This form of induction is provable from ($\textsc{PInd}_\varphi$) for $\varphi \in \Sigma^b_1$ \cite[Lemma 5.2.5]{Krajicek95}.

Unlike $\calL_{\PV}$, the language $\calL_{\text{BA}}$ of bounded arithmetic does not contain a function symbol for every function in $\FP$. However, every such $f \in \FP$ is $\Sigma^b_1$-definable in $\SOT$, meaning that there exists a $\Sigma^b_1$ formula whose interpretation over the standard model $\bbN$ defines $f$ and such that $\SOT$ proves the totality of this definition. Thus, in the rest of the paper we choose to use the theory $\SOT(\calL_{\PV})$, which is the theory $\SOT$ in the language of bounded arithmetic extended by all $\PV$ function symbols, meaning that we have a fresh symbol for each function in $\FP$, and induction is now available for all $\Sigma_1^b(\PV)$ formulas. The theory $\SOT(\calL_{\PV})$ is fully conservative over $\SOT$. In what follows we abuse notation and denote this simply as $\SOT$. 

The final key fact for us is that all $\forall\Sigma^b_1$ formulas provable in $\SOT$ are already provable in $\PVO$. That is, the theory $\SOT$ is $\forall\Sigma^b_1$-conservative over $\PVO$ \cite{Buss86}. We use this in some of the formalizations, where we carry out arguments in $\SOT$ but later appeal to its proof in $\PVO$. For a proof of this fact, see, for example Krajíček's textbook \cite[Thm. 5.3.4 and Cor. 7.2.4]{Krajicek95}.

\subsubsection{Exact counting in $\PVO$ and $\SOT$}
\label{subsec:prelim-exact-counting}
When working in theories of bounded arithmetic, counting the size of different sets requires a great deal of care. We use the wide-spread \emph{Log-notation} $n \in \Log$ as a short-hand for the formula $\exists x (|x|=n)$. A set $X$ is called a \emph{bounded definable set with parameter $y$} if there exists an arithmetic formula $\varphi(x,y)$ and some term $t(y)$ such that $X=\{ x<t(y) \mid \varphi(x,y) \}$. We also adopt the standard set-theoretic notation denoting the interval $[0, a)$ directly by $a$. Then, a Boolean circuit $C : 2^k \to 2$ naturally defines a bounded definable set $X_C = \{ x <2^k \mid C(x) = 1\}$ with parameter $C$ from which $k \in \Log$ can be extracted, and there exists a $\PV$ function $\operatorname{Count}$ which counts $\Log$-sized initial segments of circuit-definable sets; i.e., in the standard model, $\operatorname{Count}(C,a)$ is the cardinality of the set $X_C \cap [0,|a|)$.

Importantly, if we know that $2^k \in \Log$, then we are able to \emph{exactly count} the size of $X_C$ when working in $\PVO$. This proves crucial when carrying out different combinatorial arguments. As an example, arguments involving the pigeonhole principle or different averaging arguments are all possible in $\PVO$ thanks to exact counting, as long as the sizes of the sets in question are in $\Log$.

\subsubsection{Buss's witnessing theorem}
\label{subsec:prelim-witnessing}

Witnessing theorems are a central tool in the field of bounded arithmetic connecting proofs and computation. Roughly speaking, they show that existential quantifiers can be turned into explicit functions computing witnesses for these quantifiers, such that the computational complexity of such a function depends tightly on the logical strength of the theory proving the statement. We mainly rely on Buss's witnessing theorem for $\SOT$ (and hence also for $\PVO$) , capable of witnessing one level of existential quantifiers via polynomial-time functions.

\begin{theorem}[Buss's witnessing theorem \cite{Buss86}]
    \label{thm:buss-witnessing}
    Let $\varphi(x, y)$ be a $\Sigma^b_1$ formula. If $\SOT \vdash \forall x\exists y\varphi(x, y)$, then there exists a $\PV$ function $w$ such that $\SOT \vdash \forall x\varphi(x, w(x))$.
\end{theorem}

Buss's original argument uses proof theory. We refer to \citeauthor{HP93} \cite[Theorem 4.32]{HP93} for a model-theoretic proof.

\subsubsection{Cook's propositional 
translation}
\label{subsec:prelim-Cook-trans}

Following Krajíček \cite[§8.6]{Krajicek19}, we say that a theory of arithmetic $T$ \emph{corresponds} to a propositional proof system $S$ if (i) $T$ can prove the soundness of $S$ and (ii) every universal consequence $\forall x \varphi(x)$ of $T$, where $\varphi$ is quantifier-free, admits polynomial-size proofs in $S$ when suitably grounded into a sequence of propositional formulas. (Pudlák alternatively says that $S$ is a \emph{weak system} of the theory $T$ \cite{Pudlak20}.)

We are interested in the correspondence between $\PVO$ and Extended Frege ($\EF$). In this case, the process used to turn first-order formulas into propositional ones is known as (Cook's) \emph{propositional translation}, introduced in his seminal paper on $\PV$ \cite{Cook75}. Given a quantifier-free formula $\varphi(x)$, Cook's translation is a polynomial-time construction sending $\varphi$ to a sequence of polynomial-size propositional formulas $\{ \trans{\varphi}_n \}_{n \in \bbN}$ such that for every $n\in \bbN$, the formula $\trans{\varphi}_n \in \TAUT$ if and only if $\bbN \models \varphi(n)$. See \cite[§12.3]{Krajicek19} or \cite[§6.1]{Buss97} for a complete definition of the construction. Cook then observed that, under this translation, $\PVO$ and $\EF$ do indeed correspond to each other.

\begin{theorem}[Cook's correspondence theorem \cite{Cook75}]
\label{thm:translation}
    The theory $\PVO$ and the proof system Extended Frege correspond to each other. That is,
    \begin{enumerate}[label=(\roman*)] \itemsep=0pt
        \item $\PVO$ proves the soundness of $\EF$;
        \item if $\varphi(x)$ is a quantifier-free formula in the language $\calL_{\operatorname{BA}}(\PV)$ and $\PVO \vdash \forall x\varphi(x)$, then there exists a polynomial-time computable function $f$ such that for every $n \in \bbN$,  it holds that $f(1^n) : \EF \vdash \trans{\varphi}_n$.
    \end{enumerate}
\end{theorem}

A proof of the theorem can be found in Krajíček's textbook \cite[Theorem 12.4.2]{Krajicek19}.

\subsubsection{Strong proof systems of theories of arithmetic}
\label{subsec:prelim-strong-proof-system-of-T}
Let $T$ be a consistent theory extending Robinson Arithmetic $\Q$ by a set of axioms that is decidable in polynomial time (i.e., $T$ extends $\Q$ and is \emph{polynomial-time axiomatizable}). We denote by $P_T$ the propositional proof system in which a proof of the propositional tautology $\varphi$ is given by a $T$-proof of the first-order sentence $\foTautf(\ulcorner \varphi \urcorner)$ which states that the propositional formula $\varphi$ is made true by every truth-assignment for its variables; here $\ulcorner \varphi \urcorner$ stands for the code of $\varphi$ in a standard arithmetization of propositional formulas. Following Pudlák \cite{Pudlak20}, we call~$P_T$ the \emph{strong proof system of $T$}.

We will use two key facts of these proof systems. First, if $T$ is a theory for which $P_T$ is defined and $Q$ is a propositional proof system that corresponds to $T$ in the sense of \Cref{subsec:prelim-Cook-trans}, then $P_T \psim Q$ \cite[§4.2, Fact 2]{Pudlak20}. Second, by Gödel's second incompleteness theorem, $T$ does not prove the soundness of $P_T$ \cite[§4.2, Fact 3]{Pudlak20}.

\subsection{The $\Reff$ formulas}
\label{subsec:prelim-Reff}
The main character in this paper is the so-called \emph{$\Reff$ formula}. Given a propositional CNF formula $\varphi$ and a size parameter $s \in \bbN$, the formula $\Reff_s(\varphi)$ states that $\varphi$ has a Resolution refutation consisting of at most $s$ clauses.

It is important to choose an encoding that is simultaneously natural from a modeling point of view while not making the formulas artificially hard to refute. At a basic level, the main property that such an encoding should satisfy is that for a concrete $\varphi$ and $s$, the formula $\Reff_s(\varphi)$ should be satisfiable if and only if there is a Resolution refutation of $\varphi$ in at most $s$ clauses. It also seems natural to require that a Resolution refutation should be readable in polynomial time from a satisfying assignment to $\Reff$. 

While different encodings have appeared in the literature, they tend to agree on a few basic ideas. The formula $\Reff_s(\varphi)$ consists of $s$ so-called \emph{blocks} of variables, each representing a clause in the purported Resolution refutation. Each block has variables to represent the literals that appear in this block, how it was obtained (resolved or weakened from an axiom), and it contains \emph{pointer variables} to indicate from which blocks it was derived.

\paragraph{The unary encoding of Pudlák.}
Pudlák \cite{Pudlak03} uses the seemingly most standard encoding, which we refer to as the \emph{unary encoding} for $\Reff$. He used it to prove that the canonical pair of Resolution is symmetric. This encoding employs pointers in unary, meaning that for every block $B \in [s]$, there are up to $s$ additional variables to point at the blocks from which $B$ was derived.

\paragraph{The relativized unary encoding of Atserias and Müller.}
\citeauthor{AM20} \cite{AM20} start by studying Pudlák's encoding.
They proved suitable so-called \emph{index-width} lower bounds for it in Resolution, but they were unable to prove a \emph{size} lower bound for it. They then introduced a \emph{relativized} version, in which each block can be possibly \emph{enabled} or \emph{disabled}. If it is disabled, then the block is not used towards the refutation, and its associated clauses are immediately satisfied. These additional \emph{enabling variables} made it possible to prove the size lower bound from the index-width lower bound for the unrelativized encoding. We refer to this second encoding as the \emph{relativized unary encoding}.

We note, however, that the change of encoding is not the source for hardness. \citeauthor{Garlik19} \cite{Garlik19} proved that even when using the original encoding of Pudlák, the formulas are hard for Resolution whenever the underlying CNF is unsatisfiable.

\paragraph{The binary encoding of \citeauthor{dRGNPRS21}.}
In their alternative proof of the lower bound on $\Reff$ formulas, \citeauthor{dRGNPRS21} \cite{dRGNPRS21} introduce an encoding of $\Reff$ where pointers are encoded in binary. Informally, for every block $B \in [s]$, there are $O(\log s )$ variables used encode the value $B' \in [s]$ of the block(s) from which $B$ was derived. We refer to this as the \emph{binary encoding}. While this encoding also includes the enabling variables of the relativized encoding, these are inessential, since one can always assign the pointers in a dummy fashion to effectively disable a block.

We contend that the unary relativized encoding is both the most natural as well as the most versatile. We see three reasons for this:
\begin{enumerate}
\itemsep=0pt
    \item thanks to the enabling variables, one can naturally turn a Resolution refutation of $t < s$ clauses into a satisfying assignment to $\Reff_s(\varphi)$ simply by disabling $s-t$ blocks that are not needed, while in the relativized encoding one needs to fill in the remaining $s-t$ clauses with some dummy content;
    \item the enabling variables make the random restriction argument leading to the size lower bound much simpler to prove, and Garlík has shown that the hardness does not comes from this change in syntax;
    \item when using a unary encoding rather than the binary one of \citeauthor{dRGNPRS21}, one can easily restrict some pointers to get an instance of $\Reff_t(\varphi)$ for every $t < s$, while in the binary encoding, after disabling a block, the binary pointers might still be able to point to it, making the formulas more delicate to handle after applying a restriction.
\end{enumerate}

We remark that the choice between unary and binary encodings is ultimately inessential, and all the results in this paper can be reproven for the binary encoding. We choose the unary encoding mainly for reason (3) above, which simplifies the write-up.

We now define the formula in detail.

\paragraph{The variables of $\Reff_s(\varphi)$.}
\label{subsec:vars-of-reff}
Here, we assume $\varphi$ is a CNF formula over $n$ variables $x_1, \dots, x_n$ and $m$ clauses and define the following variables, where $\Lit_n \coloneq \{x_1, \dots, x_n, \neg x_1, \dots, \neg x_n\}$.

\bigskip
\begin{tabular}{lll}
    $\alit^A_\ell$ & : & literal $\ell \in \Lit_n$ is present in the clause $A \in [m]$ of $\varphi$; \\
    $\enable^B$ & : & block $B \in [s]$ is enabled; \\
    $\derived^B$ & : & block $B\in [s]$ is obtained by a Resolution step;\\ 
    $\weak^B_A$ & : & block $B\in[s]$ is obtained by weakening from clause $A \in [m]$ of $\varphi$; \\
    $\lit^B_\ell$ & : & literal $\ell\in\Lit_n$ is present in the block $B\in[s]$; \\
    $\res^B_{x_i}$ & : & block $B\in[s]$ is obtained by resolving over the variable $x_i$; \\
    $\lpoint^B_{B'}$ & : & block $B \in [s]$ is resolved on the left from block $B' \in [s]$, $B'<B$; \\
    $\rpoint^B_{B'}$ & : & block $B\in[s]$ is resolved on the right from block $B' \in [s]$, $B'<B$.
\end{tabular}
\bigskip

Building on these variables, the $\Reff_s(\varphi)$ formula is defined as follows. We write the clauses as implications for the sake of readability.

\begin{definition}[The $\Reff$ formulas]
\label{def:Reff-formulas}
Let $n, m, s \in \bbN$, and let $\Lit_n \coloneq \{x_1, \dots, x_n, \neg x_1, \dots, \neg x_n\}$ denote the of possible literals over $n$ variables. The $\Reff_s(\varphi)$ formula is built from the variables defined above, together with the conjunction of the following clauses:

{\allowdisplaybreaks
    \begin{axiomdef}{2}{Ref}
        & \left( \enable^B \land \res^B_{x_i} \land \lpoint^B_{B'} \land \lit^{B'}_\ell \right) \to \lit^{B}_\ell  &&\quad\text{ for }B, B'\in[s], B' < B, i\in[n], \ell \in \Lit_n \setminus \{ x_i\},\label{axiom:must-appear-after-res-left} \\
        & \left( \enable^B \land \res^B_{x_i} \land \rpoint^B_{B'} \land \lit^{B'}_\ell \right) \to \lit^{B}_\ell  &&\quad\text{ for }B, B'\in[s], B' < B, i\in[n], \ell \in \Lit_n \setminus \{\neg x_i\},\label{axiom:must-appear-after-res-right}\\
        & \left( \enable^B \land \weak^B_A \land \alit^A_\ell\right) \to \lit^B_\ell  &&\quad\text{ for }B\in[s], A\in[m], \ell\in \Lit_n,\label{axiom:must-appear-after-weak}\\
        & \left( \enable^B \land \derived^B \right ) \to \bigvee_{i\in[n]} \res^B_{x_i}  &&\quad\text{ for }B \in [s],\label{axiom:must-resolve} \\
        & \left( \enable^B \land \derived^B \right ) \to \bigvee_{\substack{B' \in [s] \\ B' < B}} \lpoint^B_{B'}  &&\quad\text{ for 
        }B \in [s],\label{axiom:must-point-left} \\
        & \left( \enable^B \land \derived^B \right ) \to \bigvee_{\substack{B' \in [s] \\ B' < B}} \rpoint^B_{B'}  &&\quad\text{ for }B \in [s],\label{axiom:must-point-right} \\
        & \left( \enable^B \land  \neg \derived^B \right) \to \bigvee_{A\in[m]} \weak^B_A  &&\quad\text{ for }B \in [s],\label{axiom:must-weaken} \\
        & \left( \enable^B \land  \lpoint^B_{B'} \right) \to \enable^{B'}  &&\quad\text{ for }B, B' \in [s], B' < B\label{axiom:must-enable-left} \\
        & \left( \enable^B \land  \rpoint^B_{B'} \right) \to \enable^{B'}  &&\quad\text{ for }B, B' \in [s], B' < B\label{axiom:must-enable-right} \\
        & \neg \lit^s_\ell &&\quad\text{ for }\ell \in \Lit_n ,\label{axiom:root-empty}\\
        & \enable^s. && \label{axiom:root-enabled}
    \end{axiomdef}
}\end{definition}

\begin{remark}
    Our encoding of $\Reff$ has fewer axioms than that of~\cite{AM20}. For example, we do not require that if a block $B$ is resolved on the left by variable $x$ from block $B'$, then $B'$ should contain $x$, or we do not require that for every resolution step there is a unique resolved variable. We remark that soundness still holds and $\Reff$ is satisfiable if and only if there exists a refutation of length at most $s$, which can easily be read from the satisfying assignment to the $\Reff$ formula. We also remark that the lack of these axioms does not affect the extraction algorithm or the lower bound in any way: while removing axioms could in principle make the lower bound easier to prove, the algorithm works just as well if we added the missing axioms, and our lower bound proof still goes through with the additional axioms. This more succinct encoding, however, makes it easier to formalize the upper bound construction in Resolution.
\end{remark}

\begin{remark}[Number of variables]
    \label{rem:num_of_vars}
    The formula $\Reff_s(\varphi)$ is defined over $N = \Theta(s^2 + sm + sn + mn)$ variables and $M=\Theta(s^2n^2 + smn)$ clauses. For the case when the $\alit$ variables are restricted to encode a $k$-CNF formula over $n$ variables and $s = n^c$ for some constant $c\geq 1$, we have $m = O(n^k)$ and $N = O(n^{\max\{ {2c}, {c+k} \}})$.
\end{remark}

\paragraph{Blocks and block-width.}
If a variable is part of a block $B_i$, we say that it \emph{mentions} $B_i$. An important measure for us will be the \emph{block-width} of a given clause $C$ over the variables of $\Reff_s(\varphi)$. This is defined as the number of different blocks mentioned by the variables of the literals in $C$, not counting the root block~$B_\bot$. We denote this measure by $\bw(C)$, and generalize it to refutations by taking $\bw(\pi)$ to be the maximum block-width over all the clauses in $\pi$.

\subsubsection{The $\Reff$ formulas for other proof systems}
\label{subsec:prelim-reff-prff-others}

We will also be interested in the $\Reff$ formulas for proof systems other than Resolution. In general, for a Cook--Reckhow system $Q$, we denote by $\Reff^Q(\varphi, \pi)$ the formula stating that $\pi$ is a correct $Q$-refutation of an unsatisfiable $\varphi$. For convenience, in this context we always consider all proof systems as refutational systems. The formula $\Reff^Q(\varphi, \pi)$ is simply the propositional formula that verifies that $\pi$ is accepted as a $Q$-refutation of $\varphi$ by the Boolean circuit that checks $Q$-refutations. This can be obtained by writing the computation of the circuit as a Boolean formula using the usual Tseitin encoding.

\begin{remark}[Notation]
    \label{rem:notation-reff}
    By default, the $\Reff$ formula stands for the formula as defined in \Cref{def:Reff-formulas} for Resolution refutations. If we want to refer to the $\Reff$ formula for a different proof system, we explicitly write $\Reff^Q$ for the system $Q$ in question.
\end{remark}

\subsubsection{The $\Satf$ formula and reflection principles}
\label{subsec:prelim-sat-reflection}
We also have the $\Satf(\varphi, \alpha)$ formula, encoding that a CNF formula $\varphi$ is satisfied by an assignment $\alpha$. The variables we consider are

\bigskip

\begin{tabular}{lll}
    $\alpha_i$ & : & value assigned by $\alpha$ to variable $x_i$; \\
    $\alit^A_\ell$ & : & literal $\ell \in \Lit_n$ is present in clause $A \in [m]$; \\
    $\sat^A_\ell$ & : & clause $A \in [m]$ is satisfied because literal $\ell \in \Lit_n$ evaluates to $1$ under $\alpha$.
\end{tabular}
\bigskip

We use $\alit$ instead of $\lit$ to distinguish between these variables and the $\lit$ variables of $\Reff(\varphi, s)$.

\begin{definition}[The $\Satf$ formulas]
\label{def:satf}
Let $n, m \in \bbN$, let $\varphi$ denote the set of variables of the form $\alit^A_\ell$ as above and let $\alpha$  denote the set of variables $\alpha_i$ as above. The formula $\Satf(\varphi, \alpha)$ is the CNF formula over the variables in $\varphi$, $\alpha$ and additionally all the variables $\sat^A_\ell$ above consisting of the conjunction of the following clauses,
{\allowdisplaybreaks
\begin{axiomdef}{2}{Sat}
        & \neg\sat^A_\ell \lor \alit_\ell^A &&\quad\text{ for all }A\in[m]\text{ and }\ell\in \{x_1, \dots, x_n, \neg x_1, \dots, \neg x_n\},\label{sat-1}\\
        &\neg \sat^A_{x_i} \lor \alpha_i &&\quad\text{ for all }A \in [m]\text{ and } i \in [n],\label{sat-2}\\
        &\neg \sat^A_{\neg x_i} \lor \neg \alpha_i &&\quad\text{ for all }A \in [m]\text{ and } i \in [n],\label{sat-3}\\
        & \bigvee_{i \in [n]} \sat^A_{x_i} \lor \sat^A_{\neg x_i} &&\quad\text{ for all }A\in[m].\label{sat-4}
    \end{axiomdef}}

\end{definition}

One can similarly write a $\Satf$ formula for evaluating DNF formulas in the obvious ways. It is always clear from context which version of $\Satf$ we are using, so we use the same notation for both.

\begin{proposition}
    \label{prop:sat-to-native-Res}
    For every CNF formula $\varphi$, the following statements hold:
    \begin{enumerate}[label=(\roman*)] \itemsep=0pt
    \item if $\varphi$ has a Resolution refutation of length $s$, then
    $\Satf_{\restriction\varphi}$ has a Resolution refutation of length $O(s)$;
    \item if $\Satf_{\restriction\varphi}$ has a Resolution refutation of length $s$, then 
    $\varphi$ has a Resolution refutation of length $O(s)$.
    \end{enumerate}
\end{proposition}

\begin{proof}
    Assume $\varphi$ has $n$ variables $x_1,\ldots,x_n$ and $m$ clauses $C_1,\ldots,C_m$. We first show how to go from a refutation of $\Satf_{\restriction\varphi}$ to a refutation of $\varphi$. After fixing $\varphi$, the only variables left in $\Satf$ are of type~$\sat^A_\ell$ or~$\alpha_i$. Since Resolution is closed under literal substitutions, it suffices to define a substitution $\sigma$ that replaces all the original variables by literals of~$\varphi$, and argue that all the axioms of~$\Satf_{\restriction\varphi, \sigma}$ follow from clauses of~$\varphi$. For every $i \in [n]$, every $\ell \in \{ x_1, \dots, x_n, \neg x_1, \dots, \neg x_n\}$, and every $A \in [m]$, the substitution $\sigma$ maps
    \begin{equation}
        \sigma(\alpha_i) \coloneq x_i \qquad \text{ and } \qquad\sigma(\sat^A_\ell) \coloneq \begin{cases}
            \ell &\text{ if } \ell \in C_A \\
            0 &\text{ if } \ell \not\in C_A.
        \end{cases}
    \end{equation}
    Let us inspect the axioms of $\Satf_{\restriction\varphi, \sigma}$. For an axiom $\neg \sat^A_\ell \lor \alit^A_\ell$ of type (\ref{sat-1}), if the restriction of this axiom is present is $\Satf_{\restriction\varphi}$ it is because $\ell \not\in C_A$, or else $\alit^A_\ell$ would be satisfied. Then, $\sigma(\sat^A_\ell) = 0$ and the substitution is satisfied. If the axiom is of type (\ref{sat-2}), then either $x_i \not\in C_A$, in which case $\sigma(\sat^A_{x_i}) = 0$ and~$\neg \sat^A_{x_i} \lor \alpha_i$ is satisfied, or $\sigma(\sat^A_{x_i}) = x_i$ and $\sigma(\alpha_i) = x_i$, which gives the trivial clause $\neg x_i \lor x_i$. The case for (\ref{sat-3}) is analogous. Finally, for a clause of type (\ref{sat-4}), it is easy to see that the substitution $\sigma$ maps the clause precisely to the clause $C_A$ itself, which is a clause of $\varphi$.

    To go from a refutation of $\varphi$ to a refutation of $\Satf_{\restriction\varphi}$, we do the following. First, rename all the variables~$\alpha_i$ by $x_i$. Then, for every~$A \in [m]$, note that we can resolve the corresponding clause of type~(\ref{sat-4}) with the unit clauses $\neg \sat^A_{\ell}$ of type (\ref{sat-1}) to obtain $\bigvee_{\ell\in C_A} \sat^A_\ell$. Now, for each $\sat^A_\ell$ in this clause, cut with the corresponding clause of type (\ref{sat-2}) or (\ref{sat-3}) to obtain $\bigvee_{\ell\in C_A} \ell$, which is just $C_A$. In this way we have derived every axiom $C_A$ of $\varphi$, and we can now proceed with the refutation of $\varphi$ in the natural way.
\end{proof}

Finally, we define the reflection principle for any proof system.

\begin{definition}[The $\Reflf$ formulas]
\label{def:refl-formulas}
Let $Q$ be a Cook--Reckhow propositional proof system and let $n, m, s \in \bbN$. We define $\Reflf^Q_{n, m, s} \coloneq \neg\Satf(\varphi, \alpha)  \lor \neg \Reff^Q(\varphi, \pi)$ where the $\Satf$ instance is for CNF formulas with $n$ variables and $m$ clauses and the $\Reff^Q$ instance is for $Q$-refutations of such formulas of size $s$, and refer to the sequence of tautologies $\Reflf^Q \coloneq \{\Reflf^Q_{n, m, s}\}_{n, m, s \in \bbN}$ as the \emph{reflection principle for $Q$}. 
\end{definition}

A useful property of $\EF$ is the fact that any propositional system $S$, however strong, can always be seen as a Frege-like system due to the fact that ${\EF + \Reflf^S \psim S}$ \cite[Theorem 8.4.3]{Krajicek19}.

\subsubsection{Reflection principles for first-order theories}
\label{subsec:prelim-foRefl}
The reflection principles above are the propositional analogues of the well-studied
first-order reflection principles for theories of arithmetic.
For a strong enough recursively axiomatizable theory of arithmetic $T$ and a class $\Phi$ of sentences, the schema $\CRefl{\Phi}_T$ stands for the collection of all formulas of the form
\begin{equation}
    \exists \pi\foPrff_T(\ulcorner\varphi \urcorner, \pi) \to \varphi
    \tag{$\foReflf_{T, \varphi}$}
\end{equation}
for every $\varphi \in \Phi$. Here, the notation $\ulcorner\varphi\urcorner$ stands
for the encoding of $\varphi$ in a fixed suitable arithmetization of syntax, and $\foPrff_T$
stands for the provability statement in this arithmetization.

\section{The Proof Analysis Problem: definitions and basic facts}
\label{sec:PAP-def}
For a CNF formula $\varphi(x_1, \dots, x_n)$, we denote by $\Reff_s(\varphi)$ the propositional formula claiming that there exists a Resolution refutation of $\varphi$ in size $s$. Different encodings of this formula have been considered in the literature. For our purposes, $\Reff$ consists of $s$ of \emph{blocks} of variables, each of them describing a clauses in a purported Resolution refutation of $\varphi$ of size $s$. (See \Cref{subsec:prelim-Reff} for a full rendering of the variables and clauses involved in $\Reff_s(\varphi)$.)

We are interested in the following decision problem.

\begin{definition}[The Proof Analysis Problem, $\PAP_Q$]
    \label{def:PAP}
    Let $Q$ be a propositional proof system. We define the \emph{Proof Analysis Problem for $Q$} to be the language
    \[ \PAP_Q \coloneq \{ (\varphi, \pi, 1^s) \mid \varphi \in \SAT \text{ and } \pi : {Q \vdash \neg \Reff_s(\varphi)} \}.\]
    We denote by $\PAP_Q[s(n)]$ the problem where the size parameter $s$ is restricted to be at least $s(n)$ and $n$ denotes the number of variables of $\varphi$.
\end{definition}

The problem asks, given the proof of a Resolution lower bound in a fixed proof system $Q$, to decide whether the underlying formula is satisfiable or not. Note that whenever $\varphi$ is satisfiable there is no Resolution refutation and thus any lower bound holds, so the problem is well-defined.

Analogous to the notion of whether a proof system is automatable, $\PAP$ naturally induces a notion of whether, for a given proof system, its Resolution lower bounds are \say{analyzable}.

\begin{definition}[Analyzability]
    We say that a propositional proof system $Q$ is \emph{analyzable} if there exists some constant $c > 0$ such that $\PAP_Q[n^c] \in \P$.
\end{definition}

\begin{remark}
    It might seem more intuitive to define a proof system $Q$ to be analyzable if $\PAP_Q \in 
    \P$, without restrictions on the size parameter. Note, however, that for most reasonable proof systems, the language $\PAP_Q$ taken as a whole contains some degenerate instances that make the problem trivially $\NP$-hard. For example, if the size parameter is set to $s = 1$, then certainly proving a Resolution lower bound against $\varphi$ is easy already for Resolution itself, and we can map a CNF formula $\varphi$ to the $\PAP_Q$-instance $(\varphi, \pi, 1)$ for some easy to construct $Q$-proof $\pi$ that checks there is no Resolution refutation of $\varphi$ in one clause. 
\end{remark}

It is easy to see that for every Cook--Reckhow system $Q$, the problem $\PAP_Q$ is in $\NP$. Similarly, we note that unlike automatability, analyzability is naturally downwards-closed under p-simulations. Namely, if $S$ is p-simulated by $Q$ and $Q$ is analyzable, so is $S$; this is not the case with automatability, where a search algorithm for $Q$ may not be used to search for proofs in a weaker $S$. 

The following is a corollary of the results of Atserias and Müller \cite{AM20}. Here, by $\P$-uniform we mean the standard notion of uniformity by which there is a polynomial-time descriptor Turing machine that on input $1^\ell$ outputs the circuit solving the problem for inputs of size $\ell$ (see, e.g., \cites[Definition 6.12]{AB09}{AllenderStack}).

\begin{proposition}
    \label{prop:PAP_}
    It holds that $\PAP_\Res[n^2]$ is in $\P$-uniform $\ComplexityFont{AC}^0$. That is, Resolution is analyzable.
\end{proposition}

\begin{proof}
    Let us first describe the general polynomial-time algorithm that puts $\PAP_\Res[n^2]$ in $\P$, and we later elaborate on how this can be computed in $\P$-uniform  $\ComplexityFont{AC}^0$. Indeed, by the Resolution lower bound on $\Reff$ formulas (\Cref{thm:AM-original}), there exists $\varepsilon >0 $ such that for every $s\in \bbN$, if a formula $\varphi$ over $n$ variables is unsatisfiable, then a correct Resolution refutation $\pi$ of $\Reff_s(\varphi)$ must have size $|\pi| > 2^{\varepsilon \cdot s/n}$. Given an input $(\varphi, \pi, 1^s)$ to $\PAP_\Res[n^2]$, to decide if the instance belongs in the language, it suffices to check (i) that $\pi$ is a correct Resolution refutation of $\Reff_s(\varphi)$ and (ii) that $|\pi|$ is smaller than the lower bound $2^{\varepsilon \cdot s/n}$. If (i) fails, we immediately reject, and otherwise, if (ii) fails, the input size is large enough to brute-force $\SAT$ in polynomial time. Here we use the fact that $s \geq n^2$, hence $|\pi| \geq 2^{\varepsilon \cdot s/n} \geq 2^{\varepsilon n}$, and thus the input size is large enough.

    Let us now argue that this entire computation is possible within $\P$-uniform $\ComplexityFont{AC}^0$. For the sake of precision, let us fix the following natural binary encoding for $\PAP_\Res$. An input $(\varphi, \pi, 1^s)$ will be of the form $(1^n, 1^m, C_1, \dots, C_m, 1^t, \pi, 1^s)$. Here, the first part of the tuple corresponds to the encoding of $\varphi$, a CNF formula over $n$ variables and $m$ clauses $C_1, \dots, C_m$, and we assume that these clauses are initially represented as strings of length $2n$ with indicators for every possible literal. The Resolution refutation of $\Reff_s(\varphi)$ is encoded by $1^t$ and $\pi$, where $\pi$ is an assignment to the $N \coloneq N(n,m,t, s) = \poly(n,m,t,s)$ variables of $\Reff_t(\Reff_s(\varphi))$, as per \Cref{subsec:prelim-Reff}, and the number of variables $N$ can be easily computed in polynomial time. For the purpose of unique decoding, we assume that the tuple $(1^n, 1^m, C_1, \dots, C_m, 1^t, \pi, 1^s)$ is encoded by bit-doubling: each bit is duplicated and $01$ is used as separators. We assume that there are no separators between the clauses $C_1, \dots, C_m$, so a correct input contains only five separators.

    Now, when dealing with binary strings of even length $\ell$, there are at most $O(\ell^4)$ possible ways of interpreting such strings as a tuple of the form $(1^n, 1^m, C_1, \dots, C_m, 1^t, \pi, 1^s)$. This is because we can choose values for $n$, $m$, $t$ and $s$ in the interval $[\ell/2 - 5]$, where $\ell/2 - 5$ comes from the fact that we duplicated every bit and introduced 10 bits for the five separators between $1^n$, $1^m$, $C_1, \dots, C_m$, $1^t$, $\pi$, and $1^s$, not counting separators between $C_1$ and $C_m$. One can then check that this choice of $n$, $m$, $t$ and $s$ conforms to the desired pattern: the segment for the clauses $C_1$ to $C_m$ has length exactly $2nm$, and the segment for $\pi$ has length exactly $N = N(n,m, t, s)$, as per \Cref{rem:num_of_vars}. That is, it must hold that $\ell = 2( n + m + 2nm + t + N(n,m,t,s) + s)+ 10$ and $s \geq n^2$. There are at most $O(\ell^4)$ such choices for $(n,m,t, s) \in [\ell/2 - 5]^4$, hence the upper bound. Furthermore, it is easy to see that, due to the bit-doubling, 
    every string can only encode correctly one input of the form $(1^n, 1^m, C_1, \dots, C_m, 1^t, \pi, 1^s)$, so the decoding is unique.

    The $\P$-uniform descriptor machine for inputs of even length $\ell$ now works as follows. On input $1^\ell$, for $\ell$ even, it tries all possible $O(\ell^4)$ ways of separating the lengths, and for each interpretation $(n,m,t,s)$ of the lengths it constructs a different constant-depth Boolean circuit $D_{n,m,t,s}$, as follows.
    \begin{enumerate}[label=(\alph*)] \itemsep=0pt
        \item If the interpretations of the lengths is inconsistent, in the sense that the string cannot correspond to something of the form $(1^n, 1^m, C_1, \dots, C_m, 1^t, \pi, 1^s)$, then it outputs the constant circuit $0$. \label{it:outputzero}
    
        \item If the interpretation of the lengths is valid and $|\pi| < 2^{\varepsilon \cdot s/n}$, then it simply constructs the circuit that checks that $\pi$ is a correct Resolution refutation of $\Reff_s(\varphi)$ in at most $t$ clauses; that is, it outputs the formula $\Reff_t(\Reff_s(\varphi))$, which itself depends on the variables encoding $\varphi$. Since $\Reff$ formulas are in CNF, nesting these together with $\varphi$ will result in a total depth of 5 (see \Cref{subsec:Pudlak-as-ckt} for a more detailed treatment of how the depth increases when nesting the $\Reff$ formulas). \label{it:checkcorrect}

        \item If the interpretation of the lengths is valid and $|\pi| \geq 2^{\varepsilon \cdot s/n}$, then the descriptor outputs the conjunction of two circuits: one is the same as before, checking the correctness of $\pi$ as a refutation of $\Reff_s(\varphi)$ in at most $t$ clauses, and the other is the trivial circuit of size $\poly(n,m)\cdot 2^n$ that brute-forces the satisfiability of $\varphi$. More formally, this is a big disjunction of fan-in $2^n$, where each wire goes to the formula $\Satf(\varphi, \alpha)$ from \Cref{def:satf} for different hard-wired values of $\alpha \in \{ 0,1\}^n$. Since $\Satf(\varphi, \alpha)$ is a CNF formula, this brute-forcing circuit has depth $3$, and combined with the circuit checking the correctness of $\pi$, the entire circuit has depth 5 in this case. \label{it:checkcorrect-bruteforce}
    \end{enumerate}

    For each interpretation of the lengths there is also a circuit $\operatorname{Correct}_{n, m, t, s}$ that verifies that the input correctly encodes a $\PAP_\Res$ instance of the right size. This amounts to checking that the separators are in the right place and the double-bit encoding is correctly implemented, which can all be verified in depth $3$. 

    Finally, the descriptor machine outputs the circuit
    \begin{equation}
        R_\ell \coloneq\bigvee_{(n,m,s,t) \in[\ell/2 - 5]^4} D_{n,m,t,s} \land \operatorname{Correct}_{n,m,t,s}
    \end{equation}
    consisting of the disjunction of all the circuits above for every interpretation of the lengths.

    The final circuit $R_\ell$ correctly computes $\PAP_\Res[n^2]$ on inputs of even length $\ell$, has depth $6$ and polynomial size. Indeed, the constructions~\ref{it:outputzero} and~\ref{it:checkcorrect} above both have polynomial size, and whenever we construct the exponential-size circuit in~\ref{it:checkcorrect-bruteforce}, it is with respect to a segment of the string that has itself size exponential in~$n$. Finally, the descriptor machine runs in polynomial time given only the length $\ell$ of the input string $x$, so we can conclude that $\PAP_\Res[n^2]$ is in $\P$-uniform $\ComplexityFont{AC}^0$.
\end{proof}

The fact that $\PAP_\Res$ is so easy makes it natural to ask whether the same is true for the \emph{search version} of the problem.

\begin{definition}[Search version of $\PAP$]
    For a propositional proof system $Q$, we denote by $\FPAP_Q$ the \emph{search version of the Proof Analysis Problem for $Q$}, defined as follows.
    \searchProblemStatement{$\FPAP_Q$ (search version of $\PAP_Q$)}{A CNF formula $\varphi$, a size parameter $s$ in unary and a proof $\pi$ such that $\pi : Q \vdash \neg \Reff_s(\varphi)$.}{Either a satisfying assignment for $\varphi$, if $\varphi \in \SAT$, or $0$ otherwise.}
    Similarly to $\PAP_Q[s(n)]$, we define $\FPAP_Q[s(n)]$ to be the search problem where the size parameter is at least $s(n)$. 
\end{definition}

If we impose a polynomial upper bound on the size of $\pi$, then by the lower bound on $\Reff$ formulas, there is always a satisfying assignment for $\varphi$, and the problem ${\FPAP_{\Res}}[n^c]$ for any $c > 0$ is in $\TFNP$.
The fact that $\PAP_\Res[n^2] \in \P$, does not, however, directly imply that ${\FPAP_{\Res}}[n^c] \in \FP$ for any $c > 0$. Namely, it is not clear that given a polynomial-size proof $\pi$ of $\Reff_s(\varphi)$ one can extract a satisfying assignment of $\varphi$, even if one can conclude that $\varphi$ is satisfiable.  
We show in \Cref{sec:the-algorithm} that $\FPAP_{\Res}$ is in $\FP$---although the algorithm is not quite as straightforward as the one for the decision problem.

We see $\PAP$ and the analyzability of a proof system as closely related to automatability. The following proposition captures this idea and underlines the relevance on $\PAP$ in showing hardness of automatability.

\begin{proposition}
    \label{prop:aut-anal-duality}
    Let $Q \geq \Res$. If $Q$ is both analyzable and automatable, then $\P = \NP$.
\end{proposition}

\begin{proof}
    If $Q$ is analyzable and automatable, this means that $\PAP_Q[n^c] \in \P$ for some constant $c>0$, and that there is an automating algorithm $A$ for $Q$. We call the polynomial-time algorithm for $\PAP_Q[n^c]$ an \emph{analyzer} and observe that these two combined can solve $\TSAT$ in polynomial time as follows. Given a 3-CNF formula $\varphi$, construct the formula $\Reff_{n^c}(\varphi)$, stating that $\varphi$ does not have Resolution refutations of size $n^c$. Since $Q \geq \Res$, by the upper bound construction \cites[Theorem 4.1]{Pudlak03}[Lemma 11] {AM20}, whenever $\varphi$ is satisfiable, there will be size-$n^{O(1)}$ refutations in Resolution, and hence also in $Q$, and the automating algorithm $A$ will succeed in finding some refutation in polynomial time. Feed this refutation to the $Q$-analyzer to decide whether $\varphi \in \SAT$. If the automating algorithm failed to output a polynomial-size proof, then we would already know that $\varphi \not\in \SAT$.
\end{proof}

The previous proposition can be seen as an abstract way of stating the $\NP$-hardness of automating Resolution too. Since $\PAP_\Res[n^2] \in \P$, that means that Resolution cannot be automatable unless $\P = \NP$. In \Cref{sec:PAP-EF-NP-complete} we study the possibility of analyzing algorithms for strong proof systems actually leading to the hardness of their automatability---and establish that this is highly unlikely.

\section{The extraction algorithm}
\label{sec:the-algorithm}

This section proves that the search version of the Proof Analysis Problem for Resolution is in $\FP$. Our algorithm (in fact, two algorithms) arise from closely observing the lower bound on the $\Reff$ formulas and attempting to make it fully constructive, in a style amenable to formalization in weak theories of arithmetic like $\PVO$ (in the style of Cook and Pitassi \cite{CP90}).

Recall that the lower bound can be presented in two steps: first, a random restriction argument takes a small refutation and produces a low block-width refutation of a restricted formula, followed by a block-width lower bound for this restricted formula, which overall bounds the size of the original refutation.

Our algorithm works analogously. On input a refutation $\pi$ of $\Reff_s(\varphi)$, it first finds a restriction $\rho$ such that $\pi_{\restriction \rho}$ is a refutation of a restricted version of $\Reff_s(\varphi)$ and has low block-width. Then, we have a second algorithm that, inspired by the proof of the block-width lower bound, analyses this low block-width refutation and extracts a satisfying assignment.

We present the algorithm in two steps. First, two alternatives to perform block-width reduction are described in \Cref{subsec:width-reduction-algorithm}. These correspond, respectively, to a random restriction and a deterministic restriction argument. In \Cref{subsec:width-analysis-algo} we explain how to design the algorithm that analyzes low block-width refutations, which is essentially the Prover-Delayer strategy behind the block-width lower bound for the $\Reff$ formulas. Putting them together yields the desired procedure.

\subsection{The block-width reduction algorithm}
\label{subsec:width-reduction-algorithm}

Recall that we assume that in our definition of the $\Reff$ formula there is a variable $\enable^i$ for every block $i\in[s]$ that allows us to \emph{disable} that block (see \Cref{def:Reff-formulas}). For succinctness, we often denote $\enable^i$ simply by $e_i$ and refer to is as an \emph{enabling variable}.

Let us first define a kind of restriction that will come up a lot in our arguments.

\begin{definition}[Disabling restrictions]
    \label{def:disabling-restriction}
    We say that a restriction $\rho\in\{0,1, *\}^N$ to the $N$ variables of $\Reff_s(\varphi)$ is \emph{$d$-disabling} if it satisfies that (i) exactly $d$ blocks are disabled, and the rest are all enabled, (ii) every variable belonging to a disabled block is assigned a value, and (iii) no other variable is assigned.
\end{definition}

A key property of disabling assignments is that they can never falsify any axioms of $\Reff_s(\varphi)$, all the enabling variables disappear after the restriction and $\Reff_s(\varphi)_{\restriction \rho}$ is essentially an instance of $\Reff_{s-d}(\varphi)$, except there are some pointer variables pointing to the $d$ disabled blocks that are still hanging.

A first approach to perform block-width reduction in inspired by random restriction arguments, and requires randomness.

\begin{lemma}[Randomized block-width reduction]
\label{lemma:rand-wr}
        Let $p \in [0, 1)$ and let $N$ be the number of variables of the $\Reff_s(\varphi)$ formula for some CNF formula $\varphi$ over $n$ variables. There exists a randomized algorithm $R$ taking as input $1^N$, and outputting an $\lfloor{s/2}\rfloor$-disabling restriction $\rho \in \{0, 1, * \}^{N}$, such that for every Resolution refutation $\pi$ of the formula $\Reff_s(\varphi)$, the following properties hold:
        \begin{enumerate}[label=(\roman*)]\itemsep=0pt
            \item the restriction $\rho$ does not falsify $\Reff_s(\varphi)$;
            \item with probability at least $p$, the block-width of $\pi_{\restriction \rho}$ is at most $O\left( \log |\pi| - \log (1 - p)\right)$;
            \item the running time of $R(1^N)$ is $O(N)$.
        \end{enumerate}
\end{lemma}

\begin{proof}
    The proof follows closely the random restriction argument of \citeauthor{dRGNPRS21} \autocite[Section 6.3]{dRGNPRS21}. For simplicity, let us assume $s$ is even. (Note that, without loss of generality $s$ can be even, since if $\pi$ is a correct refutation of $\Reff_s(\varphi)$ and $s$ is odd, then $\pi$ can be turned into a refutation of essentially $\Reff_{s-1}(\varphi)$ by hitting $\pi$ with the restriction that disables and fully restricts one block). 

    Assume the root block corresponds to block $B_\bot$, which is not counted towards block-width, and consider the following random restriction: pair all $s$ blocks into $s/2$ pairs, and for each pair, with probability $1/2$, decide which block in the pair is going to be disabled. Now, if a block is disabled, all of its remaining variables are assigned uniformly at random. We denote by $\rho$ the restriction obtained in this way, which the algorithm outputs.

    We claim that with probability at least $p$, the restriction succeeds in lowering the block-width of any Resolution derivation $\pi$ to $O(\log |\pi| - \log (1-p))$. Indeed, if $\ell$ is a literal corresponding to the variable $e_i$ determining whether a certain block is disabled, then $\Pr_\rho[\ell_{\restriction \rho} = 1] = 1/2$. For every other literal $\ell$ in a block~$B_i$,
    \begin{align}
        \Pr_\rho [ B_i \text{ is disabled and }\ell_{\restriction \rho} = 1 \ ] &=\Pr_\rho [ B_i \text{ is disabled}] \cdot \Pr_\rho[\ell_{\restriction \rho} = 1 \mid  B_i \text{ is disabled}] \\ &= \frac{1}{2} \cdot \frac{1}{2} = \frac{1}{4}\,.
    \end{align}
    
    Hence, for every literal $\ell$ not from $B_\bot$, $\Pr_\rho[\ell_{\restriction \rho} = 1] \geq 1/4$.

    Now, if $C$ is a clause of block-width at least $w$, we have that
        \begin{equation}
        \Pr_\rho [C_{\restriction \rho} \neq 1] \leq (3/4)^{w/2}\,,    
        \end{equation}
    where the $1/2$ in the exponent comes from the fact that if two consecutive blocks are present, meaning that they were paired together and only one of them was enabled, their values depend on each other.

    Then, if $\pi$ was indeed a Resolution derivation of $\Reff_s(\varphi)$, by a union bound,
    \begin{equation}
        \Pr_{\rho}[\pi_{\restriction \rho} \text{ has a clause of block-width at most }w] \leq \prooflength{\pi} \cdot (3/4)^{w/2}
        \leq |\pi| \cdot (3/4)^{w/2} \,,
    \end{equation}
    which is the failure probability for property (ii) in the statement. For success probability at least $p$, we want to choose $w$ such that $|\pi| \cdot (3/4)^{w/2} \leq 1 - p$. This bound is met by choosing $w \geq 2(\log (|\pi|/(1-p))/(\log 4/3))$, meaning that with probability $p$, the restriction $\rho$ will satisfy all clauses of at least this width.

    It suffices to argue that properties (i) and (iii) are also satisfied. Indeed, by the way we designed the restriction, after applying $\rho$ there are no disabling variables left and all variables in the disabled blocks have been restricted, to this is exactly $s/2$-disabling.

    As for the running time, the algorithm is simply sampling the restriction, which takes time $O(N)$.
\end{proof}

We now move on to a fully deterministic algorithm that takes as input an actual refutation $\pi$ and outputs a restriction that \emph{always} manages to reduces the block-width.

\begin{lemma}[Deterministic block-width reduction]
\label{lemma:det-wr}
    There exists a constant $c > 0$ and a deterministic algorithm taking as input a Resolution refutation $\pi$ of the formula $\Reff_s(\varphi)$ over $N$ variables and outputting a $d$-disabling restriction $\rho \in \{0,1, *\}^{N}$ with $d \leq c/2 \cdot \big( \sqrt{s\log |\pi|} \big)$ such that
    \begin{enumerate}[label=(\roman*)]\itemsep=0pt
        \item the restriction $\rho$ does not falsify $\Reff_s(\varphi)$;
        \item the block-width of $\pi_{\restriction \rho}$ is at most $c\cdot \big( \sqrt{s\log |\pi|} \big)$;
        \item the algorithm runs in time $\poly(|\pi|, s)$.
    \end{enumerate}
\end{lemma}

\begin{proof}
    We employ a greedy strategy to construct the restriction, meaning that we look at all the clauses of high block-width and we iteratively choose to restrict a literal that kills a significant fraction of these clauses.
    
    More formally, let $w$ be a parameter to be optimized later, and given $\pi$, let $W$ denote the set of all clauses in $\pi$ with block-width at least $w$. Through the following iterative process we will enable and disable some blocks. Whenever we enable a block, we will also mark it as not active by keeping track of a set $\mathsf{ActiveBlocks} \subseteq [s]$, meaning that when choosing greedily the next literal to restrict, inactive blocks are not a valid choice; and whenever we disable a block, we add it to a set $D \subseteq [s]$ to keep track of it.

    \begin{boxAlgo}{Deterministic block-width reduction}{greedy}
    Repeat the following procedure iteratively, starting with $\rho \coloneq \emptyset$, $\mathsf{ActiveBlocks} \coloneq [s]$, and $D \coloneq \emptyset$, and stop whenever $W$ is empty:

    \begin{enumerate} \itemsep=0pt
        \item Find the most frequent block $i\in \mathsf{ActiveBlocks}$ among the ones mentioned in the clauses in $W$.
        \item Look at the literal $e_i$ of the variable used to disable block $i$.
        \begin{enumerate} \itemsep=0pt
            \item If $e_i$ appears positively in at least $1/3$ of all the clauses in $W$ that mention block $i$, then set $\rho \coloneq \rho \cup \{ e_i \mapsto 1 \}$, $\mathsf{ActiveBlocks} \coloneq \mathsf{ActiveBlocks} \setminus  \{ i\}$, $W \coloneq W_{\restriction \rho}$, and go back to step (1). 

            \item If $e_i$ does not appear positively in at least $1/3$ of the clauses in $W$ mentioning block $i$, then set $\rho \coloneq \rho \cup \{e_i \mapsto 0\}$, $W \coloneq W_{\restriction \rho}$, $D \coloneq D \cup \{i\}$, and for every other variable $x$ of block $i$, 

            \begin{enumerate} \itemsep=0pt
                \item if $x$ appears in $W$ positively more often than negatively, then set $\rho \coloneq \rho \cup \{x \mapsto 1\}$, and $W \coloneq W_{\restriction \rho}$;
                \item if $x$ appears in $W$ negatively more often than positively, then set $\rho \coloneq \rho \cup \{x \mapsto 0\}$, and $W \coloneq W_{\restriction \rho}$;
                \item repeat for every variable of block $i$.
            \end{enumerate}

            \item Once all the variables of block $i$ have been taken care of, go back to step (1). 
        \end{enumerate} 
    \end{enumerate}
    \end{boxAlgo}
    
    The procedure terminates once $W$ is either empty or all clauses in $W$ mention only blocks that are no longer in $\mathsf{ActiveBlocks}$. At this point, $|D|$ blocks have been disabled. For the remaining blocks that were not mentioned by any clause in $W$, enable all of them by setting the corresponding variables $e_i \mapsto 1$. This completes the construction of the restriction~$\rho$, and the algorithm outputs $\rho$.
    
    The procedure runs for at most $s$ iterations, since each iteration takes care of one of the blocks mentioned by the clauses in the initial $W$ and we never deal with a block twice. Therefore, the algorithm runs in time $\poly(|\pi|, s)$. 
    
    As for the correctness of the algorithm, the restriction $\rho$ is $d$-disabling by construction for $d = |D|$. It is left to argue that for a suitable choice of $w$, there exists a constant $c > 0$ such that
    $d \leq c/2 \big(\sqrt{s \log |\pi|}\big)$ and the block-width of $\pi_{\restriction \rho}$ is at most $c \cdot \big(\sqrt{s \log |\pi|}\big)$.

    We want to choose $w$ so that after $\ell \leq s$ iterations, the set $W$ becomes empty. At the first iteration, by an averaging argument, we know that the most frequent block is mentioned in at least a $w/s$ fraction of~$|W|$. More generally, at iteration $\ell$, block-width might have decreased up to $w-(\ell - 1)$ and up to $\ell - 1$ blocks may have become inactive, so the same averaging argument tells us that the most frequent active block is mentioned in at least a $(w-(\ell - 1))/(s-(\ell - 1))$ fraction of the clauses. Furthermore, observe that if at a given iteration block $i$ is the most frequent active block, we are not promised to kill all the clauses mentioning $i$, but we are guaranteed to kill at least $1/3$ of them. Indeed, if we enable block $i$ that is because it appeared in at least $1/3$ of all the clauses in $W$ mentioning $i$; and otherwise we are guaranteed to restrict at least $1/2$ of the remaining at least $2/3$ fraction of the clauses in $W$ mentioning $i$, which amounts to at least $1/3$ fraction.
    
    Therefore, if $\rho_{\ell}$ is the restriction built after $\ell$ iterations, 
    \begin{align}
        |W_{\restriction \rho_\ell}| &\leq |W| \cdot {\bigg( {1- \frac{w}{3s}} \bigg)} \cdot \left( 1- \frac{w-1}{3(s-1)} \right) \cdot \dots \cdot \left( 1- \frac{w-(\ell - 1)}{3(s- (\ell - 1))} \right) \\
        &\leq |W| \cdot \left(1 - \frac{w - \ell}{3s}\right)^\ell \\
        &\leq |W| \cdot e^{- \ell \cdot \frac{w - \ell}{3s}}\,.
    \end{align}

    We want to ensure that for some $\ell \leq s$ we achieve $|W_{\restriction \rho_{\ell}}| < 1$. It suffices to have $|W| \cdot e^{- \ell \cdot \frac{w - \ell}{3s}} < 1$. Taking logarithms on both sides we have
    \begin{equation}
    \ln |W| < \ell \cdot \frac{w- \ell }{3s} \,,\end{equation}
    which holds already for $\ell = \lfloor w/2 \rfloor$, assuming $w > \sqrt{12s\ln|\pi|} \ge \sqrt{12s\ln|W|}$.

    Now it suffices to choose a constant $c$ such that $w \coloneq c \cdot  \sqrt{s \log |\pi|} > \sqrt{12s\ln|\pi|}$.  In this way we get that after at most $\ell \coloneq \lfloor w/2 \rfloor$ iterations, $W_{\restriction \rho_{\ell}} = \emptyset$ and thus the block-width of $\pi_{\restriction \rho_{\ell}}$ is also at most $c \cdot \sqrt{s\log|\pi|}$. Furthermore, note that $|D| \leq \ell$, since the algorithm only runs for at most $\ell$ iterations, meaning that $\rho_{\ell}$ is $d$-disabling for $d \leq \ell = \lfloor w/2 \rfloor \leq (c/2) \sqrt{s\log|\pi|}$, as desired.
\end{proof}

\subsection{The block-width analysis algorithm}
\label{subsec:width-analysis-algo}
Using one of the two algorithms above, we can take a refutation $\pi$ of $\Reff_s(\varphi)$ and obtain a new refutation $\pi'$ of the restricted formula $\Reff_s(\varphi)_{\restriction \rho}$ in low block-width. We can now show how to analyze this refutation, inspired by the block-width lower bound, and succeed in finding a satisfying assignment whenever one exists.

We first state the following simple but crucial technical fact used in the proof.

\begin{fact}
    \label{fact:weakening}
    Let $\varphi$ be CNF formula over $n$ variables. If $C$ is a non-tautological width-$n$ clause over the variables of $\varphi$ that is not the weakening of any clause of $\varphi$, then the unique assignment that falsifies $C$ is a satisfying assignment for $\varphi$.
\end{fact}

Now we can present the algorithm and prove its correctness.

\begin{lemma}[Assignment extraction]
\label{lemma:extraction}
    There exists a deterministic algorithm $E$ such that for every $s\in\bbN$, every $\pi$ a purported Resolution refutation of $\Reff_s(\varphi)$ for a CNF formula $\varphi(x_1, \dots, x_n)$ with $m$ clauses, and every $\rho \in \{ 0,1,*\}^N$ a $d$-disabling restriction to the $N$ variables in $\Reff_s(\varphi)$, it holds that $E(\varphi, \rho, s, \pi)$ terminates in time $\poly(|\pi|, s, n, m)$ and provides exactly one of the following outputs:
    \begin{enumerate} \itemsep=0pt
        \item[(a)] an incorrect derivation step in $\pi$;
        \item[(b)] a clause $C \in \pi_{\restriction\rho}$ of block-width at least $1/3 \lfloor (s-d - n)/n\rfloor$;
        \item[(c)] a satisfying assignment for $\varphi$.
    \end{enumerate}
\end{lemma}

\begin{proof}
    We traverse the refutation $\pi$ inspired by the Delayer's strategy in the Prover-Delayer game that yields a block-width lower bound for the restricted $\Reff$ formulas---and the correctness of the algorithm is essentially the proof of this block-width lower bound.

    Before the traversal of $\pi$ starts, we arrange the $s-d$ blocks enabled by $\rho$ in a layered manner, so that there are $n$ layers, each containing $\lfloor(s-d)/n\rfloor$ blocks (with the remainder blocks left from the flooring operations collected all in the last layer), plus one additional layer on top with a single block corresponding to the root. We see the root as laying at layer $0$, and intuitively blocks in layer $i$ will be obtained by resolving over variable $x_{i+1}$. Throughout the traversal, we keep a record $\alpha \in \{ 0, 1, *\}^N$ which we call the \emph{reservation}, in which we \say{reserve} information needed to continue the traversal. (Intuitively, this record keeps the information that the Delayer has in mind when playing against the Prover.) A reservation is a partial assignment to the variables of $\Reff_s(\varphi)$, and hence encodes patches of a potential refutation. In particular, it will encode information about how blocks are connected: say, some block $B$ will be derived from block $B'$, meaning that the reservation encodes that the left pointer of $B'$ is $B$, etc. Hence, when talking about a reservation, we will often say that a certain block has \emph{information about their parents} or \emph{information about their children}, to mean that such a connection is registered in the reservation. \Cref{fig:prove-delayer-example} represents the layout used by the Delayer for his strategy together with such a reservation.

    The algorithm proceeds as follows. First, let $\pi \coloneq \pi_{\restriction \rho}$, initialize $\alpha \coloneq \rho$, and collect in a set $D\subseteq [s]$ the $d$ blocks that are disabled by the restriction $\rho$. Note that $\rho$ is a $d$-disabling restriction, meaning that it only enables and disables blocks, and sets the value of all the variables in disabled blocks, meaning that the restriction cannot possibly falsify any axiom of $\Reff_s(\varphi)$. Note as well that in $\pi$, after hitting it with $\rho$, there are no longer resolution steps over enabling variables nor over variables belonging to a disabled block.

    In what follows, capitals letters in roman font, like $A$, refer to clauses in the refutation $\pi$, while capital letters in calligraphic font, like $\calA$, refer to clauses encoded in the blocks of $\Reff_s(\varphi)$ which are determined by assignments to the variables of $\Reff_s(\varphi)$.

\begin{boxAlgo}{Block-width analysis and assignment extraction}{prover-delayer}
Let $C$ be the root of $\pi$ and traverse the proof dag following these instructions.
  \begin{enumerate}
\itemsep=0pt
    \item If $C$ is obtained by an illegal derivation step, halt and output $C$; if $C$ is a leaf of $\pi$, halt and output failure. Otherwise, continue to Step 2.
    
    \item If $C$ was obtained from weakening a clause $C' \subseteq C$, then set $C \coloneq C'$ and move to Step 4.

    \item If $C$ was derived from clauses $A \lor v$ and $B \lor \neg v$ by resolving over $v$, attempt the following reservations according to these rules (with the condition that blocks mentioned in the set $D$ can never be reserved).

    \begin{enumerate}
    \itemsep=0pt
        
        \item If $v$ belongs to a block on layer $0 \leq i < n$, where layer $0$ stands for the root block, we have two cases:
        
        \begin{enumerate}
            \itemsep=0pt
            \item if $\alpha$ has no information about this block, update $\alpha$ by reserving two unreserved blocks on layer $i+1$ so that the block of $v$ encodes the clause $\bigvee_{j = 1}^i x_j$ and it was obtained by resolving the clauses $x_{i+1} \lor \bigvee_{j = 1}^i x_j$ and $\neg x_{i+1} \lor \bigvee_{j = 1}^i x_j$, to be encoded in the two reserved blocks on layer $i+1$;

            \item if the block was reserved in $\alpha$ but it had no children attached, then $\alpha$ already determined the clause $\calC$ that is to be encoded in this block. Then, update $\alpha$ by reserving two unreserved blocks on layer $i+1$ and so that that the block of $v$ encodes the clause $\calC$ and it was obtained by resolving the clauses $x_{i+1} \lor \calC$ and $\neg x_{i+1} \lor \calC$, to be encoded in the two reserved blocks on layer $i+1$.

            \item Otherwise, do nothing.
            
        \end{enumerate}

        \textbf{If this reservation fails} because there are not enough free blocks available, halt and output the clause $C$.
        
        \item If $v$ belongs to a block on layer $n$, we distinguish two cases:
        
        \begin{enumerate}
        \itemsep=0pt
            \item if the block is not mentioned in $\alpha$, then try to find the first axiom in $\varphi$ such that $\bigvee_{j=1}^n x_j$ is a weakening of it, and reserve it so it encodes $\calA \coloneq \bigvee_{j=1}^n x_j$ and its pointers point to this axiom;

            \item if the block was reserved in $\alpha$ but it was not pointing to any axiom, then there is a clause $\calA$ associated to it by the reservation; try to find the first clause in $\varphi$ such that $\calA$ is a weakening of it and point to it in $\alpha$.

            \item Otherwise, do nothing.
        \end{enumerate}

        \textbf{If this reservation fails} because no axiom could be found, then halt and output the assignment to the variables $x_1, \dots, x_n$ given by $\neg \calA$, the negation of the clause encoded by $\calA$.
    \end{enumerate}

        If the reservations succeeded, look at $\alpha(v)$, which is now guaranteed to be defined. If $\alpha(v) = 1$, then move to $C \coloneq B \lor \neg v$, and otherwise move to $C \coloneq A \lor v$.

    \item Clean-up the reservations in $\alpha$ as follows: $\alpha$ should only contain information about (i) blocks disabled by $\rho$ and (ii) blocks mentioned in $C$ or possibly the children of these according to $\alpha$. Erase all other information from $\alpha$.

    \item Go to Step 1.
\end{enumerate}
\end{boxAlgo}

Since the algorithm is only traversing a path inside the proof dag of $\pi$, the running time of this procedure is never longer than a polynomial in the size of $\pi$. As for the correctness of the algorithm, we now show that this behaves exactly as claimed. The central claim is that, at the beginning of each iteration, when looking at clause $C$, the following invariant is satisfied. Here, $\bw(\alpha)$ stands for the number of blocks mentioned by the variables assigned by $\alpha$, and $\bw(C)$ is the block-width of a clause $C$.

\begin{claim*}[Invariant]
    The following hold at the beginning of each iteration of the algorithm, when dealing with the clause $C$ in the traversal of $\pi$:
    \begin{enumerate}
    \itemsep=0pt
        \item[(i)] the reservation $\alpha$ falsifies $C$;
        \item[(ii)] a block $B \not\in D$ is only reserved in $\alpha$ if it is either mentioned in $C$ or its parent according to $\alpha$ is mentioned in $C$, and, in particular, $\bw(\alpha) - d\leq 3\bw(C)$;
        \item[(iii)] if $\alpha$ encodes any information about a block $B$ from layer $i$, and $B \not\in D$, then $\alpha$ also determines that $B$ contains exactly $i$ literals over the variables $x_1, \dots, x_i$ and no two literals for the same variable;
        \item[(iv)] the reservation $\alpha$ does not falsify any axiom of $\Reff_s(\varphi)$.
    \end{enumerate}
\end{claim*}

\begin{proof}[Proof sketch]
    The invariant is readily verified at the initial iteration, when $C = \bot$ and $\alpha = \rho$, the $d$-disabling restriction given as input.

    Now, by straightforward structural induction, it is easy to see that assuming that the invariant holds at the beginning of an iteration and the algorithm correctly proceeds to the next iteration without halting, the invariant holds again at the beginning of the new iteration.
\end{proof}

Note that, if the algorithm reached a clause $C$ that happened to be a leaf of $\pi$, then by point (i) of the invariant $\alpha$ would be falsifying $C$, which would mean falsifying an axiom of $\Reff_s(\varphi)$, contradicting point (iv) of the very same invariant. Thus, if $\pi$ is a correct Resolution refutation, that means that the algorithm always halts before reaching a leaf, and always for one of the following two reasons.

\begin{enumerate}
\itemsep=0pt
    \item[(a)] The algorithm attempted the reservation of a block at level $1 \leq i < n$, but there were no free blocks left. This means that $\alpha$ already reserved at least $\lfloor (s-d)/n \rfloor - 1$ blocks on that layer, and so $\bw(\alpha) \geq \lfloor(s-d)/n \rfloor - 1 + d$, since each layer $i < n$ contains exactly $\lfloor (s-d)/n \rfloor $ blocks. By point (ii) of the invariant we have that that $\bw(\alpha) - d \leq 3\bw(C)$, so putting this together we have that when outputting $C$ we are outputting a clause of block-width at least $1/3 \lfloor (s-d - n)/n\rfloor$, as desired.

    \item[(b)] The reservation $\alpha$ had a clause $\calA$ encoded in a block at layer $n$, but it failed to find an axiom of $\varphi$ that $\calA$ was a weakening of. By point (iii) of the invariant, since the block is at layer $n$, it encodes a width-$n$ clause, and by \Cref{fact:weakening}, $\neg \calA$ encodes a satisfying assignment of $\varphi$. In this case the algorithm outputs this assignment, which satisfies the desired behavior.
\end{enumerate}

This completes the proof of correctness of the algorithm.
\end{proof}

\begin{figure}
    \centering
    \begin{tikzpicture}[
    font=\Large,
    box/.style={minimum width=3.8cm, minimum height=1cm, rounded corners=12pt},
    plainbox/.style={box,draw,thick},
    orangebox/.style={box,draw=orange!90!black,fill=orange!40,very thick},
    greenbox/.style={box,draw=green!60!black,,fill=green!40,very thick},
    yellowbox/.style={box,draw=yellow!70!black,fill=yellow!50,very thick},
    ovalbox/.style={draw, very thick, blue!70!black, fill=yellow!30,
      minimum width=3.9cm, minimum height=1.2cm,ellipse}
      , arrow/.style = {-{Stealth[length=4mm]}, ultra thick},
      redarrow/.style = {arrow, red!80!black},
      yellowarrow/.style = {arrow, yellow!80!black},
      bluearrow/.style = {arrow, blue!70!black}
]

\newlength{\vsep}
\newlength{\hsep}
\newlength{\bigvsep}
\newlength{\bighsep}

\setlength{\hsep}{1cm}
\setlength{\vsep}{1cm}
\setlength{\bighsep}{3cm}
\setlength{\bigvsep}{2cm}

\node[plainbox] (a1) {};
\node[orangebox, right=\hsep of a1] (a2) {$x_1$};
\node[orangebox, right=\hsep of a2] (a3) {$\neg x_1$};
\node[plainbox, right=\bighsep of a3] (a4) {};
\node at ($(a3)!0.5!(a4) + (0,0)$) {$\scalebox{2}{$\ldots$}$};

\node[circle, red!80!black, draw, very thick, minimum size=1cm, fill=orange!40](bot) at ($(a2.north)!0.5!(a3.north) + (0,\vsep+10)$) {$\bot$};

\draw[redarrow] (bot) -- (a2.north);
\draw[redarrow] (bot) -- (a3.north);

\node[orangebox, below=\vsep of a1] (b1) {$x_1 \vee x_2$};
\node[orangebox, below=\vsep of a2] (b2) {$x_1 \vee \neg x_2$};
\node[plainbox, below=\vsep of a3] (b3) {};
\node[greenbox, below=\vsep of a4] (b4) {$x_1 \vee x_2$};
\node at ($(b3)!0.5!(b4) + (0,0)$) {\scalebox{2}{$\ldots$}};

\draw[redarrow] (a2.south) -- (b2.north);
\draw[redarrow] (a2.south) -- (b1.north);

\node[plainbox, below=\vsep of b1] (c1) {};
\node[orangebox, below=\vsep of b2] (c2) {$x_1 \vee \neg  x_2 \vee x_3$};
\node[orangebox, below=\vsep of b3] (c3) {$x_1 \vee \neg x_2 \vee \neg x_3$};
\node[plainbox, below=\vsep of b4] (c4) {};
\node at ($(c3)!0.5!(c4) + (0,0)$) {\scalebox{2}{$\ldots$}};

\draw[redarrow] (b2.south) -- (c2.north);
\draw[redarrow] (b2.south) -- (c3.north);

\node[plainbox, below=\bigvsep of c1] (d1) {};
\node[plainbox, right=\hsep of d1] (d2) {};
\node[yellowbox, right=\hsep of d2] (d3) {$C \lor \neg x_{n-2}$};
\node[greenbox, right=\bighsep of d3] (d4) {$\bigvee_{i=1}^{n-2} x_i$};
\node at ($(d3)!0.5!(d4) + (0,0)$) {\scalebox{2}{$\ldots$}};

\node[greenbox, below=\vsep of d1, , fill=green!40] (e1) {$\bigvee_{i=1}^{n-1} x_i$};
\node[yellowbox, below=\vsep of d2] (e2) {$C \lor \neg x_{n-2} \lor x_{n-1}$};
\node[yellowbox, below=\vsep of d3] (e3) {$C \lor \neg x_{n-2} \lor \neg x_{n-1}$};
\node[plainbox, below=\vsep of d4] (e4) {};
\node at ($(e3)!0.5!(e4) + (0,0)$) {\scalebox{2}{$\ldots$}};

\draw[yellowarrow] (d3.south) -- (e2.north);
\draw[yellowarrow] (d3.south) -- (e3.north);

\node[plainbox, below=\vsep of e1] (f1) {};
\node[yellowbox, below=\vsep of e2] (f2) {\large{$C \lor \neg x_{n-2} 
\lor \neg x_{n-1} \lor x_n$}};
\node[yellowbox, below=\vsep of e3] (f3) {\large{$C \lor \neg x_{n-2} 
\lor \neg x_{n-1} \lor \neg x_n$}};
\node[plainbox, below=\vsep of e4] (f4) {};
\node at ($(f3)!0.5!(f4) + (0,0)$) {\scalebox{2}{$\ldots$}};

\draw[yellowarrow] (e3.south) -- (f2.north);
\draw[yellowarrow] (e3.south) -- (f3.north);

\node at ($(c1)!0.5!(d1)+ (0,0.2)$) {\scalebox{2}{$\vdots$}};
\node at ($(c2)!0.5!(d2)+ (0,0.2)$) {\scalebox{2}{$\vdots$}};
\node at ($(c3)!0.5!(d3)+ (0,0.2)$) {\scalebox{2}{$\vdots$}};
\node at ($(c4)!0.5!(d4)+ (0,0.2)$) {\scalebox{2}{$\vdots$}};

\node[ovalbox,
      below=\vsep of f2] (g1)
{$\neg x_{n-2} \lor x_n$};

\node[ovalbox,
      below=\vsep of f3] (g2)
{$C \lor \neg x_n$};

\draw[bluearrow] (f2) -- (g1);
\draw[bluearrow] (f3) -- (g2);

\draw[<->, thick, black] 
($(a4.north east)+(1,0)$) --
node[right] {{$n$}}
($(f4.south east)+(1,0)$);

\draw[<->, thick, black] 
($(f1.south west)+(0,-3)$) --
node[below] {{$s/n$}}
($(f4.south east)+(0,-3)$);

\end{tikzpicture}
    \caption{Snapshot of the execution of {Algorithm~\ref{algorithm:prover-delayer}}. The Prover--Delayer game is played over the formula $\Reff_s(\varphi)$ for $s$ blocks, arranged in $n$ layers with $s/n$ blocks per layer. In \textcolor{orange!90!black}{orange} and in \textcolor{yellow!70!black}{yellow} are two partial trees in which the Delayer answers consistently by locally following the canonical exponential-size tree-like refutation. The \textcolor{orange!90!black}{orange} tree starts from the root, while the \textcolor{yellow!70!black}{yellow} one starts from some clause $C$ of width $n-3$ that emerged from previous queries. In \textcolor{green!60!black}{green}, blocks that were queried ``out of context'': at layer $i$, the Delayer answers with the clause $x_1 \lor x_2 \lor x_3 \lor \dots \lor x_i$ where all $i$ literals are positive. At layer $n$ there are width-$n$ clauses pointing to some axiom of $\varphi$ (in \textcolor{blue!70!black}{blue}) in a valid weakening step.}
    \label{fig:prove-delayer-example}
\end{figure}

\subsection{Putting it together}
\label{subsec:putting-algo-together}

The following is a formal restatement of \Cref{thm:algo-informal}.

\begin{theorem}
\label{thm:det-algo}
It holds that ${\FPAP_{\Res}}[n^3] \in \FP$. That is, there exists a deterministic polynomial-time algorithm solving the search version of the Proof Analysis problem for Resolution whenever the lower bound parameter satisfies $s \geq n^3$.
\end{theorem}

\begin{proof}
Let $(\varphi(x_1, \dots, x_n), \pi, 1^s)$ be an instance of the Proof Analysis Problem with $s \geq n^3$. First, check whether $\pi$ is indeed a correct Resolution refutation of $\Reff_s(\varphi)$. If not, reject. 
Otherwise, run the algorithm from \Cref{lemma:det-wr} on~$\pi$, which outputs a $d$-disabling restriction~$\rho$, with $d\leq (c/2)\sqrt{s\log |\pi|}$, such that $\pi_{\restriction \rho}$ is a Resolution refutation of $\Reff_s(\varphi)_{\restriction\rho}$ of block-width at most $c\cdot\sqrt{s\log|\pi|}$, for some fixed constant $c$. Then run the extraction algorithm from \Cref{lemma:extraction} on $\pi$ and $\rho$. Since $\pi$ is a correct refutation, it must be the case that 
the extraction algorithm from \Cref{lemma:extraction} outputs either a clause of $\pi_{\restriction \rho}$ of block-width at least $(s-d-n)/3n$
or a satisfying assignment of $\varphi$.
In the latter case, we are done. In the former case, it holds that 
$1/3\lfloor(s-d-n)/n \rfloor \leq c\cdot\sqrt{s\log|\pi|}$,
which implies $|\pi| > 2^{\varepsilon s/n^2} \ge 2^{\varepsilon \cdot n}$ for some small enough $\varepsilon$ and sufficiently large $n$. Since~$\pi$ is so large we can, in time polynomial in the size of the input, go over all $2^n$ assignments to the variables of $\varphi$ and output
a satisfying assignment of $\varphi$ if one exists, and otherwise reject.
\end{proof}

If we are interested in inputs where the lower bound is quadratic instead of cubic, then we can still achieve polynomial time at the cost of randomness.

\begin{theorem}
   \label{thm:rand-algo}
   There exists a zero-error randomized polynomial-time algorithm solving the search version of the Proof Analysis problem for Resolution whenever the lower bound parameter satisfies $s \geq n^2$. That is, ${\FPAP_{\Res}}[n^2] \in \ComplexityFont{FZPP}$.
\end{theorem}

\begin{proof}
    We carry out the proof for a fixed constant success probability $p$, but the argument works for any $p < 1$. The algorithm is essentially the same as before, except we now use the randomized width-reduction procedure in \Cref{lemma:rand-wr} at the beginning, instead of the greedy deterministic one.
    
    Observe that the randomized algorithm in \Cref{lemma:rand-wr} can be used with zero-error, because once a restriction $\rho$ is sampled, we can check if it successfully reduces the block-width to the desired bound, and run it again as many times as needed, which puts us in $\ComplexityFont{FZPP}$.

    As for why $s$ can now be allowed to be $n^2$, observe that the randomized procedure achieves better width reduction, of $O(\log |\pi|)$ whenever $p$ is a constant. Combining this with the $(s-d-n)/3n$ bound of block-width given by \Cref{lemma:extraction}, which in this case becomes $s/6n - 1/3$ because $d = s/2$, we now have $|\pi| > 2^{\Omega{(s/n)}}$, meaning that $s \geq n^2$ suffices to obtain an exponential lower bound.
\end{proof}

As discussed in introduction, the deterministic algorithm in \Cref{thm:det-algo} gives us a Levin reduction between $\TSAT$ and the Proof Size Problem for Resolution ($\PSP_\Res$, see \Cref{subsec:prelim-automatability}). Here the search version of $\TSAT$ if the one that finds satisfying assignments of satisfiable formulas, while the search version of $\PSP_\Res$ consists in finding a Resolution refutation of the right size.

\begin{corollary}
    \label{cor:Levin-NP-hard}
    There is a polynomial-time Levin reduction from the search problem for $\TSAT$ to the Proof Size Problem for Resolution. 
\end{corollary}

\begin{proof}
In~\cite{AM20}, a 3-CNF formula $\varphi$ is mapped to
the formula $\Reff_{n^2}(\varphi)$. If, instead, we map it to $\Reff_{n^3}(\varphi)$, then this is still a many-one reduction from $\TSAT$ and hardness of automatability still follows, but this is now a Levin reduction: given a satisfying assignment for $\varphi$, we can always come up with a short Resolution refutation of $\Reff_{n^3}(\varphi)$ using the standard upper bound (see \Cref{sec:Pudlak-UB}); and given a Resolution refutation $\pi$ of $\Reff_{n^3}(\varphi)$ of polynomial-size, we can extract a satisfying assignment for $\varphi$ using \Cref{thm:det-algo}.
\end{proof}

\subsection{Assignment extraction as information efficiency}
\label{subsec:Krajice-info-efficiency}
The statement that Resolution refutations of $\Reff_s(\varphi)$ must \say{leak} satisfying assignments has an eminently information-theoretic flavor. This is no coincidence, as \Cref{thm:det-algo} can be interpreted in terms of the \emph{information efficiency} of Resolution refutations, a variant of Kolmogorov complexity in the context of proof complexity.

Krajíček \cite{Krajicek22} introduced the notion of information efficiency to measure the complexity of describing propositional proofs based on the well-established concept of time-bounded Kolmogorov complexity. Fix a universal Turing machine $U$ such that $U(e, x, 1^t)$ simulates the machine with code $e$ on input $x$ for $t$ steps. For every string $x \in \{ 0,1\}^*$, the \emph{conditional time-bounded Kolmogorov complexity} (or \emph{conditional Levin complexity}) of $x$ given a string $y \in \{ 0,1\}^*$ is defined as
\begin{equation}
 \Kt(x \mid y) \coloneq \min\{|e| + \lceil \log t \rceil \mid U(e, y, 1^t) = x\}\,.
 \end{equation}
Then, for a propositional proof system $S$, the \emph{information efficiency} of a formula $\varphi \in \TAUT$ is defined as
\begin{equation}
 \info_S(\varphi) \coloneq \min \{  \Kt(\pi \mid \varphi) \mid \pi:S \vdash \varphi  \}\,.
 \end{equation}
That is, $\info_S(\varphi)$ is the is the smallest Levin-complexity of a proof $\pi$ of $\varphi$, given the formula $\varphi$. It is easy to see that for every $\varphi\in \TAUT$, the bounds $\log \size{S}{\varphi} \leq \info_S(\varphi) \leq O(\size{S}{\varphi})$ hold. The lower the information efficiency of a formula, the easier its proofs are to describe.

Because of the assignment extraction algorithm, if $\pi$ is a correct Resolution refutation of $\Reff_s(\varphi)$, then the description of $\pi$ must at least include the description of a satisfying assignment for $\varphi$. This means that the information efficiency of $\Reff_s(\varphi)$ should be (up to constant factors) precisely the same as the Levin complexity of the easiest satisfying assignment for $\varphi$. The following statement makes this precise and can be seen as an information-theoretic analogue of \Cref{thm:det-algo}.

\begin{theorem}[Assignment extraction as information efficiency]
\label{thm:info-efficiency}
    For every CNF formula $\varphi$ over $n$ variables and $\poly(n)$ clauses and every $s \geq n^3$, if $\varphi$ is satisfiable, then
        \[\info_\Res(\neg \Reff_s(\varphi))  = \Theta\left( \min \{ \Kt(\alpha \mid \varphi)   \mid \varphi(\alpha)=1 \}\right).\]
    On the other hand, if $\varphi$ is unsatisfiable, then $\info_\Res(\neg \Reff_s(\varphi)) = \Omega(n)$.
\end{theorem}

\begin{proof}
    We prove first that $\info_\Res(\neg \Reff_s(\varphi))  = O\left( \min \{ \Kt(\alpha \mid \varphi) \mid \varphi(\alpha) = 1 \} \right)$. Let $\alpha^\star$ be a satisfying assignment that minimizes $\Kt(\alpha \mid \varphi)$ with a description $e_{\alpha^\star}$ and time $t_{\alpha^\star}$, and let $P(\varphi, \alpha, s)$ denote the Turing machine that constructs Pudlák's upper bound. Then,
    \begin{align}
        \info_\Res(\neg \Reff_s(\varphi))  &\leq \Kt(P(\varphi, \alpha^\star, s) \mid \neg \Reff_s(\varphi))  \\
        &\leq |e_{\alpha^*}| + \lceil \log t_{\alpha^*}\rceil + O(1) + O(\log n) \label{step:log}\\
        &= O(\Kt(\alpha^\star \mid \varphi)) \\
        &= O\left( \min \{ \Kt(\alpha \mid \varphi) \mid \varphi(\alpha) = 1 \} \right),
    \end{align}
    where in \cref{step:log} the constant term $O(1)$ corresponds to the code of Pudlák's algorithm $P$, and since $\varphi$ is a CNF with $\poly(n)$ clauses and $P$ runs in polynomial time, the log-term accounting for running $P$ is $\lceil \log n^{O(1)}\rceil = O(\log n)$.

    For the other direction, let $\pi^\star$ be a Resolution refutation that minimizes $\info_\Res(\neg \Reff_s(\varphi))$, with description $e_{\pi^\star}$ and time $t_{\pi^\star}$. Let $A$ denote the assignment extraction algorithm from \Cref{thm:det-algo}. Then,
    \begin{align}
          \min \{ \Kt(\alpha \mid \varphi) \mid \varphi(\alpha) = 1 \} &\leq \Kt(A(\pi^\star, \varphi) \mid \varphi) \label{eq:extract}\\
          & \leq |e_{\pi^\star}| +  \lceil \log t_{\pi^\star} \rceil + O(1) + O(\log n) \label{step:log2} \\
          &= O(\Kt(\pi^\star \mid \varphi)) \\
          &= O(\Kt(\pi^\star \mid \neg\Reff_s(\varphi))) \\
          &= O\left(\info_\Res\left(\neg\Reff_s(\varphi\right)\right).
    \end{align}
    We have that \cref{eq:extract} follows because $s \geq n^3$ so \Cref{thm:det-algo} guarantees the extraction algorithm $A$ succeeds in finding a satisfying assignment, while the $O(\log n)$ term in \cref{step:log2} follows again because the algorithm $A$ runs in time $n^{O(1)}$.

    The case when $\varphi$ is unsatisfiable is an immediate corollary of the lower bound on $\Reff$ formulas (\Cref{thm:AM-original}): since $\Reff_s(\varphi)$ requires Resolution size $2^{\Omega(n)}$, then writing such a proof also requires exponential time and hence $\Kt(\pi \mid \neg\Reff_s(\varphi)) \geq \log2^{\Omega(n)} = \Omega(n)$ for every refutation $\pi$, giving $\info_{\Res}(\neg\Reff_s(\varphi)) = \Omega(n)$. (Note that if $\varphi$ does have a refutation of size $s$, then the information efficiency of $\neg \Reff_s(\varphi)$ is not defined, so we cannot claim a matching upper bound of $O(n)$.)
\end{proof}

\section{The Proof Analysis Problem for Extended Frege is $\NP$-complete}
\label{sec:PAP-EF-NP-complete}
In light of the extraction algorithm in \Cref{thm:det-algo}, it is natural to ask whether stronger proof systems are also analyzable. As shown in \Cref{prop:aut-anal-duality}, the existence of a polynomial-time proof analysis algorithm for a proof system $S$ implies that automating $S$ is~$\NP$-hard. Could this be the route towards proving that automating systems like Extended Frege is~$\NP$-hard?

This turns out to be unlikely. While for Resolution $\PAP_\Res[n^2] \in \ComplexityFont{AC}^0$ and $\FPAP_\Res[n^3] \in \FP$, already the decision problem for Extended Frege $\PAP_\EF$ turns out to be $\NP$-complete. This extends to every proof system $S$ that p-simulates Extended Frege: it is $\NP$-complete to decide whether a formula $\varphi$ is satisfiable given an $S$-proof of a Resolution lower bound on $\varphi$. We dedicate \Cref{subsec:NP-hard} to proving this. \Cref{subsec:conjecture} builds on this $\NP$-hardness result to investigate whether finding polynomial-time analysis algorithms (for weaker systems where they exist) requires proving proof complexity lower bounds first. 

\subsection{Hardness proof}
\label{subsec:NP-hard}
The idea of the hardness proof is to show that Extended Frege can prove certain Resolution lower bounds \say{agnostically}, in the sense that the proof employed will not depend on whether the underlying formula is satisfiable or not. This means that an algorithm that distinguishes between \say{true Resolution lower bounds} (those where the underlying formula is unsatisfiable) from \say{trivial ones} (those proven for satisfiable formulas) will be able to decide $\SAT$ and hence all of $\NP$. 

Our proof works by embedding every propositional formula into a suitable disjunction with the pigeonhole principle formulas. Conveniently, Haken's Resolution lower bound \cite{Haken85} for the pigeonhole principle was already formalized by Cook and Pitassi in $\PVO$.

\begin{theorem}[\citeauthor{CP90}, \citeyear{CP90} \cite{CP90}]
    \label{thm:CP90-PHP}
    There exist a positive $\varepsilon_0 \in \bbQ$ and $n_0 \in \bbN$ such that
    \[\PVO \vdash \forall n\forall\pi(\foReff(\PHP^n_{n-1}, \pi) \land n \geq n_0 \to ||\pi|| > \varepsilon_0 n).\]
\end{theorem}

Our reduction will use the following natural encoding of disjunctions as CNF formulas, with the aid of one additional variable.

\begin{definition}
    Let $\varphi = C_1 \land \cdots \land C_k$ and $\psi = D_1 \land \cdots \land D_m$ be two CNF formulas over the variables $x_1, \dots, x_n$. We denote by $\disj(\varphi, \psi)$ the formula
    \[  \disj(\varphi, \psi) \coloneq \bigwedge_{i=1}^k (\neg y \lor C_i) \land \bigwedge_{i=1}^m (y \lor D_i),\]
    where $y$ is a fresh variable not present in $\varphi$ or $\psi$.
\end{definition}

Intuitively, the variable $y$ in the definition above acts as a \say{switch}: if true, then $\varphi$ must hold; else, $\psi$ must hold. We can now prove that as long as one of the \say{disjuncts} is hard for Resolution, the entire formula must be hard as well.

\begin{theorem}
\label{thm:lwb-for-Disj}
    For every positive constant $c \in \bbN$, there exist $n_0 \in \bbN$ such that 
    \[ \SOT \vdash \forall \varphi \forall \pi \forall n \leq \varphi \forall m \leq \varphi \left(n \geq n_0 \land  \operatorname{3-CNF}(\varphi, n, m) \land |\pi| \leq n^{c} \to \neg \foReff( \disj(\varphi, \PHP_n) ), \pi)\right) .\]
\end{theorem}

\begin{proof}
    We work in $\SOT$. Suppose for contradiction that $\pi$ is indeed a Resolution refutation of $\disj(\varphi, \PHP_n)$. Consider the restriction $\rho$ mapping $y \mapsto 0$. We have that $\disj(\varphi, \PHP_n) \equiv \PHP_n$. Since $\SOT$ can prove that Resolution is closed under restrictions, it follows that $\pi_{\restriction_{\rho}}$ is a Resolution refutation of $\PHP_n$ in size $n^c$. However, by \Cref{thm:CP90-PHP}, necessarily $||\pi_{\restriction_{\rho}}|| > \varepsilon_0 n$, implying $||\pi || > \varepsilon_0 n$. This contradicts the assumption~$|\pi| \leq n^{c}$, when $n \geq n_0$ for $n_0$ chosen large enough.
\end{proof}

\begin{corollary}
\label{cor:EF-proofs}
    For every positive constant $c \in \bbN$, there exists a polynomial-time computable function $t$ such that for every 3-CNF formula $\varphi$ over a large enough number $n$ of variables, $t(\varphi)$ outputs an Extended Frege proof $\pi$ such that
    \[ \pi:  \EF \vdash \neg \Reff_{n^{c}}\left(   \disj(\varphi, \PHP_n)    \right).\]
\end{corollary}

\begin{proof}
    The formula in \Cref{thm:lwb-for-Disj} is $\forall \Pi^b_1$, so we can apply Cook's translation (\Cref{thm:translation}) to get polynomial-size $\EF$-proofs of $\neg \Reff_{n^{c}}\left(   \disj(\varphi, \PHP_n)    \right)$. Extended Frege can then polynomially and uniformly prove that the $\Reff$-like formula obtained from the translation is equisatisfiable to the $\Reff$ formulas as we have defined them in \Cref{def:Reff-formulas}. We then have that $\EF$ proves $\neg \Reff_{n^{c}}\left(   \disj(\varphi, \PHP_n)    \right)$ in polynomial size and the proofs can be produced uniformly in polynomial time.
\end{proof}

The previous upper bound in Extended Frege works for every $\varphi$, regardless of its satisfiability. This is the key idea behind the final reduction.

\begin{theorem}
\label{thm:PAP-EF-reduction}
For every positive constant $c \in \bbN$, the language $\TSAT$ reduces to $\PAP_{\EF}[n^c]$ under polynomial-time many-one Levin reductions.
\end{theorem}

\begin{proof}
    The reduction maps a 3-CNF formula $\varphi$ over $n$ variables to the instance $(\psi, \pi, 1^s)$, where $\psi$ is the formula $ \psi \coloneq \neg \Reff_{n^{c}}\left(   \disj(\varphi, \PHP_n)    \right)$, together with the Extended Frege proof $\pi$ given by the map $t(\varphi)$ in \Cref{cor:EF-proofs}, and the size parameter $1^s$ is $1^{n^{c}}$.

    By \Cref{cor:EF-proofs}, the proof $\pi$ is always a correct $\EF$-proof, regardless of the satisfiability of $\varphi$. Now, if $\varphi \in \TSAT$, then there exists a satisfying assignment for $\varphi$, which can be extended into a satisfying assignment for $\disj(\varphi, \PHP_n)$ setting the switching variable $y \mapsto 1$, and hence $\psi$ is satisfiable and $(\psi, \pi, 1^{n^{c}}) \in \PAP_\EF$. On the other hand, if $\varphi\not\in\TSAT$, then $\pi$ is still a valid $\EF$-proof but we have that $\disj(\varphi, \PHP_n)$ is unsatisfiable, since any satisfying assignment would need to satisfy either $\varphi$ or $\PHP_n$, which is impossible. Thus, we get~$(\psi,\pi,1^{n^{c}}) \not\in \PAP_{\EF}$.

    Note that this is also a Levin reduction: a satisfying assignment of $\varphi$ can be immediately extended into a satisfying assignment for $\psi$, and a satisfying assignment assignment for $\psi$ must necessarily include an assignment to $\varphi$.
\end{proof}

This yields the following formal restatement of \Cref{thm:np-hardness-informal}.

\begin{corollary}
    \label{cor:PAPS-hard-above-EF}
    For every propositional proof system $Q$ that p-simulates Extended Frege and every polynomial $s(n)$, the problem $\PAP_Q[s(n)]$ is $\NP$-complete under polynomial-time many-one Levin reductions.
\end{corollary}

\begin{proof}
    Membership in $\NP$ is trivial, since $\PAP_Q \in \NP$ for every Cook--Reckhow system $Q$. Hardness for $\PAP_\EF[s(n)]$ is given by \Cref{thm:PAP-EF-reduction}, and since $Q$ p-simulates $\EF$, this means that an instance $(\psi, \pi, 1^t)$ of $\PAP_\EF$ with $t \geq s(n)$ can be turned into an instance $(\psi, \pi', 1^{t})$, where $\pi'$ is the $Q$-proof obtained by simulating $\pi$.
\end{proof}

It should be clear that there is nothing essential about Extended Frege in the corollary above, except for the fact that it is the weakest proof system where we know how to formalize a Resolution lower bound. The proof generalizes as follows. Below, by saying that a proof system $S$ is \emph{constructively closed under \emph{modus ponens}}, we mean there exists a constant $p > 0$ such that given an $S$-proof $\pi$ of $\varphi$ and an $S$-proof $\tau$ of $\varphi \to \psi$, we can obtain a proof of $\psi$ in time $(|\pi| +|\tau|)^p$. Similarly, a proof system $S$ \emph{proves that Resolution is closed under restrictions} if there exists a constant $r>0$ 
such that given a formula $\varphi$ and a restriction $\rho$, we can obtain an $S$-proof of $\Reff_s(\varphi) \to \Reff_s(\varphi_{\restriction \rho})$ in time $(|\varphi| + s)^r$.

\begin{corollary}
    \label{cor:general-NP-hardness}
    Let $S$ be a propositional proof system constructively closed under \emph{modus ponens} and proving that Resolution is closed under restrictions and suppose there exists a constant $c > 0$ and a $\P$-uniform sequence of unsatisfiable formulas $\{ \psi_n\}_{n\in \bbN}$ such that $\psi_n$ has $n$ variables and $\poly(n)$ clauses and there exists a $\P$-uniform sequence of $S$-proofs $\{ \pi_n\}_{n\in\bbN}$ such that $\pi_n :S \vdash \neg\Reff_{n^c}(\psi_n)$.  Then, the problem $\PAP_S[n^c]$ is $\NP$-complete.
\end{corollary}

\begin{proof}[Proof sketch]
    The proof follows the argument in \Cref{thm:PAP-EF-reduction}: the $\NP$-hardness reduction starts from a formula $\varphi$ and constructs an $S$-proof of the formula $\neg \Reff_{n^c}(\disj(\varphi, \psi_n))$, which becomes an instance of $\PAP_S[n^c]$.
\end{proof}

\subsection{Does analyzability require lower bounds?}
\label{subsec:conjecture}
Despite the $\NP$-completeness results above, it remains open whether proof systems below Extended Frege are analyzable. For those proof systems that are in fact analyzable, it is natural to ask whether their analysis algorithms do inevitably rely on proof complexity lower bounds for the systems in question. After all, the only example of an analysis algorithm we have is the one for Resolution, and it heavily relies on the lower bound proof of~\cite{AM20}. It then seems reasonable to conjecture the following.

\begin{conjecture}[Analyzability requires lower bounds]
    \label{conj:analyzability-requires-lbs}
    For every propositional proof system $Q$, if $Q$ is analyzable, then $Q$ is not optimal.
\end{conjecture}

The requirement that $Q$ is not optimal is a natural formalization of lower bounds in this context: if $Q$ is not optimal, then there exists an explicit family of tautologies that are hard for $Q$, but easy for some other system. Since the lower bounds for $\Reff$-like formulas arising in the context of $\PAP_\Res$ are precisely Resolution lower bounds not provable in Resolution but easy for other systems, non-optimality captures the kind of lower bounds we expect are necessary.

We could also relax the conclusion of $Q$ not being optimal by simply $Q$ not being \emph{p-optimal}, which means that there exists an explicit family of tautologies whose $Q$-proofs are either long, or short but hard to find. Note, in particular, that $Q$ could fail to be p-optimal and still be polynomially bounded.

\begin{conjecture}[Analyzability requires lower bounds or proofs that are hard to find]
    \label{conj:original}
    For every propositional proof system $Q$, if $Q$ is analyzable, then $Q$ is not p-optimal.
\end{conjecture}

Clearly, \Cref{conj:analyzability-requires-lbs} implies \Cref{conj:original}. Building on the proof of $\NP$-hardness of $\PAP_\EF$ we can prove that both conjectures are likely true ---but very hard to prove.

\begin{proposition}
The following hold:
\begin{enumerate}[label=(\roman*)] \itemsep=0pt
    \item \Cref{conj:analyzability-requires-lbs} is true if, and only if, $\NP \neq \coNP$;
    \item \Cref{conj:original} is true if, and only if, $\P \neq \NP$.
\end{enumerate}
\end{proposition}

\begin{proof}
We start by proving (i). For the forward direction, we assume
     $\NP = \coNP$ and show that the conjecture fails. In this case there exists a polynomially bounded and hence optimal proof system $Q$. We define the following proof system $Q^\star$ based on $Q$. A proof $\pi$ in this system is of one of the following forms:
    \begin{enumerate}[label=(\alph*)] \itemsep=0pt
        \item if $\pi = \langle \varphi, \tau_1, \neg\Reff_s(\varphi), \tau_2 \rangle$, where $\varphi$ is a CNF formula, $\tau_1$ is a correct~$Q$-proof of $\neg\varphi$, and $\tau_2$ is a correct~$Q$-proof of $\neg\Reff_s(\varphi)$ for some $s\in\bbN$, then~$Q^\star$ outputs~$\neg \Reff_s(\varphi)$;
        \item if $\pi = \langle \varphi, \alpha, \neg\Reff_s(\varphi), \tau \rangle$, where~$\varphi$ is a CNF formula,~$\varphi(\alpha)=1$, and~$\tau$ is a correct~$Q$-proof of~$\neg\Reff_s(\varphi)$  for some $s\in\bbN$, then~$Q^\star$ outputs~$\neg \Reff_s(\varphi)$;
        \item if $\pi = \langle \varphi, \tau\rangle$ and $\varphi$ is not of the form $\neg\Reff_s(\psi)$ for any $s$ and~$\psi$, and~$\tau$ is a correct~$Q$-proof of~$\varphi$, then~$Q^\star$ outputs $\varphi$;
        \item in any other case, $Q^\star$ outputs a trivial tautology $p \lor \neg p$.
    \end{enumerate}
    We claim that $Q^\star$ is a Cook--Reckhow system that is both analyzable and optimal. First, it is sound, because~$Q$ is and every $Q^\star$-proof relies on correct $Q$-proofs; second, it is complete, because if $\varphi$ is not a $\neg\Reff$ formula, then one can always prove $\varphi$ via case (c) above using the completeness of $Q$, and if $\varphi = \neg\Reff_s(\psi)$ for some~$\psi$, then since $Q$ is complete, there is always a proof of $\neg \Reff_s(\psi)$ together with either a $Q$-proof of $\neg \psi$ or a satisfying assignment for $\psi$. Finally, $Q^\star$ is clearly polynomial-time computable because $Q$ is.

    It remains to argue that $Q^\star$ is both analyzable and optimal. Indeed, $\PAP_{Q^\star} \in \P$: by definition, if $\pi$ is a $Q^\star$-proof of $\neg \Reff_s(\varphi)$ for some $s$, then $\pi$ is of the form (a) or (b) above. If $\pi$ is of the form (a), then $\pi$ includes a proof of the unsatisfiability of~$\varphi$, so the analyzer immediately rejects after checking the correctness of this proof. If $\pi$ is of the form (b), then the analyzer can simply check that $\alpha$ is a correct satisfying assignment for $\varphi$ and conclude that $\varphi$ is satisfiable. At the same time, it is easy to see that $Q^\star$ is polynomially bounded, and hence optimal. Indeed, all non-$\neg\Reff$ formulas have polynomial-size proofs by case (c) as $Q$ is polynomially bounded. Similarly, in cases (a) and (b), there are always polynomial-size $\tau_1$ and $\tau_2$ to choose as $Q$ is polynomially bounded, so the entire $Q^\star$ is polynomially bounded. This means that the conjecture fails for $Q^\star$.
    
    Now, for the backwards direction, if the conjecture fails, we show that $\NP = \coNP$. For the conjecture to fail there must exist a proof system $S^\star$ that is analyzable yet optimal. Suppose $S^\star$ is analyzable by virtue of $\PAP_{S^\star}[n^c] \in \P$ for some $c > 0$. We then construct a non-deterministic polynomial-time procedure for the $\coNP$-complete set $\TUNSAT$. On an input 3-CNF formula $\varphi$ with $n$ variables and $m$ clauses, non-deterministically guess an $S^{\star}$-proof of the formula $\neg \Reff_{n^{c}}( \disj(\varphi, \PHP_n) )$. By \Cref{cor:EF-proofs}, polynomial-size such proofs exist in Extended Frege. Since $S^\star$ is optimal we have $S^\star \geq \EF$ and some polynomial-size proof is available in $S^\star$ too. Now, since $S^\star$ is analyzable, we can decide deterministically in polynomial-time whether $\varphi$ is satisfiable or not by analyzing this proof.
    
    We now prove (ii).
    If $\P = \NP$, then there exists a polynomially bounded and p-optimal proof system $Q^\star$. Since $\P = \NP$ the system $Q^\star$ is trivially analyzable, but it is p-optimal, so the conjecture fails for $Q^\star$.
    
    For the other direction, suppose $\P \neq \NP$ but \Cref{conj:original} fails. That means there exists a system that is analyzable yet p-optimal. Since $S$ is p-optimal, it follows that $S$ p-simulates $\EF$, and by \Cref{cor:PAPS-hard-above-EF} we have that $\PAP_S[n^c]$ is $\NP$-complete for every $c \in \bbN$. Since $\P \neq \NP$ we get that $\PAP_S[n^c] \not\in \P$ for any $c > 0$, a contradiction with the fact that $S$ is analyzable. \qedhere
\end{proof}

We finish by proving that under a standard assumption in circuit complexity, the analyzability of a propositional proof system is in fact equivalent to lower bounds on \emph{all} the non-trivial $\Reff$ formulas. The following statement makes this precise, and is the formal restatement of \Cref{thm:equivalence-PAP-AM-informal}. (The statement assumes some mild conditions on the behavior of propositional proof systems; see \Cref{prop:mild-properties} for formal definitions of these, and the paragraph before \Cref{cor:general-NP-hardness} for the definition of a system \emph{proving that Resolution is closed under restrictions}.)

\begin{theorem}
    \label{thm:equivalence-PAP-AM}
    Let $S$ be a propositional proof system that simulates Resolution, is constructively closed under modus ponens and contrapositions, and proves that Resolution is closed under restrictions. Assuming there exists a constant $\delta >0$ such that $\TSAT \not\in \ComplexityFont{io}\SIZE[2^{\delta n}]$, the following are equivalent:
    \begin{enumerate}[label=(\roman*)] \itemsep=0pt
        \item \label{item:PAP}  the proof system $S$ is analyzable;
        \item \label{item:AM} there are constants $c > 0$, $n_0 \geq 1$ and $\varepsilon > 0$ such that for every unsatisfiable 3-CNF formula $\varphi$ over $n \geq n_0$ variables, every $S$-proof of the formula $\neg \Reff_{n^{c}}(\varphi)$ requires size at least $2^{\varepsilon n}$.
    \end{enumerate}
\end{theorem}

\begin{proof}
    It is easy to see that if \cref{item:AM} holds for constant $c$, then $\PAP_S[n^c] \in \P$ (this is precisely the argument behind \Cref{prop:PAP_}, and does not require the circuit lower bound assumption), so we focus on the implication from \cref{item:PAP} to \cref{item:AM}. By contradiction: suppose $S$ is analyzable and hence $\PAP_S[n^c] \in \P$ for some constant $c > 0$, but \cref{item:AM} fails: then, in particular, for every $\varepsilon > 0$ there must exist a (possibly non-uniform) sequence of unsatisfiable 3-CNF formulas $\{ \psi_n\}_{n \in \bbN}$ such that for infinitely many $n$, the system $S$ can prove $\neg\Reff_{n^c}(\psi_n)$ in size at most $2^{\varepsilon n}$.

    The goal is now to construct a (possibly non-uniform) circuit family that solves $\TSAT$ in time $O(2^{\delta n})$ infinitely often by using the $\NP$-hardness reduction in \Cref{cor:PAPS-hard-above-EF} as follows. Suppose that $d$ is the constant such that $\PAP_S[n^c] \in \P$ by virtue of a family of Boolean circuits of size $n^d$, and let $r$ and $p$ be, respectively, the constants in the blow-up for proving that Resolution is closed under restrictions and the closure under under contrapositions. Let $\varepsilon \coloneq \delta/dpr$ and let $\{ \psi_n\}_{n \in \bbN}$ be the sequence of formulas as in the paragraph above for this particular choice of $\varepsilon$. Now, given a 3-CNF formula $\varphi$ over $n$ variables, suppose $n$ is one of the infinitely many values for which there is an $S$-proof $\pi_0$ of $\neg\Reff_{n^c}(\psi_n)$ in size $2^{\varepsilon n}$. From the fact that such a proof exists and $S$ proves that Resolution is closed under restrictions, it follows that there is also an $S$-proof $\pi$ of $\neg\Reff_{n^c}(\disj(\varphi, \psi_n))$ in size $2^{\varepsilon prn}$: apply first the fact that Resolution is closed under restrictions to get a proof of $\Reff_{n^c}(\disj(\varphi, \psi_n)) \to \Reff_{n^c}(\psi_n)$; then, contraposition with $\neg\Reff_{n^c}(\psi_n)$ yields the desired formula and the size of the proof is at most $|\pi_0|^{pr}$. Now, since $\PAP_S[n^c] \in \P$ and there exists a Boolean circuit of size $n^d$ solving this problem, we can use this circuit on the input $(\disj(\varphi, \psi_n), \pi, 1^{n^c})$, which we construct using $\pi_0$ as advice. The result of this computation tells us whether $\varphi$ is satisfiable, and the overall size of the circuit is $2^{\varepsilon dprn}= 2^{\delta n}$. This means $\TSAT \in \ComplexityFont{io}\SIZE[2^{\delta n}]$, contradicting our assumption.
\end{proof}

\section{The Atserias--Müller lower bound in $\PVO$}
\label{sec:AM20-in-PV}

In this section we undertake the task of formalizing the Atserias--Müller lower bound in bounded arithmetic to prove \Cref{thm:AM-in-PV-informal}. There are at least two possible approaches here, at least at a conceptual level. The first one is to formalize in $\PVO$ the correctness of the extraction algorithm from \Cref{sec:the-algorithm}, and use this to derive the lower bound along the lines of the argument in \Cref{thm:det-algo}, carefully verifying that we never exceed the reasoning power of the theory.

The alternative approach, which we choose to take, is to forget about the extraction algorithm at first glance and attempt a direct formalization of the theorem seeing it purely as a proof complexity lower bound. Then, from the form of the statement, witnessing theorems will recover the algorithm together with its proof of correctness. We opt for this second route, since it is in fact the original hint for why the algorithm should exist in the first place, and because restating the proof in this format can be insightful in itself.

Of course, the proof we choose to formalize deviates from the original one in a few crucial points, although the overall structure is preserved. The main difference is the deterministic restriction replacing the random restriction argument in \cite{AM20}. This deterministic procedure corresponds exactly to the greedy algorithm behind the width-reduction technique in \Cref{lemma:det-wr}.

We start by stating the lower bound in the language of $\PVO$, which posits that a correct Resolution refutation $\pi$ of $\Reff_s(\varphi)$ must have size at least $2^{\varepsilon s/n^2}$ whenever $\varphi \not\in \SAT$, for some fixed positive $\varepsilon$. Since we do not have exponentiation in $\PVO$, we scale down the bound and state is as $||\pi|| > \varepsilon s/n^2$, and the first-order formula becomes
\begin{equation}
\label{eq:foAM} \tag{$\foAM_{\varepsilon, n_0}$}
\begin{split}
    \foAM_{\varepsilon, n_0} \coloneq \forall \varphi \forall \pi \forall n \forall s \biggl( \Bigl( n \geq n_0 \land \operatorname{CNF}(\varphi, n) &\land \foReff(\Reff_s(\varphi),\pi) \\ &\land \forall \alpha \leq \varphi  \bigl(\neg \foSatf(\varphi, \alpha) \bigr)\Bigr) \to ||\pi|| > \varepsilon s/n^2 \biggr)\,.
\end{split}
\end{equation} 

Here the different predicates and function symbols all have the obvious intended meaning and it is readily verified that they are all computable in polynomial time and there exist function symbols for them in the language of $\PV$.

The formula itself is $\forall\Sigma^b_1$. This becomes clear when rewriting it in prenex form, which gives us the statement
\begin{equation}
    \forall \varphi \forall n \forall s \forall \pi \exists \alpha \leq \varphi \Big(n < n_0 \lor \neg \operatorname{CNF}(\varphi, n) \lor \neg\foReff(\Reff_s(\varphi), \pi)  \lor \foSatf(\varphi, \alpha) \lor||\pi|| > \varepsilon s/n^2 \Big)\,.
\end{equation}

If we were to apply Buss's witnessing theorem here, we would get a polynomial-time function witnessing the existential quantifier on $\alpha$, which corresponds precisely to the extraction algorithm.

We carry out the formalization in three subsections. \Cref{subsec:restriction-in-PV} proves the restriction argument, \Cref{subsec:block-width-lb-in-PV} carries out the block-width lower bound and \Cref{subsec:putting-it-together-PV} puts these together. We also recall that $\SOT(\PV)$, which we denote simply as $\SOT$, is $\forall\Sigma^b_1$-conservative over $\PVO$, so we will often carry out arguments in $\SOT$ instead of $\PVO$ and use the induction principles available there without further comment.

\subsection{The restriction argument}
\label{subsec:restriction-in-PV}

We want to prove the following $\forall \Sigma^b_1$ formula, stating the restriction argument:
\begin{align} \tag{$\operatorname{AM-Restriction}_c$}
    \operatorname{AM-Restriction}_c \coloneq \forall \varphi \forall n \forall s \forall \pi \biggl(\operatorname{CNF}(\varphi, n) \land \foReff(&\Reff_s(\varphi), \pi) \\ \notag \to \exists \rho \leq \pi \Bigl(& \bw(\pi_{\restriction\rho}) \leq c \cdot \left\lceil{\sqrt{s \log |\pi|}}\right\rceil
    \\ \notag &\land \Reff_s(\varphi)_{\restriction \rho}\neq \bot 
    \\& \notag \land  \operatorname{Disabling}(\rho, \varphi, n, s) \leq c/2 \cdot \left\lceil{\sqrt{s \log |\pi|}}\right\rceil \bigg)\,.
\end{align}

The formula $\operatorname{AM-Restriction}_c$ states that for every CNF formula $\varphi$ over $n$ variables and every correct Resolution refutation $\pi$ of the $\Reff_s(\varphi)$ formula, there exists a $d$-disabling restriction $\rho \in \{0, 1, *\}^N$ to the $N$ variables of $\Reff_s(\varphi)$ such that $\bw(\pi_{\restriction \rho}) \leq c\cdot \Big\lceil{\sqrt{s \log |\pi|}}\Big\rceil$ and $d \leq c/2 \cdot \Bigl\lceil{\sqrt{s \log |\pi|}}\Bigl\rceil$.

Throughout this section we freely use rationals and reals and notation like $\sqrt{\cdot}$, always assuming that some suitable rational approximation is used under the hood. This is standard for formalizations in these theories (see, for example, the standard style of the formalizations by \citeauthor{jerabek:phd-thesis} \cite{jerabek:phd-thesis}).

It is easy to see that there exists a $\PV$ function $W$ that given $\pi$ outputs the set of clauses in $\pi$ of block-width strictly larger than $c\cdot \Bigl\lceil{\sqrt{s \log |\pi|}}\Bigr\rceil$. Similarly, there is a $\PV$ function $W_{\restriction\rho}$ that given $\pi$ and $\rho$ defines the set of clauses in $\pi_{\restriction\rho}$ of block-width strictly larger than $c \cdot  {\Big\lceil{\sqrt{s \log |\pi|}}\Big\rceil}$, and the basic properties of these functions can be proven in $\SOT$.

For a given restriction $\rho$ that sets values of variables in $\ell$ different blocks, we will denote by $\tau \coloneq ((b_1, X_1), \dots, (b_\ell, X_\ell))$ an ordering $b_1, \dots, b_\ell \in [s]$ of the blocks that are mentioned by $\rho$, and, for each of them, an ordering $X_i$ of all the variables in block $b_i$. We call $\tau$ a \emph{trace} for $\rho$ and for every $k \in [\ell]$ and $S$ a set of variables of $\Reff_s(\varphi)$, we denote by $\rho_{k, S}$ the subrestriction of $\rho$ defined as
    \begin{equation}
        \rho_{k, S}(x) \coloneq \begin{cases}
            \rho(x) &\text{ if }x\text{ belongs in one of the blocks $b_1, \dots, b_k \in \tau$} \text{ or if }x\in S  \\
            * &\text{otherwise}
        \end{cases}
    \end{equation}
    and $\rho_{0, S} \coloneq *^N$. We will denote by $\rho_k$ the subrestriction $\rho_{k, \emptyset}$.

    \begin{definition}[Most-killing property]
    \label{def:most-killing}
    We will say that a restriction $\rho$ enjoys the \emph{most-killing property} with respect to a trace $\tau = ((b_1, X_1), \dots, (b_\ell, X_\ell))$ if it holds that for every $k \in[\ell]$, the block $i \in [s] \setminus \{ b_1, \dots, b_{k-1}\}$ that is the most frequent block mentioned in $W_{\restriction \rho_{k-1}}$ happens to be precisely $b_k$ and $\rho_k$ satisfies the following conditions.
    \begin{enumerate}
    \itemsep=0pt
        \item If $e_i$ appears positively in at least $1/3$ of all the clauses in $W_{\restriction \rho_{k-1}}$ that mention block $i$, then $\rho_k(e_i) = 1$ and $\rho_k(x) = *$ for every other variable $x \neq e_i$ in block $i$.
        
        \item If $e_i$ does \emph{not} appear positively in at least $1/3$ of all the clauses in $W_{\restriction \rho_{k-1}}$ that mention block $i$, then $\rho_k(e_i) = 0$ and for every other variable in block $i$, $\rho_k$ assigns values to them respecting the ordering in $X_i = ( e_i, x_{i,1}, x_{i, 2}, \dots)$ and with the following priority. For $j = 1, \dots, |X_i|-1$,
        \begin{enumerate}
        \itemsep=0pt
            \item if $x_{i, j}$ appears positively more often than negatively in $W_{\restriction \rho_{k-1, \{e_i, x_{i, 1}, \dots, x_{i, j-1}\}}}$, then $\rho_k(x_{i, j}) = 1$; 
            \item if $x_{i, j}$ appears negatively more often than positively in $W_{\restriction \rho_{k-1, \{e_i, x_{i, 1}, \dots, x_{i, j-1}\}}}$, then $\rho_k(x_{i, j}) = 0$.
        \end{enumerate}
    \end{enumerate}
    \end{definition}

    Note that verifying that $\rho$ enjoys the most-killing property with respect to $\tau$ is possible in polynomial-time and hence there is a $\PV$ function that performs the check and $\SOT$ can prove this.

    We now prove the following claim by induction.

    \begin{lemma}
    \label{lemma:induction-restriction-PV}
    Let $c \in \bbN$. The following is provable in $\SOT$. For every CNF formula $\varphi$  over $n$ variables and $\pi$ a correct Resolution refutation of $\Reff_s(\varphi)$, let $w \coloneq c\cdot \left\lceil \sqrt{s \log |\pi|} \right\rceil$ and let $W$ denote the set of clauses in $\pi$ of block-width larger than $w$. Then, for every $\ell$, the formula $\Psi(\varphi, \pi, s, \ell)$ defined as follows holds: if $\ell \leq w/2$, then there exists a restriction $\rho \in \{ 0,1,*\}^N$, where $N$ is the number of variables of the $\Reff_s(\varphi)$ formula, a trace $\tau = ((b_1, X_1), \dots, (b_\ell, X_\ell))$ and $d \leq \ell$ such that:
    \begin{enumerate}[label=(\roman*)]
    \itemsep=0pt
        \item $\rho$ is $d$-disabling;
        \item $\rho$ has the most-killing property according to $\tau$;
        \item $(3s)^\ell \cdot |W_{\restriction \rho}| \leq |W| \cdot (3s- w + \ell)^\ell$.
    \end{enumerate}
    \end{lemma}
    
    \begin{proof}
        Observe that the formula $\Psi$ is $\Sigma^b_1$. First, the existential quantifiers on $\rho$, $\tau$ and $d$ are bounded by $\pi$ and $\varphi$. The properties (i) and (ii) are checkable in polynomial time and hence there are $\PV$ relations for them and $\SOT$ can prove their basic properties. To show that item (iii) can be properly expressed by a $\PV$ formula,  note that because $\pi$ is a correct Resolution refutation of $\Reff_s(\varphi)$, it must hold that $s \leq |\pi|$ and 
        \begin{equation}
            w = c \left\lceil \sqrt{s\log |\pi|} \right\rceil \leq c\left\lceil \sqrt{s ||\pi||} \right\rceil \leq c  \left\lceil \sqrt{ |\pi| ||\pi||}\right\rceil \leq c|\pi|\,,
        \end{equation}
        where the last equality holds because $||\pi|| \leq |\pi|$. As a consequence,  $w \in \Log$, which means that for every $\ell \leq w/2$, $\ell \in \Log$ and thus $s^\ell$ exists because $s^\ell = 2^{|s|\ell} = s \# 2^\ell$.
        
        We show that $\SOT \vdash \forall \ell \Psi(\varphi, \pi, s, \ell)$. We proceed by induction on the parameter $\ell$, the length of the trace~$\tau$. Note that this corresponds to Length Induction over $\Sigma^b_1$ formulas and is hence available in $\SOT$.

        If $\ell = 0$, then we can take the empty restriction $\rho \coloneq *^N$, which disables $d = 0$ variables and the trace $\tau$ to be the empty sequence, which trivially satisfy the conditions.

        We assume now that $\Psi(\varphi, \pi, s, \ell)$ holds and we prove $\Psi(\varphi, \pi, s, \ell+1)$. If $\ell + 1 > w/2$, then we are done, so assume $\ell +1\leq w/2$. By induction hypothesis, there exist $\rho$, $d$ and $\tau$ satisfying the desired properties for $\ell$. We extend $\rho$ into $\rho'$ by following the conditions of the most-killing property. This amounts to applying one iteration of Algorithm~\ref{algorithm:greedy}. There is a $\PV$ function for this, as well as for extracting the most frequent block in the refutation and $\SOT$ can prove the basic properties of these functions. The trace is then extended by the most frequent block $b_{\ell + 1}$ mentioned in $W_{\restriction \rho}$ and the ordering of the variables can be any fixed ordering of the variables used in the construction of $\rho'$ by Algorithm~\ref{algorithm:greedy}.

        If $\rho$ was $d$-disabling for $d \leq \ell$, then we have that $\rho'$ is at most $(d+1)$-disabling, since one iteration of Algorithm~\ref{algorithm:greedy} disables at most one additional block, and we have $d+1\leq \ell + 1$. By construction, it is easy to see that $\rho'$ has the most-killing property with respect to $\tau'$. It is only left to verify that
        \begin{equation}
        (3s)^{\ell+1} \cdot |W_{\restriction \rho'}| \leq |W| \cdot (3s-w + \ell+1)^{\ell+1}\,.
        \end{equation}
        
        The set $W$ contains all the clauses in $\pi$ of block-width at least $w$. The restriction $\rho'$ sets the variables of $\ell + 1$ blocks, meaning that every clause in $W$ not trivialized by $\rho'$ has block-width at least $w - (\ell + 1)$, and these clauses are precisely the ones in  $W_{\restriction\rho'}$ by definition.

        Furthermore, the restriction $\rho'$ disables at most $\ell + 1$ blocks, so an averaging argument implies that there is a non-restricted block mentioned in at least $|W_{\restriction\rho}|(w-(\ell + 1))/(s-(\ell + 1))$ clauses. In particular, the most frequent non-restricted block must be mentioned at least that often. Since $|W|\leq |\pi|$ and $s \leq |\pi|$, we have that $|W|$ and $s$ are in $\Log$ and therefore we can use exact counting in $\SOT$ to carry out this averaging argument (see \Cref{subsec:prelim-exact-counting} for details).
        
        By the way we constructed $\rho'$ following Algorithm~\ref{algorithm:greedy} we are guaranteed to kill at least $1/3$ of all the clauses mentioning the most frequent unrestricted block $b_{\ell + 1}$. Indeed, by inspecting the conditions of the most-killing property (\Cref{def:most-killing}) if we enable block $b_{\ell+1}$ that is because it appeared in at least $1/3$ of all the clauses mentioning $b_{\ell +1 }$; and otherwise we are guaranteed to restrict at least $1/2$ of the remaining $2/3$ fraction of the clauses mentioning $b_{\ell + 1}$, which amounts to a $1/3$ fraction.

        Together, this means that
        \begin{equation}
            |W_{\restriction \rho'}| \leq |W_{\restriction \rho}| \cdot \left(1-\frac{w-(\ell + 1)}{3s}\right)\,, 
        \end{equation}
        which is the same as
        \begin{equation}
            3s \cdot |W_{\restriction \rho'}| \leq |W_{\restriction\rho}| \cdot (3s-w + \ell+1)\,.
        \end{equation}

        As $\rho'$ extends the restriction $\rho$, it is clear that $|W_{\restriction \rho'}| \leq |W_{\restriction \rho}|$, and by induction hypothesis we know that
        \begin{equation}
            (3s)^\ell \cdot |W_{\restriction \rho}| \leq |W| \cdot (3s-w + \ell)^\ell\,,
        \end{equation}
    meaning that
        \begin{align}
            3s \cdot |W_{\restriction\rho'}| &\leq |W_{\restriction\rho}| \cdot (3s-w + \ell+1) \\
            &\leq \frac{1}{(3s)^\ell}|W| \cdot (3s-w + \ell)^\ell \cdot (3s-w + \ell+1) \\
            &\leq \frac{1}{(3s)^\ell}|W|  \cdot (3s-w + \ell+1)^{\ell + 1}\,,
        \end{align}
        from which the desired $(3s)^{\ell+1} \cdot |W_{\restriction \rho'}| \leq |W| \cdot (3s-w + \ell+1)^{\ell+1}$ follows.
    \end{proof}
    
    Before we complete the proof, we have the following technical lemma that will help us deal with our bounds inside the theory.
    
    \begin{lemma}
        \label{lemma:exp-bound-PV}
        In $\PVO$, for all $n,m,p \in \Log$, if $n < m$ and $p \geq 1$, it holds that
        \[\biggl( 1 - \frac{n}{m}\biggr)^p \leq 2^{-\frac{n}{m}p}. \]
    \end{lemma}

\begin{proof}
    Since $p \in \Log$, we take $n$ and $m$ to be fixed and proceed by Length Induction on $p$. The base case $p = 1$ is covered by a proposition of     \citeauthor{MP20} \cite[Proposition 2.5]{MP20}, who formalized in $\PVO$ that for all $n, m \in \Log$ such that $ n < m $, it holds that $1- n/m  \leq 2^{-n/m}$. If the bound holds for $p$, then for $p +1$ we have
    \begin{equation}
    \biggl( 1 - \frac{n}{m}\biggr)^{p+1} = \biggl( 1 - \frac{n}{m}\biggr)^{p} \cdot\biggl( 1 - \frac{n}{m}\biggr) \leq 2^{-\frac{n}{m}p} \cdot 2^{-\frac{n}{m}} = 2^{-\frac{n}{m}(p+1)}\,.
    \end{equation}
    The first equality holds because $(1 - n/m) \in \Log$ and $\PVO$ can prove the basic properties of powers for $\Log$-sized objects; the next inequality holds by induction hypothesis and the same proposition by \citeauthor{MP20} as in the base case. This completes the proof.    
\end{proof}

    We are ready to show that $\PVO$ proves the desired restriction argument.

\begin{lemma}[Formalized deterministic width-reduction in $\PVO$]
\label{lemma:restriction-in-PV}
    It holds that $\PVO \vdash \operatorname{AM-Restriction}_4$.
\end{lemma}

\begin{proof}
    Let $\varphi$ be a CNF formula over $n$ variables and $\pi$ a correct Resolution refutation of $\Reff_s(\varphi)$ for some~$s$. Apply \Cref{lemma:induction-restriction-PV} in $\SOT$ for $\ell \coloneq w/2$, where $w \coloneq 4 \cdot \left\lceil{\sqrt{s \log |\pi|}}\right\rceil$. This gives us a restriction $\rho$ that is $d$-disabling for $d \leq \ell$ and $(3s)^\ell \cdot |W_{\restriction \rho}| \leq |W| \cdot (3s- w + \ell)^\ell$. We want to show that $\bw(\pi_{\restriction\rho}) \leq w$, that $\Reff_s(\varphi)_{\restriction \rho}\neq \bot $  and that $\rho$ is $d$-disabling, for $d \leq w/2$.

     That $\rho$ is $d$-disabling is guaranteed by the lemma, and because it is $d$-disabling, the restriction only sets the values of enabling variables and, whenever a block is enabled, it does not restrict any other variables in the block. As a consequence, no axiom of $\Reff_s(\varphi)$ is ever falsified by $\rho$. 
     
     It is only left to verify that the block-width of $\pi_{\restriction \rho}$ is at most $w$, as desired. \Cref{lemma:induction-restriction-PV} guarantees that  $|W_{\restriction \rho}| \leq |W| \cdot \left(1 - (w - \ell)/{3s} \right)^\ell$, and substituting our choice of $\ell$ we have that
     \begin{align}
         |W_{\restriction \rho}| \leq |W| \cdot \left(1 - \frac{w/2}{3s}\right)^{w/2} \leq |W| \cdot 2^{-\frac{w^2}{12s}}\,,
     \end{align}
     where the last inequality is provable in $\SOT$, as shown in \Cref{lemma:exp-bound-PV}.

     We want $|W| \cdot 2^{-{w^2}/{12s}} < 1$, so
     \begin{align}
         |W| \cdot 2^{-\frac{w^2}{12s}} < 1 &\Leftrightarrow \log|W| < \frac{w^2}{12s} \\
         &\Leftrightarrow \sqrt{12s\log|W|} < w\,,
     \end{align}
     and this last inequality is verified for our choice of $w$.

    It follows that $W_{\restriction\rho} = \emptyset$ and hence $\bw(\pi_{\restriction\rho}) \leq w$, as desired. Since $\SOT$ is $\forall\Sigma^b_1$-conservative over $\PVO$, we have the sentence in $\PVO$. 
\end{proof}

\subsection{The block-width lower bound}
\label{subsec:block-width-lb-in-PV}
The crucial part of the formalization, and the place where we most clearly see how satisfying assignments can be extracted from Resolution refutations, is in the width lower bound. The restriction argument we just covered reduces the block-width of a refutation regardless of whether the underlying formula is satisfiable or not; it is the width lower bound that only holds when the formula is unsatisfiable.

Let us first state the formula we want to prove, which we encode in the following $\forall\Sigma^b_1$ statement:
\begin{equation}
\tag{$\operatorname{AM-WLB}$}
    \begin{split}
        \operatorname{AM-WLB} \coloneq \forall \varphi \forall n \forall s \forall \pi \biggl(\operatorname{CNF}(\varphi, n) \land &\foReff(\Reff_s(\varphi), \pi)  \land \forall \alpha \leq \varphi. \neg \SAT(\varphi, \alpha)\phantom{\to} \\ &\to \exists C \leq \pi \Bigl( C \in \pi \land \bw(C) \geq 1/3 (\lfloor{s/n}\rfloor - 1 \Bigr) \biggr).
    \end{split}
\end{equation}

When writing the formula in prenex form, pulling out the universal quantifier on $\alpha$ and turning it into an existential one, it becomes clear how the statement is $\forall\Sigma^b_1$ and the extraction of a satisfying  assignment can be performed by witnessing the first existential quantifier:
\begin{equation}
    \begin{split}
    \operatorname{AM-WLB} \equiv \forall \varphi \forall n \forall s \forall \pi \exists \alpha \leq \varphi \exists C \leq \pi\biggl(&\operatorname{CNF}(\varphi, n) \land \foReff(\Reff_s(\varphi), \pi)    \\ &\to \foSatf(\varphi, \alpha) \lor \Bigl(C \in \pi \land \bw(C) \geq 1/3 (\lfloor{s/n}\rfloor - 1 \Bigr)\biggr).
    \end{split}
\end{equation}

The crucial fact that makes the block-width lower bound go through, and which fundamentally distinguishes the $\SAT$ and $\UNSAT$ cases is that, if $\varphi \in \UNSAT$, then every width-$n$ clause over the variables of $\varphi$ can be weakened from some clause of $\varphi$, while a width-$n$ clause that is \emph{not} the weakening of any axiom will necessarily encode a satisfying assignment. This was earlier stated as \Cref{fact:weakening}. We now reprove the statement in $\SOT$.

\begin{lemma}[\Cref{fact:weakening} in $\SOT$]
    \label{fact-in-PV}
    Let $\varphi$ be a Boolean formula in CNF over $n$ variables. If $C$ is a width-$n$ clause over the variables of $\varphi$ that is not the weakening of any clause of $\varphi$, then $\neg C$ encodes a satisfying assignment for $\varphi$.
\end{lemma}

\begin{proof}
    By $\Delta^b_1$-induction on the number $n$ of variables $x_1, \dots, x_n$ of $\varphi$, which is available in $\SOT$.

    If $n = 1$, then the only two width-$1$ clauses are $x_1$ and $\neg x_1$. The three possible CNF formulas are $x_1$, $\neg x_1$ and $x_1 \land \neg x_1$, and it is easy to check that the statement holds.

    Suppose the statement holds for formulas with up to $n$ variables, and let $\varphi$ be a formula over $n+1$ variables and let $C$ be a clause of width $n + 1$. Assume $x_{n+1} \in C$, and restrict $x_{n+1} \mapsto 0$. We obtain a new formula $\varphi_0$ and a clause $C_0$. For every clause $D$ in $\varphi$ we have the following cases:
    \begin{enumerate}\itemsep=0pt
        \item[(a)] the variable $x_{n+1}$ appeared in $D$ negatively, and hence $D$ is satisfied by setting $x_{n+1} \mapsto 0$;
        \item[(b)] the variable $x_{n+1}$ appeared in $D$ positively, or the variable did not appear at all; in both cases we have that this is a clause over variables $x_1, \dots, x_n$ and, by assumption, $C$ and hence also $C_0$ is not a weakening of it. By induction hypothesis, $\neg C_0$ and then also $\neg C$ encode a satisfying assignment to $D$.
        
    \end{enumerate}
    
    Hence, every clause in $\varphi$ is satisfied by $\neg C$.  If $\neg x_{n+1} \in D$ the the same argument goes through by restricting $x_{n+1} \mapsto 1$. This completes the proof.
\end{proof}

To prove the rest of the width lower bound, we show the following lemma first, which essentially corresponds to the invariant proved within \Cref{lemma:extraction} to derive the correctness of the assignment extraction algorithm. We restate the invariant here as a property of a path in the Resolution refutation.

\begin{definition}[Reservation invariant]
\label{def:reservation-invariant}
Let $\pi$ be a Resolution refutation of $\Reff_s(\varphi)$ for some CNF formula $\varphi(x_1, \dots, x_n)$, let $N$ be the number of variables of $\Reff_s(\varphi)$, $d \leq \operatorname{depth}(\pi)$, and let $\calC = (C_1, \dots, C_d)$ be a length-$d$ path in $\pi$ starting from $C_1 = \bot$. Let the $s$ blocks of $\Reff_s(\varphi)$ be arranged in a layered manner, so that there are $n$ layers, each containing $\lfloor s/n\rfloor$ blocks, with the remainder blocks left from the flooring operations collected all in the last layer, plus one additional layer on top with a single block corresponding to the root. We say that $\calC$ satisfies the \emph{reservation invariant} with respect to a sequence $A = (\alpha_1, \dots, \alpha_d)$ of restrictions $\alpha_i \in \{ 0,1,*\}^N$ if, for every $i \in [d]$,

\begin{enumerate}
\itemsep=0pt
    \item[(i)] the restriction $\alpha_i$ falsifies $C_i$;
    \item[(ii)] a block is only mentioned in $\alpha_i$ if it is either mentioned in $C_i$ or its parent according to $\alpha_i$ is mentioned in $C_i$, and, in particular, $\bw(\alpha_i)  \leq 3\bw(C_i)$;
    \item[(iii)] if $\alpha_i$ restricts some variable in block $B$ from layer $\ell$, then it does so in the following way: $B$ must contain exactly $\ell$ literals over the variables $x_1, \dots, x_\ell$ and no two literals for the same variable;
    \item[(iv)] the restriction $\alpha_i$ does not falsify any clause of $\Reff_s(\varphi)$;
    \item[(v)] the path determined by $\calC$ corresponds to $A$ in the following way: if $C_i$ was derived from $C_{i+1}$ by resolving over variable $v$, then $\alpha_{i+1}(v)$ is defined and $v$ appears in $C_{i+1}$ with polarity $1- \alpha_{i+1}(v)$.
\end{enumerate}
\end{definition}

The lemma now states that a path exists that satisfies the invariant or, alternatively, correctly leads to a wide clause or a satisfying assignment.

\begin{lemma}
    \label{lemma:induction-width-PV}
    The following is provable in $\SOT$. For every CNF formula $\varphi$ over $n$ variables and $\pi$ a correct Resolution refutation of $\Reff_s(\varphi)$, for every $d$, the formula $\Psi(\varphi, \pi, s, d)$ defined as follows holds: if $1\leq d \leq \operatorname{depth}(\pi)$, then there exists a sequence $\calC = (C_1, \dots, C_d)$ of clauses of $\pi$ and a sequence $A = (\alpha_1, \dots, \alpha_d)$ of restrictions such that at least one of the following conditions holds:
    \begin{enumerate}
    \itemsep=0pt
        \item[(a)] there is $C_i \in \calC$ that satisfies $\bw(C_i) \geq 1/3 (\lfloor{s/n}\rfloor - 1 )$;
        \item[(b)] there is $C_i$ in and its corresponding restriction $\alpha_i$ that contain a satisfying assignment to $\varphi$;
        \item[(c)] the sequence $\calC$ is a path in $\pi$ starting at $C_1 = \bot$ and the reservation invariant holds for $\calC$ with respect to $A$.
    \end{enumerate}
\end{lemma}

\begin{proof}
    The formula $\Psi$ is $\Sigma^b_1$ because the existential quantifiers on $\calC$ and $A$ are bounded by $\pi$ and items (a)-(c) are all checkable by $\PV$ functions and $\SOT$ proves their basic properties. We can then proceed by induction on the parameter $d$, which corresponds to Length Induction over a $\Sigma^b_1$ formula, available in $\SOT$.

    For the base case, $d = 1$ and we can pick $C_1 = \bot$ and $\alpha_1 = *^N$. It is immediate to verify that the reservation invariant is satisfied, so item (c) holds.

    Now, suppose the statement holds for $d$, and we prove it for $d + 1$. If $d+1 > \depth(\pi)$, then we are done. Otherwise, by induction hypothesis there exists a sequence $\calC = (C_1, \dots, C_d)$ and a series of restrictions $A=(\alpha_1, \dots, \alpha_d)$. If (a) or (b) hold for them, then we simply extend $\calC$ with $C_{d+1} \coloneq \bot$ and $\alpha_{d+1}\coloneq *^N$ and either (a) or (b) will still hold.

    If we are in case (c) but (a) and (b) failed, then $\alpha_d$ does not falsify any clause of $\Reff_s(\varphi)$ (as per item (iv) of the reservation invariant), and hence the clause was derived from some previous clause in $\pi$, either by weakening or by a resolution step. Here we simply move to one of the children in $\pi$ and extend $\alpha_d$ into $\alpha_{d+1}$ following the reservation strategy in Algorithm~\ref{algorithm:prover-delayer}. More precisely, we run one iteration of Algorithm~\ref{algorithm:prover-delayer} starting the traversal from $C_d$ and taking $\alpha_d$ as the restriction kept in memory by the algorithm (this corresponds to Steps 2-4 of Algorithm~\ref{algorithm:prover-delayer}). Crucially, there is a $\PV$ function that carries out this computation.

    Let us analyze the outcome of this procedure. If the iteration of Algorithm~\ref{algorithm:prover-delayer} succeeds in extending the reservation from $\alpha_d$ to $\alpha_{d+1}$, then by the way $\alpha_{d+1}$ is constructed from $\alpha_d$ and by the fact that the induction hypothesis guaranteed that $\alpha_d$ satisfies the invariant, we have that $\alpha_{d+1}$ will immediately satisfy the reservation invariant too. In this case, item (c) holds and we are done.
    
    We argue that the reservation process cannot fail. If it did, it must be that the process failed at Step 3b or 3c of Algorithm~\ref{algorithm:prover-delayer}. 
    \begin{itemize}
        \item If the process failed in Step 3b, the algorithm attempted the reservation of a block at layer $1 \leq i < n$, but there were no free blocks left. This means that $\alpha_d$ already reserved at least $\lfloor s/n \rfloor - 1$ blocks on that layer, and so $\bw(\alpha_d) \geq \lfloor s/n \rfloor - 1$, since each layer $i < n$ contains exactly $\lfloor s/n \rfloor $ blocks. By point (ii) of the reservation invariant we have that $\bw(\alpha_d) \leq 3\bw(C_d)$, so putting this together we have that $C_d$ has block-width at least $1/3 (\lfloor s/n\rfloor - 1)$, contradicting that we were not in case (a).

        \item If the process failed in Step 3c, this implies the reservation $\alpha_d$ had a clause $\calA$ encoded in a block at layer $n$, but it failed to find a clause of $\varphi$ that $\calA$ was a weakening of. By point (iii) of the invariant, since the block is at layer $n$, $\calA$ is a width-$n$ clause, and by \Cref{fact-in-PV} in $\SOT$, $\neg \calA$ encodes a satisfying assignment of $\varphi$, contradicting that we were not in case (b).
\end{itemize}
This completes the induction.
\end{proof}

We are now ready to prove $\operatorname{AM-WLB}$.

\begin{lemma}[Formalized block-width lower bound in $\PVO$]
\label{lemma:prover-delayer-in-PV}
    It holds that $\PVO \vdash \operatorname{AM-WLB}$.
\end{lemma}

\begin{proof}
    Working in $\SOT$, let $\pi$ be a correct Resolution refutation of $\Reff_s(\varphi)$ for a CNF formula $\varphi$ over $n$ variables. We can apply \Cref{lemma:induction-width-PV} for $d = \operatorname{depth}(\pi)$ to obtain a sequence of clauses $\calC=(C_1, \dots, C_d)$ of $\pi$ and a sequence of restrictions $A = (\alpha_1, \dots, \alpha_d)$, satisfying one of conditions (a)-(c) in the statement. We claim that condition (c) cannot occur. Suppose it did. If $\calC$ was really a path in $\pi$ starting at the root $C_1 = \bot$, since $\pi$ is a correct Resolution refutation this means the underlying graph of $\pi$ is a DAG and therefore the clause $C_d$ at depth $d$ must be a leaf. It is not hard to see that $\SOT$ can prove this. This implies that $C_d$ is a clause of $\Reff_s(\varphi)$. Now, items (i) and (iv) of the reservation invariant in \Cref{def:reservation-invariant} contradict each other: the restriction $\alpha_d$ must falsify $C_d$, which is a clause of $\Reff_s(\varphi)$; but item (iv) promised that $\alpha_d$ would not falsify any clause of the $\Reff_s(\varphi)$.

    The only viable option then is that we are in case (a) or (b). In case (a), some $C_i \in \pi$ that appears in the sequence $(C_1, \dots, C_d)$ has block-width at least $1/3 (\lfloor{s/n}\rfloor - 1 )$. In case (b), we immediately get a satisfying assignment for $\varphi$.

    Since $\operatorname{AM-WLB}$ is a $\forall\Sigma^b_1$ sentence and $\SOT$ is $\forall\Sigma^b_1$-conservative over $\PVO$, we get that $\PVO \vdash \operatorname{AM-WLB}$.
\end{proof}

\subsection{Formalization of the final lower bound statement}
\label{subsec:putting-it-together-PV}

We are ready to combine the restriction argument and the block-width lower to get a size lower bound.

\begin{theorem}
\label{thm:AM-in-PV}
   There exist a positive $\varepsilon \in \mathbb{Q}$ and $n_0 \in \bbN$ such that $\PVO \vdash \operatorname{AM}_{\varepsilon, n_0}$.
\end{theorem}

 \begin{proof}
    Let $\varepsilon \in \mathbb{Q}$ and $n_0 \in \bbN$ be universal constants that can be computed from the rest of the argument.
    Working in $\PVO$, let $\pi$ be a Resolution refutation of $\Reff_s(\varphi)$ for a CNF formula $\varphi$ over $n$ variables, with $n \geq n_0$. We apply \Cref{lemma:restriction-in-PV} to conclude that $\PVO\vdash \operatorname{AM-Restriction}_4$, from which we get in $\PVO$ that there is some restriction $\rho$ such that $\pi_{\restriction \rho}$ is a Resolution refutation of $\Reff_s(\varphi)_{\restriction\rho}$ and the block-width of $\pi_{\restriction \rho}$ is at most $4\cdot\sqrt{s\log|\pi|}$. Furthermore, $\rho$ is $d$-disabling for some $d \leq 2 \lceil \sqrt{s\log|\pi|} \rceil$.
    
    Before proceeding, extend $\rho$ into a restriction $\rho'$ that further enables all blocks not touched by $\rho$ and sets the values of pointers that are pointing at disabled blocks. In this way, $\Reff_s(\varphi)_{\restriction \rho'}$ becomes exactly $\Reff_{s-d}(\varphi)$, as $\rho$ was $d$-disabling, and $\pi_{\restriction \rho'}$ is a refutation of $\Reff_{s-d}(\varphi)$.

    Assume now that $\varphi$ is unsatisfiable, or else we are already done. By \Cref{lemma:prover-delayer-in-PV} we have that $\PVO \vdash \operatorname{AM-WLB}$. Since we assume $\varphi$ is unsatisfiable, applying this block-width lower bound to $\pi_{\restriction\rho'}$ we conclude that there exists a clause in $\pi_{\restriction\rho'}$ of block-width at least $1/3 (\lfloor (s-d)/n\rfloor - 1)$. It then holds that 
    \begin{equation}
    1/3 (\lfloor (s-d)/n\rfloor - 1) \leq 4\cdot \Bigl\lceil{\sqrt{s\log|\pi|}}\Bigr\rceil\,,
    \end{equation}
    and it is not hard to see that $\PVO$ can derive from this inequality that $||\pi|| > \varepsilon s/n^2$ for some small enough $\varepsilon$ and large enough $n$.
 \end{proof}

\section{Pudlák's upper bound in Resolution}
\label{sec:Pudlak-UB}

The fact that the $\Reff_s(\varphi)$ formulas admit short Resolution refutations whenever $\varphi$ is satisfiable is originally due to Pudlák \cite[Theorem 4.1]{Pudlak03}. However, technically speaking, our definition of $\Reff$ is in the \emph{relativized} form as used by Atserias and Müller, who crucially showed that the upper bound still works even in the presence of the disabling variables \cite[Lemma 11]{AM20}, and the proof goes through for the encoding with pointer variables in binary too \cite[Lemma 2.1.i]{dRGNPRS21}.

Our goal is to show that the upper bound construction can be proven correct in Resolution itself. The main technicality to overcome is that we want to show that from a satisfying assignment $\alpha$ to a CNF formula $\varphi$, one can construct not only a Resolution refutation of $\Reff(\varphi)$, but that this refutation can itself be encoded as a satisfying assignment to the formula $\Reff(\Reff(\varphi))$, and that the correctness of this can be certified in Resolution.  

Let $P(\varphi, \alpha, s)$ denote the circuit that constructs a refutation of $\Reff_s(\varphi)$ given the satisfying assignment~$\alpha$. Intuitively, we want to derive a propositional formula that states
\begin{equation}
\left({\Satf(\varphi, \alpha) \land \pi =P(\varphi, \alpha, s)}\right) \to \Reff(\Reff_s(\varphi), \pi)\,.
\end{equation}
Our goal is to refute the negation of this formula in Resolution, which will not be a CNF formula. However, since Resolution is implicationally complete, the implication can be obtained by showing instead that there is a derivation of the form
\begin{equation}
  \Satf(\varphi, \alpha) \land \pi =P(\varphi, \alpha, s)  \vdash_{\Res}  \Reff(\Reff_s(\varphi), \pi)\,.\label{eq:pneg-dnf}
\end{equation}
The main obstacle then will be to show that the circuit $P$ can be presented in a simple way that can also be handled by Resolution.

In \Cref{subsec:natural-lang-description} we describe the construction in natural language, in a style that streamlines the previous existing proofs. In \Cref{subsec:Pudlak-as-ckt} we describe the circuit $P$, which turns out to be a very low depth circuit that Resolution will be able to reason about. Finally, \Cref{subsec:correct-Res} puts it together to derive the correctness of the construction in Resolution.

\subsection{Description of the construction}
\label{subsec:natural-lang-description}
The first step in the construction of the upper bound is to describe the conjunct $\pi = P(\varphi, \alpha, s)$ in the formula~(\ref{eq:pneg-dnf}). The circuit $P$ relates $\varphi$ and $\alpha$, which are variables of $\Satf(\varphi, \alpha)$, to the variables $\pi$ of $\Reff(\Reff_s(\varphi))$. Before we describe the circuit $P$, we give a natural language explanation of what the upper bound construction is doing.

We remark that in the following description, 
the formula $\varphi$ and 
a satisfying assignment $\alpha$ 
have 
been fixed. 
This means that there are no longer any $\alit^A_\ell$ variables encoding the formula (see \Cref{def:satf}). 
To simplify notation, we sometimes view $\alpha$ as a function over literals so that $\alpha(x_i) = \alpha_i$ and $\alpha(\neg x_i) = \neg \alpha_i$.
 
\paragraph{General structure of the construction.}
The goal of the refutation is to derive, for every $B \in [s]$, the clause
\[ \tag{$\True(B, \alpha)$} \True(B, \alpha) \coloneq \neg \enable^B \lor \bigvee_{\substack{\ell \in \Lit_n\\\alpha(\ell) = 1}} \lit^B_\ell \,, \]
encoding that if block $B$ is enabled then it contains a clause that is satisfied by $\alpha$.
Each $\True(B, \alpha)$ is derived from $\True(C, \alpha)$ for all $C \in [B-1]$, and 
from $\True(s, \alpha)$ one can easily derive the empty clause (by resolving with the unit clauses $\enable^s$ and $\neg \lit^s_\ell$ for all $\ell$ such that $\alpha(\ell)=1$).

We will derive the clause $\True(B, \alpha)$ by first deriving the clauses
\begin{equation}
 \derived^B \lor \True(B, \alpha) \qquad \text{ and } \qquad  \neg \derived^B \lor \True(B, \alpha)
\end{equation}
and then applying one Resolution step.

\paragraph{The derivation of $\derived^B \lor \True(B, \alpha)$.}

For every $A \in [m]$ pick an arbitrary
literal $\ell$ in clause $A$ of $\varphi$ which is made true by~$\alpha$. By weakening the axiom (\ref{axiom:must-appear-after-weak}, $B, A, \ell$), which is
$\neg \enable^B \lor \neg \weak^B_A  \lor \lit^B_\ell$
since the variables $\alit^A_\ell$ are no longer present once a formula has been fixed,
we obtain the clause
\[\neg \enable^B \lor \neg \weak^B_A  \lor \bigvee_{\substack{\ell \in \Lit_n\\\alpha(\ell) = 1}}  \lit^B_\ell \label{D-2} \tag{$L_1(B,A)$} \,.\]
Now cut successively (\ref{axiom:must-weaken}, $B$) with (\ref{D-2}) for $A\in [m]$ to get $\derived^B \lor \True(B, \alpha)$. That is, there will be $m$ lines $L_2(B,1)$ to $L_2(B,m)$, each of the form
\[ \neg \enable^B \lor \derived^B \lor \bigvee_{i = A+1}^m \weak^B_i \lor 
\bigvee_{\substack{\ell \in \Lit_n\\\alpha(\ell) = 1}}  
\lit^B_\ell 
\label{Li}\tag{$L_2(B,A)$} \,,\]
such that $L_2(B,A+1)$ is obtained by resolving $L_2(B,A)$ with $L_1(B,A+1)$
over variable $\weak^B_{A+1}$. Note that $L_2(B,m)$ is precisely $\derived^B \lor \True(B, \alpha)$.

It is not hard to see that, for each $B\in [s]$, this derivation of $\derived^B \lor \True(B, \alpha)$ consists of $\Theta(m)$ resolution steps. 

\paragraph{The derivation of $\neg \derived^B \lor \True(B, \alpha)$.}
We assume that for every $C \in [B-1]$ we have derived $\operatorname{True}(C, \alpha)$.
We first carry out the following derivation for every $C \in [B - 1]$ and $i \in [n]$, leading to auxiliary clauses (\ref{R-2}) defined below.

Let us suppose that $\alpha(x_i) = 0$ (the dual case is analogous but following the right-hand side pointers). 
For $j\in [n]$, let $\ell_j = x_j$ if 
$\alpha(x_j) = 1$ and 
$\ell_j = \lnot x_j$ if 
$\alpha(x_j) = 0$. Note that for each $j\in [n]$ we have the axiom (\ref{axiom:must-appear-after-res-left}, $B,C,i,\ell_j$), which is
\[ \neg \enable^B \lor \res^B_{x_i} \lor \neg \lpoint^B_C \lor \neg \lit^C_{\ell_j} \lor \lit^B_{\ell_j} 
\tag{$\axiomMEARL(B, C, i, j)$}
\label{R-must-appear-left} \,.\]
Successively resolving $\True(C, \alpha)$ with (\ref{R-must-appear-left}) 
over variable $\lit^C_{\ell_j}$
for 
$j\in[n]$
we get
\[ \neg \enable^B \lor \neg \enable^C \lor \neg \res^B_{x_i} \lor \neg \lpoint^B_C \lor \bigvee_{\substack{\ell \in \Lit_n\\\alpha(\ell) = 1}}  \lit^B_\ell \label{R-1} \tag{$R_1(B, C, i)$} \,.\]
This consists of $n$ lines,
which we refer to as $R_1(B, C, i,j)$ for $j\in [n]$. Note that $R_1(B, C, i,n)$ refers to the same line (\ref{R-1}). 
We cut (\ref{R-1}) over variable $\enable^{C}$ with the clause $\axiomMEL \coloneqq \lnot \enable^B \lor \lnot \lpoint^B_{C} \lor \enable^{C}$, which is (\ref{axiom:must-enable-left}$, B, C$), to obtain 
\[ \neg \enable^B \lor \neg \res^B_{x_i} \lor \neg \lpoint^B_C \lor \bigvee_{\substack{\ell \in \Lit_n\\\alpha(\ell) = 1}}  \lit^B_\ell  \label{R-2} \tag{$R_2(B, C, i)$} \,.\]
We now cut (\ref{axiom:must-resolve}, $B$) with (\ref{R-2}) 
over variable $\res^B_{x_i}$
for every $i \in [n]$ to get
\[ \neg \enable^B \lor \neg \derived^B \lor \neg \lpoint^B_C \lor \bigvee_{\substack{\ell \in \Lit_n\\\alpha(\ell) = 1}}  \lit^B_\ell  \label{R-3} \tag{$R_3(B, C, i)$} \,.\]
Finally cutting over the variable $\lpoint^B_{C}$
the clause $\axiomMPL \coloneqq \lnot \enable^B \lor \lnot \derived^B \lor \bigvee_{C \in [B-1]} \lpoint^B_{C}$, which is 
(\ref{axiom:must-point-left}, $B$), with (\ref{R-3}) for each $C \in [B-1]$, we get
\[ \neg \enable^B \lor \neg \derived^B \lor \bigvee_{\substack{\ell \in \Lit_n\\\alpha(\ell) = 1}}  \lit^B_\ell \,,
 \label{R-4} \tag{$R_4(B)$}
\]
which is exactly $\neg \derived^B \lor \True(B, \alpha)$. 
This step consists of $B-1$ lines, which we denote by $\axiomMPLC{C}$ for  $C \in [B-1]$, where the last line $\axiomMPLC{B-1}$ is (\ref{R-4}).

We note that, for each $B\in [s]$, we derive $\neg \derived^B \lor \True(B, \alpha)$ in $\Theta(sn^2)$ resolution steps. 

\paragraph{Contradiction from $\True(s)$.}
Once we have derived $\True(s)$, we resolve it with (\ref{axiom:root-enabled}) to get 
\begin{equation}
    \bigvee_{\substack{\ell \in \Lit_n\\\alpha(\ell) = 1}} \lit^B_{\ell}\,,
    \end{equation}
which we then resolve with the axioms (\ref{axiom:root-empty}, $\ell$) for the $n$ literals $\ell$ such that $\alpha(\ell) = 1$ to get the empty clause.

We observe that the total number of clauses in this resolution refutation of $\Reff_s(\varphi)$, for a satisfiable formula $\varphi$, is $\Theta(s(m+sn^2))$.

\subsection{The construction as a low-depth circuit}
\label{subsec:Pudlak-as-ckt}
The expression $\pi = P(\varphi, \alpha, s)$ in the formula (\ref{eq:pneg-dnf}) is a short-hand for a conjunction of clauses describing the upper bound construction for refuting $\Reff_s(\varphi)$ as presented in Section~\ref{subsec:natural-lang-description}.
As noted above, the number of clauses in this refutation is $\tau = \Theta(s(m+sn^2))$, where $n$ is the number of variables and $m$ is the number of clauses of $\varphi$.
We only define $\pi = P(\varphi, \alpha, s)$ for $\pi$ being a proof of length $t\geq \tau$, as we will only consider this short-hand in this parameter regime. This subsection is dedicated to the description of the conjunction of clauses encoding $\pi = P(\varphi, \alpha, s)$.
Each clause~$\calC$ that appears in the construction, including the axioms of $\Reff$ that are used in cuts,
is mapped to some $\calB \in [\tau]$ which corresponds to its position in the proof.
Since this a one-to-one mapping, we sometimes abuse notation and identify $\calC$ with $\calB$. 

The key observation is that each variable of $\pi$ is a function of at most two variables of $\varphi$ and $\alpha$,
and thus the construction can be expressed as a depth-$2$ circuit. This is because if step $\calB$ of the construction is, say, to derive a clause $\calC$ (over variables of $\Reff_s(\varphi)$), then we will add the unit clauses
\begin{equation}
    \enable^\calB   \land \bigwedge_{z\in \calC} \lit_z^\calB \land \bigwedge_{z \not\in \calC}\neg \lit_z^\calB \,,
    \label{eq:setting-enbl-and-lit}
\end{equation}
encoding that clause $\calC$ is the clause at step $\calB$ in the proof.
If the construction obtains this clause by a Resolution step, we add the conjunct $\derived^\calB$ and if it is by a weakening of an axiom, we add the conjunct $\neg \derived^\calB$. Similarly, we include appropriate unit clauses of the variables of type $\lpoint$, $\rpoint$ and $\weak$. 
For the most part, the construction is a ready-made template that only depends on $\alpha$ and $\varphi$ in a few crucial points. In this case, we will include clauses that encode this dependence.
For example, suppose that at step~$\calB$ we derive the clause 
\begin{equation}
\bigvee_{\substack{\ell \in \Lit_n\\\alpha(\ell) = 1}} \lit^B_\ell
\end{equation}
that does depend on $\alpha$. Here $B$ is a block of the inner $\Reff$ formula, while this clause itself will be instantiated as a block $\calB$ of the outer $\Reff$. We then add, for $i\in [n]$, the clauses encoding
\begin{equation}
    \lit^\calB_{\lit^B_{x_i}} \leftrightarrow \alpha_i \qquad \text { and } \qquad \lit^\calB_{\lit^B_{\neg x_i}} \leftrightarrow \neg \alpha_i\,.
    \label{eq:litalpha}
\end{equation}

We first describe in more detail how the variables related to the lines (\ref{D-2}) and (\ref{Li}) are defined. Line (\ref{D-2}) is the only one that depends on both $\alpha$ and $\varphi$; the rest of the construction depends solely on~$\alpha$.
For subsequent lines in the construction, we simply describe the parts that depend on $\alpha$ as the others are completely determined by the ready-made template, as explained above.

Fix $A\in [m]$, and let $\Lblock \in [t]$ denote
the block corresponding to (\ref{D-2}).
We include the unit clauses
\begin{equation}
    \enable^{\Lblock} \land \lnot \derived^{\Lblock} 
    \label{L1i-detail}
\end{equation}
determining that this line is enabled and is not derived. We also include unit clauses $\lnot \res^{\Lblock}_{z}$, $\lnot \lpoint^{\Lblock}_{\calB}$ and $\lnot \rpoint^{\Lblock}_{\calB}$ for all $\calB \in [t]$ (since $\Lblock$ is not obtained by a Resolution step).
We set the $\lit$ variables as explained above, that is,
we have clauses encoding, for $i\in[n]$,
\begin{equation}
    \lit^{\Lblock}_{\lit^B_{x_i}}  \leftrightarrow \alpha_i \qquad \text{ and } \qquad \lit^{\Lblock}_{\lit^B_{\neg x_i}}  \leftrightarrow \neg\alpha_i \,,
\end{equation}
the unit clauses
\begin{equation}
\lit^{\Lblock}_{\neg \enable^B} \land \lit^{\Lblock}_{\neg \weak_A^B} \,,
\end{equation}
and the unit clauses $\lnot \lit^{\Lblock}_{\ell'}$ encoding that no other literal $\ell'$ (not appearing above) is present in the clause.
It remains to define the $\weak^{\Lblock}$
variables.
Note that (\ref{D-2})
can be obtained by weakening any of the axioms (\ref{axiom:must-appear-after-weak}, $B, A, \ell$), for any $\ell$ that appears in clause $A$ of $\varphi$ and that is made true by~$\alpha$. 
We therefore set, for every $i\in[n]$, 
\begin{equation}
    \weak^{\Lblock}_{(\ref{axiom:must-appear-after-weak}, B, A, x_i)}  \leftrightarrow \alit^A_{x_i} \land \alpha_i \qquad \text{ and } \qquad 
    \weak^{\Lblock}_{(\ref{axiom:must-appear-after-weak}, B, A, \neg x_i)}  \leftrightarrow \alit^A_{\neg x_i} \land \neg\alpha_i \,.
    \label{lit-and-alpha}
\end{equation}

Here the variables $\alit^A_{x_i}$ and $\alit^A_{\neg x_i}$ are part of the $\Satf(\varphi, \alpha)$ formula (as in \Cref{def:satf}), and determine which literals appear in what clauses of $\varphi$.
Observe that we are technically establishing that the clause (\ref{D-2}) will have many weakening pointers. While this is somewhat unusual in real Resolution refutations, the refutation is still sound and none of the axioms in \Cref{def:Reff-formulas} are violated. The advantage is that we do not have to establish which is the, say, first satisfied literal of every axiom. Even though (potentially) having multiple weakening pointers is not strictly necessary, it 
simplifies the clauses needed to express $\pi = P(\varphi, \alpha, s)$.

Now fix $A\in [m]$ and let $\Lblock$ be the block corresponding to (\ref{Li}). We add the unit clauses
\begin{equation}
    \enable^{\Lblock} \land \derived^{\Lblock} \land \res^{\Lblock}_{\weak^B_{A}}
    \land \bigwedge_{z \neq \weak^B_{A}} \lnot \res^{\Lblock}_{z}
    \label{Li-detail}
\end{equation}
that determine that this line is enabled, it is derived and obtained by resolving over variable $\weak^B_{A}$. We also include unit clauses enforcing all $\weak^{\Lblock}$ variables to be $0$ (since $\Lblock$ is not a weakening of any axiom). The $\lit^{\Lblock}$ variables are encoded as explained above.
As for the pointers, we include the unit clauses
\begin{equation}
    \rpoint^{\Lblock}_{L_1(B,A)} \land \lpoint^{\Lblock}_{L_2(B,A-1)} \land \bigwedge_{\calB \neq L_1(B,A)} \lnot \rpoint^{\Lblock}_{\calB}\land \bigwedge_{\calB \neq L_2(B,A-1)} \lnot \lpoint^{\Lblock}_{\calB}\,.
    \label{Li-more-detail}
\end{equation}

For the second set of clauses, used in the derivation of $\neg \derived^B \lor \True(B, \alpha)$, we describe the parts that depend on $\alpha$.
These can be split into two groups, the parts that depend only on whether we consider $\lpoint$ or $\rpoint$, that is, only on $\alpha_i$, and those that depend also on $\ell_j$.
We start with the former.
If in the construction we described above a clause on a line corresponding to, say, block $\calR$
contains the literal $\neg \lpoint^B_{C}$ for the case when $\alpha_i = 0$, then we
include the clauses encoding
\begin{equation}
    \lit^{\calR}_{\lnot \rpoint^B_{C}}  \leftrightarrow \alpha_i \qquad \text{ and } \qquad \lit^{\calR}_{\lnot \lpoint^B_{C}}  \leftrightarrow \neg\alpha_i \,, \label{eq:litRpoint}
\end{equation}
so that the literal $\lnot \rpoint^B_{C}$ is present iff $\alpha(x_i) = 1$ and 
$\lnot \lpoint^B_{C}$ is present if and only if $\alpha(x_i) = 0$. 
Similarly if it contains the literal $ \lpoint^B_{C}$, we include $\lit^{\calR}_{ \rpoint^B_{C}}  \leftrightarrow \alpha_i$
and $\lit^{\calR}_{ \lpoint^B_{C}}  \leftrightarrow \neg\alpha_i$.
For $\calR = \axiomMPLC{C}$ (one of the intermediate steps to derive \ref{R-4}), we include clauses
\begin{equation}
    \res^{\calR}_{\rpoint^B_C} \leftrightarrow \alpha_i \qquad \text{ and } \qquad \res^{\calR}_{\lpoint^B_C} \leftrightarrow \neg\alpha_i \,,
    \label{eq:resRpoint}
\end{equation}
encoding that the resolved variable is either 
$\rpoint^B_C$ or $\lpoint^B_C$, depending on $\alpha_i$. For $\calR = \axiomMEL$ we include
\begin{equation}
    \weak^{\calR}_{(\ref{axiom:must-enable-right}, B, C)} \leftrightarrow \alpha_i \qquad \text{ and } \qquad 
    \weak^{\calR}_{(\ref{axiom:must-enable-left}, B, C)} \leftrightarrow \neg\alpha_i \,, \label{eq:weakRpoint}
\end{equation}
and for $\calR = \axiomMPL$ we include
\begin{equation}
    \weak^{\calR}_{(\ref{axiom:must-point-right}, B)} \leftrightarrow \alpha_i \qquad \text{ and } \qquad 
    \weak^{\calR}_{({\ref{axiom:must-point-left}}, B)} \leftrightarrow \neg\alpha_i \,,
    \label{eq:weakRRef56}
\end{equation}
encoding, in each case, which axiom $\calR$ is a weakening of.

We now move to the parts that (also) depend on $\ell_j$.
If the line corresponding to block $\calR$ contains $\lit^{B'}_{\ell_j}$ for some block ${B'}$
we include 
 \begin{equation}
    \lit^{\calR}_{\lit^{B'}_{x_j}}
    \leftrightarrow \alpha_j  \qquad \text{ and } \qquad 
    \lit^{\calR}_{\lit^{B'}_{\neg x_j}}\leftrightarrow \neg \alpha_j \,.
\end{equation}
Similarly, if it contains $\neg \lit^{B'}_{\ell_j}$ for some block ${B'}$
we include $\lit^{\calR}_{\neg \lit^{B'}_{x_j}}\leftrightarrow \alpha_j $ and
$\lit^{\calR}_{\neg \lit^{B'}_{\neg x_j}}\leftrightarrow \neg \alpha_j $.
If $\calR$ is a line 
$R_1(B,C,i,j)$ (one of the intermediate steps to derive \ref{R-1}) that is obtained by resolving over variable $\lit_{\ell_j}^C$, we include the clauses
 \begin{equation}
    \res^{\calR}_{\lit_{x_j}^C}  \leftrightarrow \alpha_j \qquad \text{ and } \qquad \res^{\calR}_{\lit_{\neg x_j}^C}  \leftrightarrow \neg\alpha_j \,.
    \label{eq:resRlit}
\end{equation}
Finally, for $\calR$ being \ref{R-must-appear-left}, we include
 \begin{equation}
    \weak^{\calR}_{(\ref{axiom:must-appear-after-res-right}, B,C,i,x_j)}  \leftrightarrow \alpha_i \land \alpha_j \,,
    \qquad 
    \text{  } \qquad \weak^{\calR}_{(\ref{axiom:must-appear-after-res-left}, B,C,i, x_j)}  \leftrightarrow \neg\alpha_i \land \alpha_j \,,
    \label{eq:weakR2alphas1}
\end{equation}
 \begin{equation}
    \weak^{\calR}_{(\ref{axiom:must-appear-after-res-right}, B,C,i,\neg x_j)}  \leftrightarrow \alpha_i  \land \neg\alpha_j \,,
    \qquad \text{ and } \qquad \weak^{\calR}_{(\ref{axiom:must-appear-after-res-left}, B,C,i,\neg x_j)}  \leftrightarrow \neg\alpha_i  \land \neg\alpha_j \,.
    \label{R-more-detail-last}
\end{equation}

We have thus far described how we assign variable referring to blocks $\calB$ of $\pi$ for $\calB \leq \tau$ (where we recall $\tau = \Theta(s(m+sn^2))$ is the number of clauses in the refutation of $\Reff_s(\varphi)$ described in Section~\ref{subsec:natural-lang-description}). 
If $t>\tau$, we can extend the construction easily by copying all
clauses on block $\tau$ to the last block $t$ (to make sure we satisfy that the last block is enabled and contains the empty clause, and is derived from the correct previous clauses) and for all other $\calB \in [t-1]$ with $\calB > \tau$ we include unit clauses setting all variables to $0$, in particular, we include clauses $\neg \enable^{\calB}$ encoding that these blocks are not used in the refutations described by $P(\varphi, \alpha, s)$.

We denote by $\pi = P(\varphi, \alpha, s)$ all the conjuncts describing the refutation of $\Reff_s(\varphi)$ as an assignment to the outer $\Reff$ formula. 

\subsection{Correctness of the construction in Resolution}
\label{subsec:correct-Res}



Before we formally show that
the construction presented above is provably correct in Resolution,
we need to explicitly describe the formula
$\Reff_t(\Reff_s(\varphi), \pi)$.
This formula states that $\pi$ is a valid
Resolution refutation of the formula $\Reff_s(\varphi)$, stating that $\varphi$
has a size-$s$ Resolution refutation.
Recall that the variables of a formula $\Reff_t(\psi, \pi)$
are the variables describing $\pi$ and those describing $\psi$.
When $\psi = \Reff_s(\varphi)$,
most of the variables of $\psi$ are
already set, and the only remaining ones
are the variables describing $\varphi$.
These variables determine which clauses of
$\Reff_s(\varphi)$ exist,
that is, they indicates which of the (\ref{axiom:must-appear-after-weak}) axioms 
are present in $\Reff_s(\varphi)$. The other axioms describing $\Reff_s(\varphi)$ (see \Cref{def:Reff-formulas}) do not depend on $\varphi$.

Recall that the formula $\varphi$ has $n$ variables and $m$ clauses, and that the literals that appear in each clause are determined by the variables $\alit$,
that is, $\alit^A_\ell$ is set to $1$ if the literal $\ell$ appears in clause $A$ of $\varphi$.
For every literal $\ell\in \Lit_n$ and for every $A\in[m]$, if $\alit^A_\ell$ is set to $1$
then the axiom $( \enable^B \land \weak^B_A ) \to \lit^B_\ell $,
which is (\ref{axiom:must-appear-after-weak}) with $\alit^A_\ell = 1$, 
appears in $\Reff_s(\varphi)$,
and if $\alit^A_\ell$ is set to $0$ then it does not appears in $\Reff_s(\varphi)$.

In order to write
the propositional formula $\Reff_t(\Reff_s(\varphi), \pi)$, we need to know how many clauses appear in 
$\Reff_s(\varphi)$. This, in turn, depends on
$\varphi$. To deal with this, we include all potential axioms $( \enable^B \land \weak^B_A ) \to \lit^B_\ell $, and we let the variables $\alit$ deactivate the axiom that should not exist, which can be done by enforcing that
no block in $\pi$ can be a weakening of
that axiom. Formally,
let ${(\ref{axiom:must-appear-after-weak}, B, A, \ell)}$ indicate the position in $\Reff_s(\varphi)$ of the potential (\ref{axiom:must-appear-after-weak}) axiom defined over $B, A$ and $\ell$.
For every block $\calB$ of $\pi$,
for every $B\in[s], A\in[m]$ and $\ell\in \Lit_n$,
we include in $\Reff_t(\Reff_s(\varphi), \pi)$ the clause
$\alit^A_\ell \lor \neg \weak^\calB_{(\ref{axiom:must-appear-after-weak}, B, A, \ell)}$.

We are now ready to show that the construction is provably correct in Resolution. 

\begin{theorem}[Pudlák's upper bound in Resolution]
    \label{thm:ub-in-Res}
    For every $n, m, s \in \bbN$, 
    there is some $\tau=\Theta(s(m+sn^2))$ such that for $t\geq \tau$
    there exist uniform polynomial-size Resolution derivations of the form
    \[ \Satf(\varphi, \alpha) \land \pi =P(\varphi, \alpha, s)  \vdash_{\Res}  \Reff_t(\Reff_s(\varphi), \pi) \,,\]
    where the $\Satf$ formula is for CNF formulas with $n$ variables and $m$ clauses, and $\pi = P(\varphi, \alpha, s)$ stands for the conjunction of clauses describing the upper bound construction, as per \Cref{subsec:Pudlak-as-ckt}. Moreover, the derivations can be described uniformly in polynomial time.
\end{theorem}

\begin{proof}
Let $\tau = \Theta(s(m+sn^2))$ be the number of clauses in the refutation of $\Reff_s(\varphi)$ described in Section~\ref{subsec:natural-lang-description}, and let $\mathfrak{C}$ denote the set of clauses in the formula $\Reff_t(\Reff_s(\varphi), \pi)$. 
%
%
The clauses $\pi = P(\varphi, \alpha, s)$ encode the correct computation of the construction, meaning that if $\Satf(\varphi, \alpha)$ holds, then indeed $\pi$ encodes a correct Resolution refutation, so all clauses $c \in \frakC$ are sound inferences: for every $c \in \frakC$, we have $\Satf(\varphi, \alpha) \land \pi =P(\varphi, \alpha, s) \models c$. We only need to verify that there are efficient derivations in Resolution for each $c \in \frakC$.

Observe that the clauses of the form $\alit^A_\ell \lor \neg \weak^\calB_{(\ref{axiom:must-appear-after-weak}, B, A, \ell)}$ are either present in 
$\pi =P(\varphi, \alpha, s)$ (by \ref{lit-and-alpha}), if $\calB$ denotes
the block corresponding to (\ref{D-2}),
or otherwise can be obtained by weakening the unit clause $\neg \weak^\calB_{(\ref{axiom:must-appear-after-weak}, B, A, \ell)}$, which in this case is part of $\pi = P(\varphi, \alpha, s)$.
For all other clauses $c \in \frakC$,
we argue that it is efficiently derivable in Resolution by considering different cases, according 
to which type of axiom (\ref{axiom:must-appear-after-res-left} - \ref{axiom:root-enabled}) 
it is, 
as well as to the block $\calB \in [t]$ it refers to,
where
we say that a clause $c \in \frakC$ \emph{refers to block $\calB \in [t]$} if 
the $\enable$ variable at the beginning of the clause is $\enable^\calB$
(and clauses (\ref{axiom:root-empty}) and
(\ref{axiom:root-enabled}) refer to the root $t$).
For most cases, we show that Resolution can obtain $c$ from only $\pi =P(\varphi, \alpha, s)$ (namely, the clauses in $\Satf(\varphi, \alpha)$ are not necessary for the derivation), either by weakening a clause or by a constant-size derivation.
There are only some clauses $c$ of type (\ref{axiom:must-weaken}) which 
are not semantically implied by 
$\pi =P(\varphi, \alpha, s)$---these, we argue,
can be derived from
$\Satf(\varphi, \alpha) \land \pi =P(\varphi, \alpha, s)$ in Resolution in size $O(n)$.


Recall that the clauses in $\pi = P(\varphi, \alpha, s)$ are such that each variable of $\pi$ is a function of at most two variables of $\varphi$ and $\alpha$. By inspecting the description of the refutation as a circuit in \Cref{subsec:Pudlak-as-ckt}, it is clear that the only variables of $\pi$ that depend on both $\alit$ and $\alpha$ variables are precisely the $\weak^\calL_{(\ref{axiom:must-appear-after-weak}, B, A, \ell)}$ variables in line~(\ref{lit-and-alpha}). All other variables are either enforced to be true or false by unit clauses or depend only on $\alpha$ variables.
We first argue that for all clauses $c \in \mathfrak{C}$ that do not contain the $\weak^\calL_{(\ref{axiom:must-appear-after-weak}, B, A, \ell)}$ variables in line~(\ref{lit-and-alpha}), it holds that
$\pi =P(\varphi, \alpha, s) \vdash_{\Res} c$. In these cases, the argument is very straightforward, and $c$ will follow by weakening or by one application of the Resolution rule followed by weakening.

Before we split into cases according to the axiom type of the clause $c$, 
we make one simple observation regarding the blocks $\calB \in [t]$ which are larger than $\tau$.
Recall that the output of $P$ is a refutation of size $\tau$, but we may set $t$ to be a larger value. In this case, all non-root blocks larger than $\tau$ will simply be disabled.
This implies that any clause $c$ that refers to a block $\calB$ such that $\tau < \calB <t$, it holds that $c$ contains the literal $\neg \enable^{\calB}$ and therefore can be obtained by weakening the unit clause $\neg \enable^{\calB}$, 
which will appear in $\pi =P(\varphi, \alpha, s)$.
We therefore assume $\calB \leq \tau$ or $\calB = t$ (in which case $\pi =P(\varphi, \alpha, s)$ contains the unit clause $\enable^{\calB}$).

\begin{description}
    \item[If $c$ is a clause  (\ref{axiom:root-empty}) or (\ref{axiom:root-enabled}):] that is, either 
    $$ \neg \lit^t_\ell 
        \quad\text{or}\quad \enable^t \,,$$
    then these clauses are present in
    $\pi =P(\varphi, \alpha, s)$. Indeed, since the $\tau$-th clause of the refutation of $\Reff_s(\varphi)$ as presented in Section~\ref{subsec:natural-lang-description} is the empty clause, by (\ref{eq:setting-enbl-and-lit}) we have that the clauses $\neg \lit^\tau_\ell $ and $\enable^\tau$ are present in $\pi =P(\varphi, \alpha, s)$, and therefore so are $\neg \lit^t_\ell $ and $\enable^t$, since we copied such clauses for $t$.
    
    \item[If $c$ is a clause (\ref{axiom:must-point-left})
    or (\ref{axiom:must-point-right}):] that is, either
        $$\left( \enable^{\calB} \land \derived^{\calB} \right) \to \bigvee_{\substack{{\calB}' \in [s] \\ {\calB}' < {\calB}}} \lpoint^{\calB}_{{\calB}'} \quad\text{or}\quad \left( \enable^{\calB} \land \derived^{\calB} \right) \to \bigvee_{\substack{{\calB}' \in [s] \\ {\calB}' < {\calB}}} \rpoint^{\calB}_{{\calB}'} \,,$$
 then we have two cases. If ${\calB}$ corresponds to a clause obtained from a weakening of an axiom, then we have that the unit clause $\neg \derived^{\calB}$ is in $\pi =P(\varphi, \alpha, s)$ (by (\ref{L1i-detail})), and therefore the clause $c$ can be derived by weakening. If ${\calB}$ corresponds to a clause obtained from a derivation step, 
 then $\pi =P(\varphi, \alpha, s)$ contains unit clauses $\lpoint^{\calB}_{{\calB}'}$ and $\rpoint^{\calB}_{{\calB}'}$ for some $\calB'$ (by (\ref{Li-more-detail})) and we again conclude that $c$ can be derived by weakening.

    \item[If $c$ is a clause (\ref{axiom:must-resolve}):] that is,
    $$\left( \enable^{\calB} \land \derived^{\calB} \right ) \to \bigvee_{x} \res^{\calB}_{x}\,,$$
    we have a similar case. Either $\calB$ was not obtained by a derivation (and the clause $\lnot \derived^{\calB}$ is in $\pi =P(\varphi, \alpha, s)$), or  there is an $x$ for which the unit clause $\res^{\calB}_{x}$ is present in $\pi =P(\varphi, \alpha, s)$ or can be derived from two clauses $\res^{\calB}_{x} \lor \alpha_i$ and $\res^{\calB}_{x} \lor \lnot \alpha_i$ which are in $\pi =P(\varphi, \alpha, s)$ (by (\ref{eq:resRpoint})). The clause~$c$ can then be obtained by weakening either $\lnot \derived^{\calB}$ or $\res^{\calB}_{x}$.

\item[If $c$ is a clause (\ref{axiom:must-enable-left}) or (\ref{axiom:must-enable-right}):] that is, either
    $$\left( \enable^{\calB} \land  \lpoint^{\calB}_{{\calB}'} \right) \to \enable^{{\calB}'}  \quad\text{or}\quad
        \left( \enable^{\calB} \land  \rpoint^{\calB}_{{\calB}'} \right) \to \enable^{{\calB}'} \,, $$
    then either $\calB$ was obtained by a derivation step involving the clause $\calB'$, in which case  $\pi =P(\varphi, \alpha, s)$ contains the unit clause $\enable^{{\calB}'}$, or it was not,
    in which case  $\pi =P(\varphi, \alpha, s)$ contains the unit clauses $\lnot \lpoint^{\calB}_{{\calB}'}$ and $\lnot \rpoint^{\calB}_{{\calB}'}$. Either way, $c$ can be derived by a weakening step.

    \item[If $c$ is a clause (\ref{axiom:must-appear-after-res-left})
    or (\ref{axiom:must-appear-after-res-right}):] that is, either
         $$\left( \enable^{\calB} \land \res^{\calB}_{x_i} \land \lpoint^{\calB}_{{\calB}'} \land \lit^{{\calB}'}_\ell \right) \to \lit^{{\calB}}_\ell  
        \quad\text{or}\quad \left( \enable^{\calB} \land \res^{\calB}_{x_i} \land \rpoint^{\calB}_{{\calB}'} \land \lit^{{\calB}'}_\ell \right) \to \lit^{{\calB}}_\ell  \,,$$
    then, in the only case when $c$ cannot be obtained by weakening, it can be derived by first resolving over $\alpha_i$ the clauses $\neg \lit^{\calB'}_\ell \lor \alpha_i$ and $ \lit^{\calB}_\ell \lor \neg \alpha_i$ (or $\neg \lit^{\calB'}_\ell \lor \neg \alpha_i$ and $ \lit^{\calB}_\ell \lor \alpha_i$), to obtain 
    $\neg \lit^{\calB'}_\ell \lor \lit^{\calB}_\ell$, and then weakening.
    
    \item[If $c$ is a clause (\ref{axiom:must-appear-after-weak}) that does not contain $\weak^{\Lblock}_{(\ref{axiom:must-appear-after-weak}, B, A, \ell)}$:] that is, 
     $$\left( \enable^{\calB} \land \weak^{\calB}_A \land \alit^A_{\ell}\right) \to \lit^{\calB}_\ell \,, $$
     then either it can be obtained by weakening, or it can be derived by  resolving a clause encoding $\lit^{\calB}_\ell \leftrightarrow \alpha$ or $\lit^{\calB}_\ell \leftrightarrow \lnot \alpha$ of the type (\ref{eq:litRpoint}) with a clause encoding $\weak^{\calB}_A \leftrightarrow \alpha_i$ or $\weak^{\calB}_A \leftrightarrow \lnot \alpha_i$ of the type (\ref{eq:weakRpoint}) or (\ref{eq:weakRRef56}), and then weakening.
     
     \item[If $c$ is a clause (\ref{axiom:must-weaken}) that does not contain $\weak^{\Lblock}_{(\ref{axiom:must-appear-after-weak}, B, A, \ell)}$:] that is, 
     $$\left( \enable^{\calB} \land  \neg \derived^{\calB} \right) \to \bigvee_{A\in[m]} \weak^{\calB}_A  \,,$$
     then either it can be obtained by weakening, or it can be obtained by resolving a clause encoding $\weak^{\calB}_A \leftrightarrow \alpha_i$ with a clause encoding $\weak^{\calB}_A \leftrightarrow \lnot \alpha_i$, for two clauses of the type (\ref{eq:weakRpoint}) or of the type (\ref{eq:weakRRef56}), and then weakening.
\end{description}

It remains to argue that if $c$ is one of the axiom (\ref{axiom:must-appear-after-weak}) or (\ref{axiom:must-weaken}) that contains the variable $\weak^{\Lblock}_{(\ref{axiom:must-appear-after-weak}, B, A, \ell)}$ from (\ref{lit-and-alpha}), 
Resolution can derive $c$ efficiently from $\Satf(\varphi, \alpha) \land \pi =P(\varphi, \alpha, s)$.
We first argue that Resolution can efficiently derive axioms $c$ of type (\ref{axiom:must-appear-after-weak}). Observe that the only case when $c$ cannot be obtained by weakening is when $c$ is
\begin{equation}
    \left( \enable^{\Lblock} \land \weak^{\Lblock}_{(\ref{axiom:must-appear-after-weak}, B, A, \ell)} \right) \to \lit^{\Lblock}_{\ell'} \,,
\end{equation}
for $\ell' = \lit^B_{\ell}$. Indeed, 
$\alit^{(\ref{axiom:must-appear-after-weak}, B, A, \ell)}_{\ell'} = 0$ for all $\ell'$ that are not
$\neg \enable^{B}, \neg \weak^{B}_{A}$ or 
$\lit^{B}_{\ell}$ (since the axiom $(\ref{axiom:must-appear-after-weak}, B, A, \ell)$ without the $\alit^{A}_{\ell}$ variable is $( \enable^{B} \land \weak^{B}_{A} ) \to \lit^{B}_{\ell}$) and 
$\pi =P(\varphi, \alpha, s)$ contains the
clauses $\lit^\calL_{\neg \enable^B}$ and $\lit^\calL_{\neg \weak^B_A}$,
since $\Lblock$ corresponds to clause (\ref{D-2}),
that is,
\begin{equation}
    \neg \enable^B \lor \neg \weak^B_A  \lor \bigvee_{\substack{\ell \in \Lit_n\\\alpha(\ell) = 1}}  \lit^B_\ell \,. \label{axiomRef3}
\end{equation}
Now when $\ell' = \lit^B_{\ell}$, we can derive the clause $\neg \weak^{\Lblock}_{(\ref{axiom:must-appear-after-weak}, B, A, \ell)} \lor \lit^{\Lblock}_{\ell'}$ by 
resolving a clause of (\ref{eq:litalpha}) 
with a clause of (\ref{lit-and-alpha}),
and then (\ref{axiomRef3}) can be obtained by weakening.

Finally, we argue that if $c$ is one of the axiom (\ref{axiom:must-weaken}) that depend on (\ref{lit-and-alpha}), Resolution can derive it in $O(n)$ steps. Recall that, by (\ref{lit-and-alpha}), for all $i\in [n]$ and $A\in[m]$, the formula $\pi = P(\varphi, \alpha, s)$  contains the clause
\begin{equation}
    \weak^{\Lblock}_{(\ref{axiom:must-appear-after-weak}, B, A, x_j)}  \lor \neg \alit^A_{x_j} \lor \neg \alpha_j
    \label{eq:clause1}
\end{equation} 
and the clause
\begin{equation}
    \weak^{\Lblock}_{(\ref{axiom:must-appear-after-weak}, B, A, \neg x_j)}  \lor \neg \alit^A_{\neg x_j} \lor \alpha_j \,.
    \label{eq:clause2}
\end{equation} 
By resolving the clause (\ref{eq:clause1}) with (\ref{sat-1}, $A, x_j$) and then with (\ref{sat-2}, $A, j$), we obtain the clause 
\begin{equation}
\weak^{\Lblock}_{(\ref{axiom:must-appear-after-weak}, B, A, x_i)}  \lor \sat^A_{x_i}\,.
\label{eq:clause3}
\end{equation} 
Similarly, we can derive 
\begin{equation}
\weak^{\Lblock}_{(\ref{axiom:must-appear-after-weak}, B, A, x_i)}  \lor \sat^A_{\lnot x_i}
\label{eq:clause4}
\end{equation} 
by resolving (\ref{eq:clause2}) with (\ref{sat-1}, $A, \neg x_i$) and then with (\ref{sat-3}, $A, i$).
Resolving (\ref{sat-4}) with each (\ref{eq:clause3}) and (\ref{eq:clause4}) for all $i\in [n]$, we obtain 
\begin{equation}
    \bigvee_{i \in [n]} \left(\weak^{\Lblock}_{(\ref{axiom:must-appear-after-weak}, B, A, x_i)} \lor \weak^{\Lblock}_{(\ref{axiom:must-appear-after-weak}, B, A, \neg x_i)} \right) \,,
\end{equation} 
from which $c$ can be derived by weakening.
This completes the proof.
\end{proof}

\section{Extended Frege proves that automating Resolution is $\NP$-hard}
\label{sec:AM-in-EF}

We now have the lower bound (\Cref{thm:AM-in-PV}) needed to show the $\NP$-hardness of automatability formalized in $\PVO$. This is a $\forall\Sigma^b_1$ sentence, meaning that the existential quantifiers can be witnessed by polynomial-time functions and via Cook's translation turned into propositional formulas with short Extended Frege proofs.

\begin{theorem}[Extraction algorithm in $\EF$]
\label{thm:lb-formal}
    There exists a uniform family of polynomial-size circuits $\{ E_{n,m,s,t} \}_{n,m,s,t \in \bbN}$ such that for every CNF formula $\varphi$ over $n$ variables and $m$ clauses, if $\varphi$ is satisfiable and $\pi$ is a Resolution refutation of $\Reff_s(\varphi)$ for $s \geq n^3$, then the circuit $E_{n,m,s,t}(\varphi, \pi)$ outputs a satisfying assignment for $\varphi$. Furthermore, this is provable in Extended Frege, that is,  there exists a polynomial $\ell(n, m, s, t)$ such that
    \[ \EF \vdash_{\ell(n,m,s,t)} \Reff_{t}(\Reff_s(\varphi), \pi) \to \Satf(\varphi, E_{n,m,s,t}(\varphi, \pi)) \,,\]
    and these proofs can be generated uniformly in polynomial time.
\end{theorem}

\begin{proof}
    By \Cref{thm:AM-in-PV}, we know there exists a positive $\varepsilon \in \bbQ$ and $n_0 \in \bbN$ such that $\PVO \vdash \foAM_{\varepsilon, n_0}$, the first-order statement encoding the lower bound on $\Reff$ formulas (\ref{eq:foAM}). Since $\foAM_{\varepsilon, n_0}$ is a $\forall\Sigma^b_1$ formula, by Buss's witnessing theorem (\Cref{thm:buss-witnessing}), there exists a $\PV$ function $E_0$ such that 
    \begin{equation}
    \label{eq:pv-trans}
    \begin{split}
    \PVO \vdash \forall \varphi \forall n \forall s \forall \pi \Big(n < n_0 \lor \neg \operatorname{CNF}(\varphi, n) &\lor \neg \foReff(\Reff_{s}(\varphi), \pi) \\ &\lor \foSatf(\varphi, E_0(\varphi, n, s, \pi)) \lor||\pi|| > \varepsilon s/n^2 \Big)\,.
    \end{split}
    \end{equation}

    This means that, when plugging a formula $\varphi$ that is indeed a CNF formula over $n > n_0$ variables and the length of $\pi$ is $|\pi| \leq 2^{\varepsilon n}$ and $s\geq n^3$, then the witnessing function $E_0$ correctly succeeds in finding a satisfying assignment of $\varphi$, and this is all provable in $\PVO$. Building on this $E_0$, we can further define a new polynomial-time function $E$ that first checks whether $|\pi|\leq 2^{\varepsilon n}$, and then runs $E_0$, and otherwise performs a brute-force check for a satisfying assignment for $\varphi$. It is not hard to see, building on \cref{eq:pv-trans}, that $\PVO$ can now prove that this function $E$ is a polynomial-time algorithm that always succeeds in finding a satisfying assignment if one exists, regardless of the size of $\pi$.

    By applying Cook's translation (\Cref{thm:translation}) to \cref{eq:pv-trans} and restricting $s \geq n^3$, we obtain a uniform family of $\EF$ proofs showing the correctness of $E$ as an extraction algorithm. Note that $\foReff(\varphi, \pi)$ is the first-order analogue of the propositional formula $\Reff_s(\varphi, \pi)$ that we have been studying until now. Cook's translation on $\foReff(\varphi, \pi)$ gives a propositional formula with the same behavior as $\Reff_s(\varphi)$. While $\trans{\foReff}_{n,  s}$ may not be syntactically the same as  $\Reff_s(\varphi, \pi)$, it is not hard to see that $\PVO \vdash \forall n \forall s \left(\foTautf(\trans{\foReff}_{n,s}) \leftrightarrow \foTautf \left(\ulcorner {\Reff_s(\varphi)} \urcorner \right) \right)$, making them effectively equivalent.

    Since the proofs generated by Cook's translation are all polynomial-size and uniform, we conclude that there exists a polynomial $\ell(n,m,s,t)$ such that for every $s\geq n^3$,
    \begin{equation}
        \EF \vdash_{\ell(n,m,s,t)} \Reff_{t}(\Reff_s(\varphi, \tau), \pi) \to \Satf(\varphi, E_{n,m, s, t}(\varphi, \pi))\,,
    \end{equation}
where $E_{n,m,s,t}$ is the polynomial-size circuit encoding the computation of the extraction algorithm $E$ on CNF formulas with $n$ variables and $m$ clauses, size parameter $s$ and the refutation being analyzed consist of $t$ clauses.
\end{proof}

Similarly, the upper bound construction of Pudlák, which we formalized in Resolution, clearly goes through in stronger systems too (\Cref{thm:ub-in-Res}).

Suppose now that $A = \{ A_{n, m, s}\}_{n \in \bbN}$ is a family of circuits that automate Resolution, meaning that there is a constant $c \in \bbN$ such that on input a CNF formula $\varphi$ over $n$ variables and $m$ clauses, the circuit $A_{n, m,s}(\varphi)$ outputs a Resolution refutation of $\varphi$ of size $s^c$ if a refutation of size $s$ exists. Consider the propositional formulas stating the correctness of this algorithm,
\begin{equation}
    \Autf^A_{n, m, s} \coloneq \Reff_s(\varphi, \pi) \to \Reff_{s^c}(\varphi, A_{n,m,s}(\varphi)). \label{form:aut-EF} \tag{$\Autf^A$}
\end{equation}

We want to prove that if Extended Frege can efficiently derive the (\ref{form:aut-EF}) formulas above for some sequence of circuits $\{ A_{n, m,s}\}_{n,m,s \in \bbN}$, then it can also show that there is a family of small circuits $\{C_{n}\}_{n\in\bbN}$ solving $\SAT$ on 3-CNF formulas. We encode this as
\[ \Satf^C_{n} \coloneq \Satf(\varphi, \alpha) \to \Satf(\varphi, C_{n}(\varphi)). \label{form:C-sat} \tag{$\Satf^C$}\]

We show that $\EF$ can derive (\ref{form:C-sat}) from (\ref{form:aut-EF}).

\begin{theorem}[Resolution is $\NP$-hard to automate, in $\EF$]
    Suppose Resolution is automatable by a sequence of circuits $A = \{ A_{n,m, s}\}_{n,m,s \in \bbN}$, and suppose $\EF \vdash_{\mathrm{poly}} \Autf^A_{n, m, s}$. Then, there exists a family of circuits $C =\{C_{n}\}_{n\in\bbN}$ of size $n^{O(1)}$ such that $\EF \vdash_{n^{O(1)}} \Satf^C_{n}$.
\end{theorem}

\begin{proof}
    Let $E= \{E_{n,m,s, t}\}_{n,m,s,t \in \bbN}$ be the sequence of polynomial-size circuits carrying out the computation of the extraction algorithm, as in \Cref{thm:lb-formal}, and let $\{P_{n,m, s}\}_{n,m,s\in\bbN}$ be the sequence of circuits constructing the canonical refutations of Pudlák. By \Cref{thm:lb-formal}, \Cref{thm:ub-in-Res} and the assumptions in the statement, we have that Extended Frege can efficiently prove
    \begin{enumerate}[label=(\roman*)]
    \itemsep=0pt
        \item the correctness of the upper bound on $\Reff$ formulas, $\EF \vdash_{\mathrm{poly}} \Satf_{n}(\varphi, \alpha) \to \Reff_{t}(\Reff_s(\varphi), P(\varphi, \alpha, s))$,
        for any $t \geq \tau$, where $\tau=\tau(n,m,s) = \poly(n,m,s)$ is the parameter in \Cref{thm:ub-in-Res};
        \item the correctness of the automating circuits $A$, $
        \EF \vdash_{\mathrm{poly}} \Reff_s(\varphi, \pi) \to \Reff_{s^c}(\varphi, A_{n,m, s}(\varphi))$;
        \item the correctness of the extraction algorithm $E$, $\EF \vdash_{\mathrm{poly}} \Reff_{t}(\Reff_s(\varphi), \pi) \to \Satf_{n}(\varphi, E_{n,m,s, t}(\varphi, \pi))$, as long as $s \geq n^3$.
    \end{enumerate}

    The circuit $C_{n}$ solving $\SAT$ will be $C_{n}(\varphi) \coloneq E_{n,8n^3, n^3, t^c}(\varphi, A_{n,m,t}(\Reff_{n^3}(\varphi)))$ for a large enough $t \geq \tau = \poly(n)$, where $\tau$ is again the parameter in \Cref{thm:ub-in-Res}. Namely, we use the automating circuit $A$ to produce a Resolution refutation from which we can extract a satisfying assignment using the extraction algorithm $E$. In between, the upper bound statement guarantees that a small refutation exists, and thus $A$ will succeed in finding a not much larger one. The parameters are set up so that the three statements above correctly fit with each other: given a 3-CNF formula $\varphi$ with $n$ variables (and at most $8n^3$ clauses), we know that since $\varphi$ is satisfiable, there is a size-$\tau$ refutation of $\Reff_s(\varphi)$, where $\tau = \poly(n, s)$ and $s = n^3$. The extraction algorithm is guaranteed to work when $s \geq n^3$, so we choose to run the automating algorithm on $\Reff_{n^3}(\varphi)$. The automating algorithm is guaranteed to find a refutation of size $t^c$, and the extraction algorithm obtains a satisfying assignment from it. The fact that $C$ is provably correct in $\EF$ then follows immediately by a chain of \emph{modus ponens} on items (i)-(iii) above.
\end{proof}

\section{Universality of $\Reff$ formulas}
\label{sec:Ref-formulas-applications}

The formalizations in \Cref{thm:AM-in-PV} and \Cref{thm:ub-in-Res} have the interesting consequence of tightly relating the provability of arbitrary formulas to the provability of associated Resolution lower bounds. In \Cref{subsec:prop-fragments} we make this precise for Extended Frege and stronger systems. In \Cref{subsec:automatability} we exploit this characterization to show that looking for proofs of tautologies in strong proof systems is equivalent to efficiently looking for proofs of Resolution lower bounds. Finally, \Cref{subsec:unprovability} exploits similar ideas to provide examples of true exponential Resolution lower bounds that are unprovable in different propositional systems and first-order theories.

Our results apply to a wide range of proof systems, assuming some mild conditions on their behaviour. 

\begin{definition}
    We say that a proof system $S$ is \emph{reasonably strong} if it is polynomially equivalent to $\EF + A$ for some set $A$ of tautologies recognizable in polynomial time.
\end{definition}

Recall that for every Cook--Reckhow system $S$, we have $\EF + \Reflf^S \psim S$ (see also \Cref{subsec:prelim-Frege} for the definition of systems of the form $\EF + A$). Reasonably strong proof systems have some features that make them nice to work with. We state them here for convenience.

\begin{proposition}[{\cite[Theorem 2.4.4]{KrajicekPC}}]
    \label{prop:mild-properties}
   The following hold for every reasonably strong proof system $S$:
    \begin{enumerate}[label=(\roman*)]
    \itemsep=0pt
        \item the system $S$ is constructively closed under formula substitutions, i.e., given a proof of $\varphi(x_1, \dots, x_n)$ in size $s$ and formulas $\psi_1, \dots, \psi_n$, we can obtain a proof of $\varphi(\psi_1, \dots, \psi_n)$ in time $\poly(s, n, |\psi_1|, \dots, |\psi_n|)$;
        \item the system $S$ is constructively closed under \emph{modus ponens}, i.e., given an $S$-proof of $\varphi$ in size $s$ and an $S$-proof of $\varphi \to \psi$ in size $t$, we can obtain an $S$-proof of $\psi$ in time $\poly(s, t)$;
        \item the system $S$ is constructively closed under contrapositions, i.e., given an $S$-proof of $\neg \psi$ in size $s$ and an $S$-proof of $\varphi \to \psi$ in size $t$, we can obtain an $S$-proof of $\neg \varphi$ in time $\poly(s, t)$;
        \item the system $S$ \emph{implicationally simulates Resolution}: for every two CNF formulas $\varphi$ and $\psi$ such that Resolution performs the derivation $\varphi \vdash_\Res \psi$ in size $s$, there is an $S$-proof of $\varphi \to \psi$ in size $\poly(s)$.
    \end{enumerate}
\end{proposition}




We need one more definition before we proceed.

\begin{definition}
    Let $\varphi(x_1, \dots, x_n)$ be a CNF formula and let $ 1 \leq k \leq n$. We define $\varphi[k]$ to be the CNF formula $\varphi \land \mathsf{Ext}^k_n$, where $\mathsf{Ext}^k_n$ consists of all clauses encoding the extensions
    \[ y_{\{\ell_1, \dots, \ell_k\}} \leftrightarrow \ell_1 \land \dots \land \ell_k \]
    for every set of literals $\{\ell_1, \dots, \ell_k\}$ over the variables $x_1, \dots, x_n$, where the $y$-variables are fresh variables.
\end{definition}

It is well-known that if $\Res(k)$ refutes $\varphi$ in size $s$, then Resolution refutes $\varphi[k]$ in size $O(ks)$, and if Resolution refutes $\varphi[k]$ in size $s$, then $\Res(k)$ refutes $\varphi$ in size $O(ks)$ \cite[Lemmas 1-2]{AB04}.

The following lemma is an important consequence of the formalization of Pudlák's upper bound.

\begin{lemma}
    \label{lemma:UB-consequence-basic}
    Let $S \geq \Res$ be a propositional proof system closed under literal substitutions, modus ponens, contrapositions, and implicationally simulating Resolution.
    Then, there exists $p(n, m, s) = \poly(n,m, s)$ such that for every $n$-variate CNF formula $\varphi$ with $m$ clauses, and every $t\geq p(n,m,s)$ if $S \vdash \neg \Reff_{t}(\Reff_s(\varphi))$ in size $\ell$, then $S \vdash \neg \varphi[2]$ in size $\poly(n, m, s, t, \ell)$.
\end{lemma}

\begin{proof}
    Assume that $S \vdash \neg \Reff_{t}(\Reff_s(\varphi), \pi)$ in size $\ell$. By \Cref{thm:ub-in-Res}, we have that Resolution has polynomial-size derivations of the form 
    \begin{equation}
    \SAT_{\restriction{\varphi}}(\alpha) \land  \pi = P_{\restriction{\varphi}}(\alpha) \vdash_\Res  \Reff_{t}(\Reff_s(\varphi), \pi)\,,
    \end{equation}
        where $\varphi$ is fixed but $\alpha$ is encoded by free variables, and we define $p(n, m, s) = \tau(n, m, s) = \poly(n, m, s)$ to be the parameter in \Cref{thm:ub-in-Res}. Since $S$ implicationally simulates Resolution, it follows that $S$ has polynomial-size refutations of
        \begin{equation}
    \SAT_{\restriction{\varphi}}(\alpha) \land  \pi = P_{\restriction{\varphi}}(\alpha)\to  \Reff_{t}(\Reff_s(\varphi), \pi)\,.
    \end{equation}
        Since $S$ is also closed under contrapositions, we have that there are polynomial-size $S$-proofs of
        \begin{equation}
\neg\left(\SAT_{\restriction{\varphi}}(\alpha)  \land  \pi = P_{\restriction{\varphi}}(\alpha)\right)\,.
        \end{equation}
        Now, by \Cref{prop:sat-to-native-Res}.(i), we can use a substitution on the variables of $\Satf$ to get a proof of $\neg (\varphi(\alpha)\land \pi = P_{\restriction{\varphi}}(\alpha))$. By the way we defined the clauses in $\pi = P_{\restriction{\varphi}}(\alpha)$ in \Cref{subsec:Pudlak-as-ckt}, once $\varphi$ has been fixed, there are no variables of type $\alit$ left. By inspecting all the extension axiom from~(\ref{L1i-detail}) to~(\ref{R-more-detail-last}) we see that all extension axioms relate variables of $\pi$ to literals over the $\alpha$ variables. Crucially, the width of these extensions is at most two, given by the extensions~(\ref{eq:weakR2alphas1}) and~(\ref{R-more-detail-last}). Furthermore, all possible extensions over two $\alpha$ literals are available, meaning that $ \left(\varphi(\alpha)\land \pi = P_{\restriction{\varphi}}(\alpha) \right) \equiv \varphi \land \mathsf{Ext}^2_n = \varphi[2]$, so this is in fact a proof of $\neg\varphi[2]$. 
\end{proof}

\subsection{Propositional fragments of the Atserias--Müller lower bound}
\label{subsec:prop-fragments}

As a consequence of \Cref{thm:lb-formal} and \Cref{thm:ub-in-Res}, we can characterize the exact fragment of the Atserias--Müller lower bound efficiently provable by every strong enough propositional proof system. Namely, if $\EF$ can argue that $\varphi$ is unsatisfiable, then it can also argue that $\Reff_{n^3}(\varphi)$ is hard for Resolution; and, conversely, if $\Reff_{n^3}(\varphi)$ is provably hard for Resolution, then $\EF$ can argue that $\varphi$ itself is unsatisfiable. The key insight here is that the statement \say{$\Reff_{n^3}(\varphi)$ is hard for Resolution} is exactly a $\Reff$ formula itself, of the form $\Reff(\Reff(\varphi))$.

The following is a formal restatement of \Cref{thm:characterization-informal}.
\begin{theorem}
\label{thm:fragment-characterization}
Let $S$ be a reasonably strong proof system. There is $p(n, m, s) =\poly(n,m, s)$ such that for every CNF formula $\varphi(x_1, \dots, x_n)$ with $m$ clauses, every $t \in \bbN$, and every $s \geq n^3$,
\begin{enumerate}[label=(\roman*)]
\itemsep=0pt
    \item if $S \vdash \neg \varphi$ in size $\ell$, then $S \vdash \neg \Reff_t(\Reff_{s}(\varphi_n))$ in size $\poly(n, m, s, t, \ell)$;
    \item if $S \vdash \neg \Reff_{p(n,m,s)}(\Reff_s(\varphi_n))$ in size $\ell$, then $S \vdash \neg \varphi_n$ in size $\poly(n, m, s, \ell)$.
\end{enumerate}
Furthermore, in both cases, given an $S$-proof of the left-hand side statement, one can obtain a proof of the right-hand side statement in polynomial time.
\end{theorem}

\begin{proof} In the following $\Reff$ formulas $\varphi$ is a fixed formula, while the only free variables are the $\pi$ and $\alpha$ variables.
    \begin{enumerate}[label=(\roman*)]
    \itemsep=0pt
        \item If $S \vdash \neg \varphi$ via some proof $\pi$ of size $\ell$, then by \Cref{prop:sat-to-native-Res}, $S$ also proves $\neg \SAT_{\restriction{\varphi}}(\alpha)$ and we can easily substitute the extraction circuit $E_{\restriction{\varphi}}(\pi)$ into $\alpha$, getting a proof $\pi'$ of $\neg \SAT_{\restriction{\varphi}}(E_{\restriction{\varphi}}(\pi))$. By \Cref{thm:lb-formal}, $\EF$ proves $\Reff_{t}(\Reff_s(\varphi, \tau), \pi) \to \Satf(\varphi, E(\varphi, \pi))$ and by the p-simulation $S \psim \EF$ these proofs are also available in $S$. Hence, by contraposition and the restriction on $\varphi$, $S$ derives $\neg \Reff_{t}(\Reff_s(\varphi, \tau), \pi)$, and it is easy to see that the total blow-up incurred by these additional derivations is at most $\poly(n,m,s,t, \ell)$.

        \item This backwards direction follows from \Cref{lemma:UB-consequence-basic} above. Since $S$ is reasonably strong, the conditions of \Cref{lemma:UB-consequence-basic} apply, giving us polynomial-size refutations of $\varphi[2]$. The system $S$ is in particular closed under clause substitutions, so we can substitute the extension axioms introduced by $\varphi[2]$ to get a refutation of $\varphi$. \qedhere
    \end{enumerate}
\end{proof}

    \subsection{Automatability in terms of $\Reff$ formulas}
\label{subsec:automatability}

We will say that a proof system $S$ is \emph{(weakly) automatable on $\Reff$ formulas} if there is an algorithm $A$ that behaves like an automating algorithm whenever the input formula is a $\Reff_s(\varphi)$ formula for some $\varphi$ and $s$, and is allowed to behave in any other way otherwise.

We remark that these are the $\Reff$ formulas \emph{for Resolution}, meaning that proof search in arbitrarily strong proof systems can be completely characterized by the proof search of Resolution lower bounds! The following theorems are the formal restatements of \Cref{thm:aut-for-Ref-informal}.

\begin{theorem}
    \label{thm:aut-on-Ref-EF}
    Let $S$ be a reasonably strong proof system. Then, the following are equivalent:
    \begin{enumerate}
        \item[(i)] $S$ is automatable;
        \item[(ii)] $S$ is automatable on $\Reff$ formulas.
    \end{enumerate}
\end{theorem}

\begin{proof}
One direction is trivial. For the other implication, let $A$ be an automating algorithm that is only guaranteed to work on tautologies of the form $\neg \Reff_s(\cdot)$, and let $\varphi$ be an unsatisfiable CNF formula with $n$ variables and $m$ clauses. We describe a new algorithm $A'$ that correctly automates $S$ on all inputs, meaning that if $\neg \varphi$ has an $S$-refutation in size $t$, then $A'$ outputs a refutation in time $t^{O(1)}$.

By the formalization of the correctness of the extraction algorithm in \Cref{thm:lb-formal}, we obtained \Cref{thm:fragment-characterization} from where it follows that if $S$ has a refutation of $\varphi$ in size $t$, then there exists an $S$-refutation of $\Reff_{p(n, m)}(\Reff_{n^3}(\varphi))$ for a fixed polynomial $p(n,m)$, in size $\poly(n,m,t)$. This is precisely a $\Reff(\cdot)$ formula, so we can run $A(\Reff_{p(n,m)}(\Reff_{n^3}(\varphi)))$ to obtain a refutation $\pi$ in $S$ with at most a fixed polynomial blow-up. Now, by the backwards direction of \Cref{thm:fragment-characterization}, we known that $S$ has a refutation of $\varphi$ in size $\poly(n,m, t)$ and, we can obtain this efficiently from $\pi$. Thus, we have an automating algorithm that given $\varphi$, which had a refutation in $S$ in size $t$, outputs a refutation in size $t^{O(1)}$.
\end{proof}

Since nothing in the argument above prevents us from talking about proofs in stronger systems, the result immediately carries over to weak automatability. For the following statement, we use the well-known fact that a proof system $S$ is weakly automatable if, and only if, there exists an automatable proof system $Q$ that simulates $S$. If we restrict ourselves to $\Reff$ formulas, then it is easy to see that $S$ is weakly automatable on $\Reff$ formulas if, and only if, there is a proof system $Q$ automatable on $\Reff$ formulas and simulating $S$ over $\Reff$ formulas.

\begin{theorem}
    \label{thm:W-aut-on-Ref-EF}
    Let $S$ be a reasonably strong proof system. Then, the following are equivalent:
    \begin{enumerate}
        \item[(i)] $S$ is weakly automatable;
        \item[(ii)] $S$ is weakly automatable on $\Reff$ formulas.
    \end{enumerate}
\end{theorem}

\begin{proof}
    The argument is identical. The only difference is that instead of an automating algorithm for $S$ we have a proof system $Q$ that simulates $S$ on $\Reff$ formulas and is automatable on $\Reff$ formulas. Then, if $\neg \varphi$ can be refuted in size $t$ in $S$, we can still guarantee (by \Cref{thm:fragment-characterization}) that there is a proof of $\Reff_{p(n,m)}(\Reff_{n^3}(\varphi))$, for a fixed polynomial $p(n,m)$, in size $\poly(n,m,t)$. We look for these proofs with the automating algorithm for $Q$. Since $Q$ itself may not be reasonably strong itself, we turn this $Q$-proof into a proof in the system $\EF + \Reflf^Q$, which p-simulates $Q$, and for which \Cref{thm:fragment-characterization} does apply.
\end{proof}

The argument above crucially requires that $S$ is at least as strong as Extended Frege. The reason is that we require the correctness statement of the extraction algorithm to be provable in $S$. If $\varphi$ has a refutation of size $t$, then \Cref{thm:fragment-characterization}.(i) guarantees that there is a size-$t^{O(1)}$ refutation of $\Reff(\Reff(\varphi))$, so the algorithm can focus on looking for these latter refutations instead. For an arbitrary proof system $S \geq \Res$ that does not simulate $\EF$, if $S$ proves $\neg \varphi$ in size $t$, then we cannot guarantee that it also proves $\neg \Reff(\Reff(\varphi))$ in size~$t^{O(1)}$.

This problem can be solved if $S$ is Resolution itself; then, the semantics of the $\Reff$ formula compensate for the weakness of the system, since now the $\Reff$ formulas are about the proof system itself.

\begin{theorem}
    \label{thm:aut-on-Ref-Res}
    The following hold for Resolution:
    \begin{enumerate}
        \item[(i)] Resolution is automatable if, and only if, it is automatable on $\Reff$ formulas;
        \item[(ii)] Resolution is weakly automatable if, and only if, it is weakly automatable on $\Reff$ formulas.
    \end{enumerate}
\end{theorem}

\begin{proof}
    The first item is easy, while the second one requires a bit more work.
    \begin{enumerate}
        \item[(i)] This first item is a somewhat trivial byproduct of the hardness of automating Resolution \cite{AM20}: if Resolution is automatable on $\Reff$ formulas, this can be used to decide $\SAT$ and hence $\P = \NP$; but then every proof system, including Resolution, is automatable.
        
        \item[(ii)] Suppose Resolution is weakly automatable on $\Reff$ formulas. This means there exists a proof system~$Q$ that simulates Resolution on $\Reff$ formulas and is automatable on $\Reff$ formulas. Now, suppose a CNF formula $\varphi$ with $n$ variables and $m$ clauses has a Resolution refutation in size $t$. Then, run the automating algorithm for $Q$ on formulas of the form $\Reff_{p(n, m, s)}({\Reff_s(\varphi)})$ for $p(n, m,s)$ the polynomial from \Cref{thm:fragment-characterization} for $s = 1$, $s = 2$, and so on, via dovetailing, until a refutation is found in $Q$. Note that, since there exists a refutation of $\varphi$ in size $t$, that means that $\Reff_{t}(\varphi)$ is satisfiable and hence $\Reff_{p(n,m, t)}({\Reff_{t}(\varphi)})$ is unsatisfiable but easy to refute in Resolution. Hence, the algorithm automating $Q$ on $\Reff$ formulas must find a refutation $\pi$ in size $\poly(n, m,t)$.

        At this point we would like to turn this refutation $\pi$ into a refutation of $\varphi$. We can turn $\pi$ into a refutation in $\EF + \Reflf^Q$, which p-simulates $\EF$ and is closed under \emph{modus ponens}. Since \Cref{thm:fragment-characterization} applies to these stronger systems, we just described an algorithm that finds a refutation of $\varphi$ in time $\poly(n,m, t)$ in the system $\EF + \Reflf^Q \psim \Res$, where $t$ is the size of the shortest Resolution refutation. We can conclude that Resolution is weakly automatable. \qedhere
    \end{enumerate}
\end{proof}

By the classical result showing that Resolution is weakly automatable if and only if $\Res(k)$ is weakly automatable for every constant $k \geq 1$ \cite[Theorem 8]{AB04}, we have the following corollary.
\begin{corollary}
    For every $k \geq 1$, $\Res(k)$ is weakly automatable if, and only if, $\Res(k)$ is weakly automatable on $\Reff$ formulas.
\end{corollary}
    
    \subsection{Unprovability of Resolution lower bounds}
\label{subsec:unprovability}

The power of \Cref{thm:fragment-characterization} can be applied to study the provability of Resolution lower bounds. In what follows, we will say that a formula $\neg \Reff_s(\varphi)$ constitutes a \emph{true Resolution lower bound} if $\varphi$ is unsatisfiable and does not have Resolution refutations of size $s$. Could there be a proof system that is polynomially bounded on all such formulas? And, in the first-order setting, is there perhaps a theory of arithmetic capable of proving all true Resolution lower bounds?

\subsubsection{Propositional unprovability}
We show that this is not the case: assuming $S$ is strong enough and not polynomially bounded, there will be true Resolution lower bounds that $S$ cannot derive in polynomial size.

We first exemplify this by giving explicit exponential lower bounds on $\Reff$ formulas for bounded-depth Frege. The following is a formal restatement of \Cref{cor:PHP-Ref-ACO-informal}.

\begin{theorem}[Hard $\Reff$ formulas for bounded-depth Frege]
    For every $d \leq O(\log n / \log \log n)$, there are polynomials $p(n)$ and $q(n)$ such that the formulas in the sequence $ \{ \neg \Reff_{p(n)} (\Reff_{q(n)}(\PHP_n)) \}_{n\in \bbN} $ are all tautological but require size $\exp(\Omega({n^{1/(4d-2)}}))$ to be proven in depth-$d$ Frege systems.
\end{theorem}

\begin{proof}
    The formula $\PHP_n$ has $N=\Theta(n^2)$ variables and $\Theta(n^3)$ clauses. Consider first the formulas $\Reff_{q(n)}(\PHP_n)$. Here, to apply the original bound on $\Reff$ formulas (\Cref{thm:AM-original}), $q(n)$ has to be at least a square in the number of variables of $\PHP_n$, so we define $q(n) \coloneq N^2 =\Theta(n^4)$. 
    Let the size parameter $p(n)$ of the outer $\Reff$ formula be the polynomial $p$ from \Cref{lemma:UB-consequence-basic} applied to $N$ and $q(n)$.
    That the formulas $\Reff_{p(n)}(\Reff_{q(n)}(\PHP_n))$ are tautologies follows form the fact that $\PHP_n$ is unsatisfiable, meaning that the lower bound on $\Reff$ formulas guarantees that $\Reff_{q(n)}(\PHP_n)$ is exponentially hard for Resolution and hence the second nesting of $\Reff$ remains unsatisfiable. The statement then follows from \Cref{lemma:UB-consequence-basic}: if depth-$d$ Frege refutes $\Reff_{p(n)}(\Reff_{q(n)}(\PHP_n))$ in size $\ell$, then depth-$d$ Frege refutes $\PHP_n[2]$ in size $\poly(n,\ell)$. Substituting these $2$-extension axioms increases the depth of every line in the proof by at most one, giving a depth-$(d+1)$ Frege refutation of  $\PHP_n$ in size $\poly(n,\ell)$, but, by \Cref{thm:Hastad-PHP}, depth-$(d+1)$ Frege requires size $\exp({\Omega(n^{1/(4d-2)})})$ to refute $\PHP_n$.
\end{proof}

There is nothing special about the pigeonhole principle in the argument above, and we can generalize this to any reasonably strong system where some lower bound is known. The following is a formal restatement of \Cref{thm:prop-unprovable-informal}.

\begin{theorem}[Propositional unprovability of Resolution lower bounds]
    \label{cor:prop-unprovability}
    For every propositional proof system $S \geq \Res$ implicationally simulating Resolution and closed under contrapositions and clause substitutions, if $S$ is not polynomially bounded, there exists a family of unsatisfiable CNF formulas $\{\psi_n\}_{n\in\bbN}$ on $N = \poly(n)$ variables and of size $|\psi_n| = \poly(n)$ such that
    \begin{enumerate}
        \item[(i)] they require size at least $2^{N^{\Omega(1)}}$ to be refuted in Resolution;
        \item [(ii)] there is a polynomial $p(n)$ such that the formulas $\{ \neg \Reff_{p(n)}(\psi_n)\}_{n\in\bbN}$ are tautological but do not have polynomial-size proofs in $S$.
    \end{enumerate}
\end{theorem}

\begin{proof}
    Since $S$ is not polynomially bounded, there exists a sequence of unsatisfiable 3-CNF formulas $\{ \varphi_n \}_{n\in\bbN}$, each over $n$ variables, that cannot be refuted by $S$ in polynomial size. This means that the formulas of the form $\Reff_{n^2}(\varphi_n)$ are unsatisfiable, or else the satisfying assignment would be a correct Resolution refutation which could be turned into a correct refutation in $S$. By the lower bound on $\Reff$ formulas (\Cref{thm:AM-original}), $\Reff_{n^2}(\varphi_n)$ requires size $2^{\Omega(n)}$ to be refuted in Resolution, meaning that $\Reff_{p(n)}(\Reff_{n^2}(\varphi_n))$ is unsatisfiable for every polynomial $p(n)$. If we choose $p$ to be the polynomial in \Cref{lemma:UB-consequence-basic}, it now follows from this same lemma that $S$ cannot possibly have polynomial-size proofs of the family $\{\neg \Reff_{p(n)}(\Reff_{n^2}(\varphi_n))\}_{n \in \bbN}$, or else there would be polynomial-size $S$-proofs of $\{ \neg \varphi_n[2]\}_{n\in\bbN}$. Since $S$ is closed under clause substitutions, we could substitute the $2$-extensions of $\neg\varphi_n[2]$ and get polynomial-size proofs of $\{ \neg \varphi_n\}_{n\in\bbN}$, a contradiction.
    
    The formula family $\{ \psi_n\}_{n\in\bbN}$ in the statement is then precisely given by $\psi_n \coloneq \Reff_{n^2}(\varphi_n)$. By \Cref{rem:num_of_vars} the formula $\psi_n$ has $N \coloneq \poly(n)$ variables and polynomial size, which yields the rest of the parameters in the statement. \qedhere
\end{proof}

The informal statement in \Cref{thm:prop-unprovable-informal} is obtained in particular by taking $p(n)$ to be $N^2$.

\begin{remark}
    Since every proof system $S$ can be p-simulated by $\EF + \Reflf^S$, and the latter is closed under \emph{modus ponens} and clause substituions, that means that the previous statement applies to every proof system $S$. This implies that, unless $\NP = \coNP$, no proof systems can efficiently derive all true Resolution lower bounds. We note, however, that this follows already from the classical work of \citeauthor{Iwama97} \cite{Iwama97}, who showed that the Proof Size Problem for Resolution ($\PSP_\Res$) is $\NP$-complete. His many-one reduction mapped a 3-CNF formula $\varphi$ over $n$ variables to a new formula $\Phi$ over $O(n \log n)$ variables, such that $\varphi$ is satisfiable if and only if $\Phi$ has a Resolution refutation of size $S(n) = O(n^3\log n)$. In the case when $\varphi$ is unsatisfiable, however, Iwama only proves a lower bound of $S(n) + g(n)$, for some function $g(n) = O(S(n)/n)$. This suffices for his purposes, but it is certainly not enough to get the superpolynomial gap needed for hardness of automatability obtained in~\cite{AM20}. In our statement the parameters achieve almost optimal range: the formulas are exponentially hard, but a slightly subquadratic lower bound is not efficiently provable.
\end{remark}
    \subsubsection{First-order unprovability of Resolution lower bounds}

Our main result here is that a first-order version of the propositional unprovability results above hold  unconditionally in the context of strong enough theories of arithmetic. The formal statement and proof follow.

\begin{theorem}
    \label{thm:fo-unprovable}
    Let $T$ be a consistent first-order extending $\SOT$ with a set of axioms recognizable in polynomial time and admitting $\bbN$ as a model. Then, there exits a sequence of unsatisfiable propositional formulas $\{ \psi_{k, s} \}_{{k, s} \in \bbN}$ described uniformly by a polynomial-time algorithm given $k$ and $s$ in unary, where $\psi_{k,s}$  has $N = \poly(k, s)$ variables, and such that
    \begin{enumerate} \itemsep=0pt
        \item[(i)] Resolution refutations of the formula $\psi_{k, s}$ require size $2^{N^{\Omega(1)}}$;
        \item[(ii)] there are constants $c >0$ and $N_0 \in \bbN$ such that the lower bound expressed by the first-order sentence $\forall k \forall s\forall \pi \left(  N > N_0 \land \operatorname{Ref}_{\Res}(\psi_{k, s}, \pi) \to |\pi| >  N^{c} \right)$ is independent of $T$.
    \end{enumerate}
\end{theorem}

\begin{proof}
    Given the theory $T$, consider the strong proof system of $T$, denoted $P_T$, as defined in \Cref{subsec:prelim-strong-proof-system-of-T}. Since $T$ is polynomial-time axiomatizable, the system $P_T$ is a well-defined Cook--Reckhow proof system. For convenience, we see $P_T$ as a refutation system, and we consider the propositional formulas encoding its reflection principle $\neg\Reflf^{P_T}_{k, s} \coloneq \Reff^{P_T}_{s}(\varphi, \tau) \land \Satf_{k}(\varphi, \alpha)$, as defined in \Cref{def:refl-formulas}, where $k$ denotes the number of variables of $\varphi$ and $s$ is the size of the proofs considered. We assume for convenience that $\varphi$ is always a 3-CNF formula and has hence $O(k^3)$ clauses. Let us also assume that the formula $\neg\Reflf^{P_T}_{k, s}$ itself is written as a 3-CNF formula and consists of $n\coloneq \poly(k, s)$ variables. Consider now the propositional formula $\psi_{k, s} \coloneq \Reff_{n^2}(\neg\Reflf^{P_T}_{k,s})$ stating that there is a size-$n^2$ Resolution refutation of $\neg\Reflf^{P_T}_{k, s}$. It is easy to verify following \Cref{rem:num_of_vars} that  $\psi_{k, s}$ has $N \coloneq \poly(n)$ variables and size $\poly(n)$.
    
    Observe, first, that the sequence $\{ \psi_{k, s} \}_{{k, s}\in\mathbb{N}}$ can be generated uniformly in polynomial time given $k$ and $s$ in unary. Next, note that all the formulas $\psi_{k, s}$ are unsatisfiable: otherwise, there would be size-$n^2$ Resolution refutations of the reflection principle of $P_T$, and this would amount to Resolution polynomially simulating $P_T$. In more detail, suppose $F$ is a 3-CNF formula over $k$ variables with a size-$s$ refutation in $P_T$. If Resolution can refute $\neg\Reflf^{P_T}_{k, s}(\varphi, \tau, \alpha)$ in size $n^2$, then by restricting $\varphi$ by $F$ and $\tau$ by the size-$s$ $P_T$-refutation in question, we get a size-$n^2$ refutation of $\Satf(F, \alpha)$. Resolution can then refute $F$ in its usual \say{native} encoding via \Cref{prop:sat-to-native-Res}, so we get a Resolution refutation of $F$ in size $n^{O(1)} = \poly(k, s)$, concluding $\Res \geq P_T$. This is a contradiction, because since $T \supseteq \SOT$ we know that $P_T$ simulates Extended Frege \cite[Section 4.2, Fact 2]{Pudlak20}, and it is well-known that Resolution is strictly weaker than $\EF$.

    Finally, because all of the formulas $\neg\Reflf^{P_T}_{k, s}$ are unsatisfiable, the lower bound on $\Reff$ formulas (\Cref{thm:AM-original}) guarantees that for large enough $n = \poly(k, s)$, the formula $\psi_{k, s} = \Reff_{n^2}(\neg\Reflf^{P_T}_{k, s})$ requires size $2^{\Omega(n^2/n)} = 2^{\Omega(n)}$ to be refuted in Resolution. Since $\psi_{k, s}$ has $N = \poly(n)$ variables, the size lower bound in terms of $N$ is $2^{N^{\Omega(1)}}$. In particular, that means that for every $c >0$, there is some $N_0 \in \bbN$ such that the much weaker lower bound statement $\forall k \forall s \forall \pi \left( N > N_0 \land |\pi| \leq N^c \to  \neg\foReff(\psi_{k, s}, \pi)\right)$ holds in the standard model $\bbN$, which is a model of $T$. Thus the lower bound is consistent with $T$.

    We now provide a model where the sentence fails. By Gödel's second incompleteness theorem, $T$ cannot prove the soundness of $P_T$ \cite[Section 4.2, Fact 3]{Pudlak20}. Thus, there is a model $M$ of $T$ where the reflection principle of $P_T$ fails. We claim that there is $c > 0$ and a suitable $N_0 \in \bbN$ such that the sentence $\forall k \forall s \forall \pi \left(N > N_0 \land |\pi| \leq  N^{c} \to  \neg\foReff(\psi_{k, s}, \pi)\right)$ fails in this model as well. Indeed, since the reflection of $P_T$ fails in $M$, there exist some nonstandard 3-CNF formula $\varphi_0 \in M \setminus \bbN$, a $P_T$-refutation $\tau_0 \in M \setminus \bbN$ and $\alpha_0\in M \setminus \bbN$  such that $M \models \operatorname{Ref}_{P_T}(\varphi_0, \tau_0) \land \foSatf(\varphi_0, \alpha_0)$. This formula $\varphi_0$ has some nonstandard number of variables $k_0 \in M \setminus \bbN$, and similarly $\tau_0$ has some length $s_0 \coloneq |\tau_0|\in M \setminus \bbN$. We then have that in the model $M$, the formula $\neg\Reflf^{P_T}_{k_0, s_0}(\varphi, \tau, \alpha)$ is satisfiable.
    
    Now, since $M$ is a model of $T$ and $T$ extends $\SOT$, by \Cref{thm:ub-in-Res} we have that the upper bound on $\Reff$ formulas holds in $M$. That is,
    \begin{equation}
    \label{eq:pudlak-fo}
        \begin{split}
        M \models \forall \varphi \forall n \forall m \forall \alpha \forall t \Bigr( \operatorname{CNF}(\varphi, n, m) &\land \foSatf(\varphi, \alpha) \\ &\to \exists \pi \left( |\pi| \leq c_P \cdot t(m+tn^2) \land \operatorname{Ref}_{\Res}( \Reff_t ( \varphi ), \pi)\right) \Bigr)
        \,,
        \end{split}
    \end{equation}
    where $c_P$ is the constant in the Big-Theta term $\Theta(t(m+tn^2))$ bounding the size of upper bound construction in \Cref{thm:ub-in-Res}. (Technically speaking \Cref{thm:ub-in-Res} proved the previous sentence in Resolution, but a simple inspection of the proof reveals that the construction is perfectly uniform and can be immediately formalized in $\SOT$ ---and, for that matter, even in much weaker theories.)
    
    Hence, substituting $\varphi$ for the 3-CNF formula $\neg\Reflf^{P_T}_{k_0, s_0}$ on $n$ variables and $m \leq 8n^3$ clauses, and $t$ for $n^2$, we get
    \begin{equation}
    M \models \exists \pi \left( |\pi| \leq c_P \cdot (8n^5 + n^6) \land \operatorname{Ref}_{\Res}( \psi_{k_0, s_0}, \pi)\right)\,.
        \end{equation}
    Since the number $N$ of variables of $\psi_{k_0, s_0}$ is $N = n^{\Theta(1)}$, the bound on $|\pi|$ is $c_P \cdot (8n^5 + n^6)  \leq N^{O(1)}$ for large enough $N$, it follows that there is $c > 0$ and $N_0 \in \bbN$ such that the first-order sentence $\forall k \forall s \forall \pi \left(N > N_0 \land |\pi| \leq  N^{c} \to  \neg\foReff(\psi_{k, s}, \pi)\right)$ fails in $M$.
\end{proof}

With a couple of extra tricks, we can extend the result to any theory $T$ containing just Robinson Arithmetic rather than all of $\SOT$. We remark as well that there is nothing essential about $T$ being polynomial-time axiomatizable, and one could generalize this to any recursively enumerable extension of $\mathsf{Q}$ via a padding trick \cite{Craig53}. The following is a formal version of \Cref{thm:fo-unprovable-informal}.

\begin{corollary}
    \label{cor:fo-unprovable}
        Let $T$ be a consistent first-order extending $\mathsf{Q}$ with a set of axioms recognizable in polynomial time. Then, there exits a sequence of unsatisfiable propositional formulas $\{ \psi_{k, s} \}_{{k, s} \in \bbN}$ described uniformly by a polynomial-time algorithm given $k$ and $s$ in unary, where $\psi_{k,s}$ has $N = \poly(k, s)$ variables, and such that
    \begin{enumerate} \itemsep=0pt
        \item[(i)] Resolution refutations of the formula $\psi_{k, s}$ require size $2^{N^{\Omega(1)}}$;
        \item[(ii)] there is a constant $c > 0$ and $N_0 \in \bbN$ such that the lower bound expressed by the first-order sentence $\forall k \forall s\forall \pi \left(N > N_0 \land \operatorname{Ref}_{\Res}(\psi_{k, s}, \pi) \to |\pi| >  N^{c} \right)$ is unprovable in $T$, but true in $\bbN$.
    \end{enumerate}
\end{corollary}

\begin{proof}
    Define $T' \coloneq \SOT + \PORefl_T$. Observe that $T'$ satisfies all the conditions of \Cref{thm:fo-unprovable}. First, the theory is consistent because $\bbN \models T'$. This is because $\bbN \models \SOT$ and, as $T$ is consistent and extends $\mathsf{Q}$, which refutes every false $\Pi_1$-sentence, $\PORefl_T$ is a true sentence. Second, $T'$ is clearly recursively axiomatizable because $\SOT$ is, and in particular the axioms can be recognized in polynomial time. Hence, \Cref{thm:fo-unprovable} guarantees a sequence of unsatisfiable propositional formulas $\{ \psi_{k, s} \}_{k, s\in \bbN}$ with the desired properties. In particular, the statement $\forall k \forall s\forall \pi \left(N > N_0 \land |\pi| \leq N^{c} \to  \neg\operatorname{Ref}_{\Res}(\psi_{k, s}, \pi)\right)$ is unprovable in $T'$, where $c$ and $N_0$ are again given by \Cref{thm:fo-unprovable}. On the one hand, this means the sentence is not provable in $T'$. Thus it cannot be proven in $T$ either, since, by construction, $T' = \SOT + \PORefl_T$ is $\Pi_1$-conservative over $T$
    and the sentence is a $\Pi_1$-sentence. On the other hand,
    as argued in the proof of \Cref{thm:fo-unprovable}, the lower bound statement is true in $\bbN$ by the lower bound on $\Reff$ formulas (\Cref{thm:AM-original}).
\end{proof}

\section*{Acknowledgements}
\label{sec:acknowledgements}
\addcontentsline{toc}{section}{\nameref{sec:acknowledgements}}

We are indebted to Ján Pich for insightful initial discussions on the problem and particularly for the idea behind the proof of \Cref{thm:np-hardness-informal}. A sketch of his original idea for the proof, based on a reduction from MCSP, is outlined in Section B.2 of the conference version of this paper \cite{AAdRK25-FOCS}.
We thank Emil \citeauthor{JerabekStack} \cite{JerabekStack} for the technical insight on how to simulate contraposition inferences in Resolution, implicit in the proof of \Cref{thm:ub-in-Res}.
We also thank Jonas Conneryd for useful comments and careful proofreading of a draft of this work, and Anupam Das, Stefan Grosser, Antonina Kolokolova, Jan Krajíček, Jakob Nordström, Kilian Risse, and Rahul Santhanam, as well as the anonymous FOCS reviewers, for helpful comments and suggestions. 

This collaboration started at the
Proof Complexity and Beyond workshop at Mathematisches Forschungs\-institut Oberwolfach, in March 2024.
Part of this work was carried out during the Proof Complexity Workshop at the University of Oxford in September of 2024 and during the \emph{Kaleidoscope de la complexité} spring school in April 2025 at the Centre International de Rencontres Mathématiques (CIRM) in Marseille, France. We also gratefully acknowledge that we have benefited greatly from being part of the Basic Algorithms Research Centre (BARC) environment financed by the Villum Investigator grant~54451.

Noel Arteche was supported by the Wallenberg AI, Autonomous Systems and Software Program (WASP) funded by the Knut and Alice Wallenberg Foundation. Albert Atserias was supported by the Spanish Research Agency through grant PID2022-138506NB-C22 (PROOFS BEYOND) and the Severo Ochoa and María de Maeztu Program for Centers and Units of Excellence in R\&D (CEX2020-001084-M). Susanna F. de Rezende received funding from the Knut and Alice Wallenberg grant \mbox{KAW 2023.0116}, ELLIIT, and the Swedish Research Council grant \mbox{2021-05104}. Erfan Khaniki was supported by the Royal Society University Research Fellowship URF$\backslash$R1$\backslash$211106 \say{Proof complexity and circuit complexity: a unified approach}.


\printbibliography[
        heading=bibintoc
]

@article {BPR00,
    AUTHOR = {Bonet, Maria Luisa and Pitassi, Toniann and Raz, Ran},
     TITLE = {On interpolation and automatization for {F}rege systems},
   JOURNAL = {SIAM J. Comput.},
  FJOURNAL = {SIAM Journal on Computing},
    VOLUME = {29},
      YEAR = {2000},
    NUMBER = {6},
     PAGES = {1939--1967},
   MRCLASS = {03F20 (03D15 68Q15 68Q17)},
  MRNUMBER = {1756400},
MRREVIEWER = {Piotr Wojtylak},
       DOI = {10.1137/S0097539798353230},
}

@article{Pudlak03,
  title={On reducibility and symmetry of disjoint $\mathbf{NP}$ pairs},
  author={Pudl{\'a}k, Pavel},
  journal={Theoretical Computer Science},
  volume={295},
  number={1-3},
  pages={323--339},
  year={2003},
  publisher={Elsevier},
  doi = {https://doi.org/10.1016/S0304-3975(02)00411-5}
}

@article{Haken85,
title = {The intractability of resolution},
journal = {Theoretical Computer Science},
volume = {39},
pages = {297-308},
year = {1985},
note = {Third Conference on Foundations of Software Technology and Theoretical Computer Science},
issn = {0304-3975},
doi = {https://doi.org/10.1016/0304-3975(85)90144-6},
author = {Armin Haken}
}

@inproceedings{dRGNPRS21,
  title={Automating algebraic proof systems is $\mathbf{NP}$-hard},
  author={de Rezende, Susanna F. and G{\"o}{\"o}s, Mika and Nordstr{\"o}m, Jakob and Pitassi, Toniann and Robere, Robert and Sokolov, Dmitry},
  booktitle={Proceedings of the 53rd Annual ACM SIGACT Symposium on Theory of Computing},
  pages={209--222},
  year={2021},
doi = {10.1145/3406325.3451080}
}

@article{AM20,
  title={Automating resolution is $\mathbf{NP}$-hard},
  author={Atserias, Albert and M{\"u}ller, Moritz},
  journal={Journal of the ACM (JACM)},
  volume={67},
  number={5},
  pages={1--17},
  year={2020},
  publisher={ACM New York, NY, USA},
  doi = {10.1145/3409472}
}

@article{CP90,
title = {A feasibly constructive lower bound for resolution proofs},
journal = {Information Processing Letters},
volume = {34},
number = {2},
pages = {81-85},
year = {1990},
doi = {10.1016/0020-0190(90)90141-J},
author = {Stephen Cook and Toniann Pitassi}
}

@book{AB09,
  title={Computational Complexity: A Modern Approach},
  author={Arora, Sanjeev and Barak, Boaz},
  year={2009},
  publisher={Cambridge University Press},
  DOI = {a10.1017/CBO9780511804090}
}

@article {BY24,
    AUTHOR = {Buss, Sam and Yolcu, Emre},
     TITLE = {Regular resolution effectively simulates resolution},
   JOURNAL = {Inform. Process. Lett.},
  FJOURNAL = {Information Processing Letters},
    VOLUME = {186},
      YEAR = {2024},
     PAGES = {Paper No. 106489, 4},
   MRCLASS = {03F20},
  MRNUMBER = {4712753},
       DOI = {10.1016/j.ipl.2024.106489},
}

@article{MP20,
title = {Feasibly constructive proofs of succinct weak circuit lower bounds},
journal = {Annals of Pure and Applied Logic},
volume = {171},
number = {2},
pages = {102735},
year = {2020},
doi = {10.1016/j.apal.2019.102735},
author = {Moritz Müller and Ján Pich}
}

@article {AB04,
    AUTHOR = {Atserias, Albert and Bonet, Mar\'{\i}a Luisa},
     TITLE = {On the automatizability of resolution and related
              propositional proof systems},
   JOURNAL = {Inform. and Comput.},
  FJOURNAL = {Information and Computation},
    VOLUME = {189},
      YEAR = {2004},
    NUMBER = {2},
     PAGES = {182--201},
   MRCLASS = {03B35 (03F20 68N17 68Q25 68T15)},
  MRNUMBER = {2039508},
MRREVIEWER = {Alessandro Berarducci},
       DOI = {10.1016/j.ic.2003.10.004},
}

@book {Krajicek95,
    AUTHOR = {Kraj\'{\i}\v{c}ek, Jan},
     TITLE = {Bounded arithmetic, propositional logic, and complexity
              theory},
 PUBLISHER = {Cambridge University Press, Cambridge},
      YEAR = {1995},
   MRCLASS = {03-02 (03D15 03Fxx 68Q15)},
  MRNUMBER = {1366417},
MRREVIEWER = {Constantine Dimitracopoulos},
       DOI = {10.1017/CBO9780511529948},
}

@book {Krajicek19,
    AUTHOR = {Kraj\'{\i}\v{c}ek, Jan},
     TITLE = {Proof complexity},
 PUBLISHER = {Cambridge University Press, Cambridge},
      YEAR = {2019},
   MRCLASS = {03-02 (03F20 03F30 68Q15)},
  MRNUMBER = {3929744},
MRREVIEWER = {Anahit Artashes Chubaryan},
       DOI = {10.1017/9781108242066},
}

@article {CR79,
    AUTHOR = {Cook, Stephen A. and Reckhow, Robert A.},
     TITLE = {The relative efficiency of propositional proof systems},
   JOURNAL = {J. Symbolic Logic},
  FJOURNAL = {The Journal of Symbolic Logic},
    VOLUME = {44},
      YEAR = {1979},
    NUMBER = {1},
     PAGES = {36--50},
   MRCLASS = {03B05 (03B35 03D15 03F20)},
  MRNUMBER = {523487},
MRREVIEWER = {Richard Statman},
       DOI = {10.2307/2273702},
}

@book {Buss86,
    AUTHOR = {Buss, Samuel R.},
     TITLE = {Bounded arithmetic},
 PUBLISHER = {Bibliopolis, Naples},
      YEAR = {1986},
   MRCLASS = {03F30 (03D15 03D20 03F05 03F35 68Q25)},
  MRNUMBER = {880863},
MRREVIEWER = {Peter Clote},
}

@book {HP93,
    AUTHOR = {H\'{a}jek, Petr and Pudl\'{a}k, Pavel},
     TITLE = {Metamathematics of first-order arithmetic},
 PUBLISHER = {Springer-Verlag, Berlin},
      YEAR = {1993},
   MRCLASS = {03-02 (03D15 03F30 03H15 11U09 11U10)},
  MRNUMBER = {1219738},
MRREVIEWER = {Roman Murawski},
       DOI = {10.1007/978-3-662-22156-3},
}

@article{Oliveira24,
  title = {Meta-Mathematics of Computational Complexity Theory},
  author = {Oliveira, Igor C.},
  journal = {Electronic Colloquium on Computational Complexity (ECCC)},
  number = {TR25-041},
  year = {2025},
  url = {https://eccc.weizmann.ac.il/report/2025/041/},
}

@incollection {Buss98,
    AUTHOR = {Buss, Samuel R.},
     TITLE = {First-order proof theory of arithmetic},
 BOOKTITLE = {Handbook of proof theory},
    SERIES = {Stud. Logic Found. Math.},
    VOLUME = {137},
     PAGES = {79--147},
 PUBLISHER = {North-Holland, Amsterdam},
      YEAR = {1998},
   MRCLASS = {03F30 (03C62 03D15 03H15 68Q15)},
  MRNUMBER = {1640326},
MRREVIEWER = {Constantine Dimitracopoulos},
       DOI = {10.1016/S0049-237X(98)80017-7},
}

@incollection {Buss97,
    AUTHOR = {Buss, Samuel R.},
     TITLE = {Bounded arithmetic and propositional proof complexity},
 BOOKTITLE = {Logic of computation ({M}arktoberdorf, 1995)},
    SERIES = {NATO Adv. Sci. Inst. Ser. F: Comput. Systems Sci.},
    VOLUME = {157},
     PAGES = {67--121},
 PUBLISHER = {Springer, Berlin},
      YEAR = {1997},
   MRCLASS = {03F20 (03F30 68Q15)},
  MRNUMBER = {1492461},
MRREVIEWER = {Jan Kraj\'{\i}\v{c}ek},
}

@article{KPT91,
	author = {Jan Kraj{\'\i}{\v c}ek and Pavel Pudl{\'a}k and Gaisi Takeuti},
	doi = {https://doi.org/10.1016/0168-0072(91)90043-L},
	journal = {Annals of Pure and Applied Logic},
	number = {1},
	pages = {143-153},
	title = {Bounded arithmetic and the polynomial hierarchy},
	volume = {52},
	year = {1991},}

@incollection {Cook75,
    AUTHOR = {Cook, Stephen A.},
     TITLE = {Feasibly constructive proofs and the propositional calculus
              (preliminary version)},
 BOOKTITLE = {Seventh {A}nnual {ACM} {S}ymposium on {T}heory of {C}omputing
              ({A}lbuquerque, {N}.{M}., 1975)},
     PAGES = {83--97},
 PUBLISHER = {Association for Computing Machinery, New York},
      YEAR = {1975},
   MRCLASS = {68A20 (02E10)},
  MRNUMBER = {502226},
MRREVIEWER = {Walter Oberschelp},
}

@inproceedings {Cobham65,
    AUTHOR = {Cobham, Alan},
     TITLE = {The intrinsic computational difficulty of functions},
 BOOKTITLE = {Logic, {M}ethodology and {P}hilos. {S}ci. ({P}roc. 1964
              {I}nternat. {C}ongr.)},
     PAGES = {24--30},
 PUBLISHER = {North-Holland, Amsterdam},
      YEAR = {1965},
   MRCLASS = {02.80},
  MRNUMBER = {207561},
MRREVIEWER = {J. R. B\"{u}chi},
}

@misc{Pudlak20,
      title={Reflection principles, propositional proof systems, and theories}, 
      author={Pavel Pudlák},
      year={2020},
      eprint={2007.14835},
      archivePrefix={arXiv},
}

@inproceedings{AAdRK25-FOCS,
  title={The Proof Analysis Problem}, 
      author={Arteche, Noel and Atserias, Albert and de Rezende, Susanna F. and Khaniki, Erfan},
  booktitle={2025 IEEE 66th Annual Symposium on Foundations of Computer Science (FOCS)},
  pages={2558--2579},
  year={2025},
  organization={IEEE},
  doi={10.1109/FOCS63196.2025.00133}
}

@incollection {Hastad23,
    AUTHOR = {H\aa stad, Johan},
     TITLE = {On small-depth {F}rege proofs for {PHP}},
 BOOKTITLE = {2023 {IEEE} 64th {A}nnual {S}ymposium on {F}oundations of
              {C}omputer {S}cience---{FOCS} 2023},
     PAGES = {37--49},
 PUBLISHER = {IEEE Computer Soc., Los Alamitos, CA},
      YEAR = {2023},
   MRCLASS = {03F20},
  MRNUMBER = {4720252},
       DOI = {10.1109/FOCS57990.2023.00010},
}

@inproceedings{HR22,
  title={On bounded depth proofs for {T}seitin formulas on the grid; revisited},
  author={H{\aa}stad, Johan and Risse, Kilian},
  booktitle={2022 IEEE 63rd Annual Symposium on Foundations of Computer Science (FOCS)},
  pages={1138--1149},
  year={2022},
  organization={IEEE}
}

@article {Ajtai94,
    AUTHOR = {Ajtai, M.},
     TITLE = {The complexity of the pigeonhole principle},
   JOURNAL = {Combinatorica},
  FJOURNAL = {Combinatorica. An International Journal on Combinatorics and
              the Theory of Computing},
    VOLUME = {14},
      YEAR = {1994},
    NUMBER = {4},
     PAGES = {417--433},
       DOI = {10.1007/BF01302964},
}

@article {KPW95,
    AUTHOR = {Kraj\'{\i}\v{c}ek, Jan and Pudl\'{a}k, Pavel and Woods, Alan},
     TITLE = {An exponential lower bound to the size of bounded depth
              {F}rege proofs of the pigeonhole principle},
   JOURNAL = {Random Structures Algorithms},
  FJOURNAL = {Random Structures \& Algorithms},
    VOLUME = {7},
      YEAR = {1995},
    NUMBER = {1},
     PAGES = {15--39},
       DOI = {10.1002/rsa.3240070103},
}

@article {PBI93,
    AUTHOR = {Pitassi, Toniann and Beame, Paul and Impagliazzo, Russell},
     TITLE = {Exponential lower bounds for the pigeonhole principle},
   JOURNAL = {Computational Complexity},
  FJOURNAL = {Computational Complexity},
    VOLUME = {3},
      YEAR = {1993},
    NUMBER = {2},
     PAGES = {97--140},
MRREVIEWER = {Uwe Sch\"{o}ning},
       DOI = {10.1007/BF01200117},
}

@incollection {Iwama97,
    AUTHOR = {Iwama, Kazuo},
     TITLE = {Complexity of finding short resolution proofs},
 BOOKTITLE = {Mathematical foundations of computer science 1997
              ({B}ratislava)},
    SERIES = {Lecture Notes in Comput. Sci.},
    VOLUME = {1295},
     PAGES = {309--318},
 PUBLISHER = {Springer, Berlin},
      YEAR = {1997},
   MRCLASS = {68Q25 (03D15 03F20)},
  MRNUMBER = {1640232},
       DOI = {10.1007/BFb0029974},
}

@article {Krajicek22,
    AUTHOR = {Kraj\'{\i}\v{c}ek, Jan},
     TITLE = {Information in propositional proofs and algorithmic proof
              search},
   JOURNAL = {J. Symb. Log.},
  FJOURNAL = {The Journal of Symbolic Logic},
    VOLUME = {87},
      YEAR = {2022},
    NUMBER = {2},
     PAGES = {852--869},
   MRCLASS = {03F20 (68Q11 68Q30)},
  MRNUMBER = {4438323},
       DOI = {10.1017/jsl.2021.75},
}

@MISC {JerabekStack,
    TITLE = {Modus ponens style inferences in {R}esolution},
    AUTHOR = {Emil Jeřábek},
    HOWPUBLISHED = {Theoretical Computer Science Stack Exchange},
    year = {2025},
    URL = {https://cstheory.stackexchange.com/q/55062}
}

@MISC {AllenderStack,
    TITLE = {$\mathbf{P}$-uniform vs. $\mathbf{DLOGTIME}$-uniform $\mathbf{AC}^0$},
    AUTHOR = {Eric Allender},
    HOWPUBLISHED = {Theoretical Computer Science Stack Exchange},
    year = {2025},
    URL = {https://cstheory.stackexchange.com/q/55461}
}

@inproceedings{ACG24,
  author = {Noel Arteche and Gaia Carenini and Matthew Gray},
  title = {Quantum Automating $\mathbf{TC}^0$-{F}rege Is {LWE}-Hard},
  booktitle = {39th Computational Complexity Conference, {CCC} 2024, July 22-25, 2024, Ann Arbor, MI, {USA}},
  pages = {15:1--15:25},
  year = {2024},
  doi = {10.4230/LIPIcs.CCC.2024.15},
  _bib2doi_selected = {dblp:/rec/conf/coco/ArtecheCG24.bib},
  _bib2doi_confirmed = {true},
  _bib2doi_finished = {true},
}

@article {Craig53,
    AUTHOR = {Craig, William},
     TITLE = {On axiomatizability within a system},
   JOURNAL = {J. Symbolic Logic},
  FJOURNAL = {The Journal of Symbolic Logic},
    VOLUME = {18},
      YEAR = {1953},
     PAGES = {30--32},
      ISSN = {0022-4812},
   MRCLASS = {02.0X},
  MRNUMBER = {55278},
       DOI = {10.2307/2266324},
}

@inproceedings{BP96,
  author = {Beame, P. and Pitassi, T.},
  booktitle = {Proceedings of 37th Conference on Foundations of Computer Science},
  title = {Simplified and improved resolution lower bounds},
  year = {1996},
  pages = {274-282},
  doi = {10.1109/SFCS.1996.548486},
  _bib2doi_selected = {dblp:/rec/conf/focs/BeameP96.bib},
  _bib2doi_confirmed = {true},
  _bib2doi_finished = {true},
}

@inproceedings{PS22,
  author = {Pich, J\'{a}n and Santhanam, Rahul},
  title = {Learning Algorithms Versus Automatability of {F}rege Systems},
  booktitle = {49th International Colloquium on Automata, Languages, and Programming (ICALP 2022)},
  pages = {101:1--101:20},
  year = {2022},
  volume = {229},
  timestamp = {Sun, 06 Oct 2024 01:00:00 +0200},
  biburl = {https://dblp.org/rec/conf/icalp/PichS22.bib},
  bibsource = {dblp computer science bibliography, https://dblp.org},
  doi = {10.4230/LIPIcs.ICALP.2022.101},
  _bib2doi_selected = {dblp:/rec/conf/icalp/PichS22.bib},
  _bib2doi_confirmed = {true},
}

@article{BDGMP04,
  title = {Non-automatizability of bounded-depth {F}rege proofs},
  author = {Bonet, María Luisa and Domingo, Carlos and Gavaldà, Ricard and Maciel, Alexis and Pitassi, Toniann},
  journal = {computational complexity},
  volume = {13},
  pages = {47--68},
  year = {2004},
  publisher = {Springer},
  timestamp = {Mon, 26 Oct 2020 00:00:00 +0100},
  biburl = {https://dblp.org/rec/journals/cc/BonetDGMP04.bib},
  bibsource = {dblp computer science bibliography, https://dblp.org},
  doi = {10.1007/s00037-004-0183-5},
  _bib2doi_selected = {dblp:/rec/journals/cc/BonetDGMP04.bib},
  _bib2doi_confirmed = {true},
}

@article{KP98,
  title = {Some Consequences of Cryptographical Conjectures for $\mathsf{S}_2^1$ and $\mathsf{EF}$},
  author = {Krajíček, Jan and Pudlák, Pavel},
  journal = {Information and Computation},
  volume = {140},
  number = {1},
  pages = {82--94},
  year = {1998},
  doi = {10.1006/inco.1997.2674},
  _bib2doi_selected = {dblp:/rec/journals/iandc/KrajicekP98.bib},
  _bib2doi_confirmed = {true},
  _bib2doi_finished = {true},
}

@article{AR08,
	author = {Alekhnovich, Michael and Razborov, Alexander A.},
	journal = {SIAM Journal on Computing},
	number = {4},
	pages = {1347-1363},
	title = {Resolution Is Not Automatizable Unless $\mathbf{W}[\text{P}]$ Is Tractable},
	volume = {38},
	year = {2008},
    doi = {doi:10.1137/06066850X}
}

@article{GL10,
author = {Galesi, Nicola and Lauria, Massimo},
title = {On the Automatizability of Polynomial Calculus},
year = {2010},
issue_date = {August 2010},
publisher = {Springer-Verlag},
address = {Berlin, Heidelberg},
volume = {47},
number = {2},
doi = {10.1007/s00224-009-9195-5},
journal = {Theor. Comp. Sys.},
month = aug,
pages = {491–506},
numpages = {16}
}

@incollection {Atserias13,
    AUTHOR = {Atserias, Albert},
     TITLE = {The proof-search problem between bounded-width resolution and
              bounded-degree semi-algebraic proofs},
 BOOKTITLE = {Theory and applications of satisfiability testing---{SAT}
              2013},
    SERIES = {Lecture Notes in Comput. Sci.},
    VOLUME = {7962},
     PAGES = {1--17},
      YEAR = {2013},
   MRCLASS = {03B35 (03F20 68Q25)},
  MRNUMBER = {3109018},
       DOI = {10.1007/978-3-642-39071-5\_1},
}

@article {AM11,
    AUTHOR = {Atserias, Albert and Maneva, Elitza},
     TITLE = {Mean-payoff games and propositional proofs},
   JOURNAL = {Inform. and Comput.},
  FJOURNAL = {Information and Computation},
    VOLUME = {209},
      YEAR = {2011},
    NUMBER = {4},
     PAGES = {664--691},
   MRCLASS = {03F20 (91A44)},
  MRNUMBER = {2790836},
MRREVIEWER = {Jan Kraj\'{\i}\v{c}ek},
       DOI = {10.1016/j.ic.2011.01.003},
}

@incollection {PH11,
    AUTHOR = {Huang, Lei and Pitassi, Toniann},
     TITLE = {Automatizability and simple stochastic games},
 BOOKTITLE = {Automata, languages and programming. {P}art {I}},
    SERIES = {Lecture Notes in Comput. Sci.},
    VOLUME = {6755},
     PAGES = {605--617},
 PUBLISHER = {Springer, Heidelberg},
      YEAR = {2011},
   MRCLASS = {68T15 (03B35 03F20 68Q17 68Q25 91A15 91A43)},
  MRNUMBER = {2874140},
MRREVIEWER = {Piotr Faliszewski},
       DOI = {10.1007/978-3-642-22006-7\_51},
}

@article{B20,
  title = {Automating regular or ordered resolution is $\mathbf{NP}$-hard},
  author = {Bell, Zo{\"e}},
  journal = {Electronic Colloquium on Computational Complexity (ECCC)},
  number = {TR20-105},
  year = {2020},
  url = {https://eccc.weizmann.ac.il/report/2020/105},
  _bib2doi_selected = {dblp:/rec/journals/eccc/Bell20.bib},
  _bib2doi_confirmed = {true},
  _bib2doi_finished = {true},
  keywords = {keepurl},
}

@inproceedings{P23,
  author = {Papamakarios, Theodoros},
  title = {Depth-$d$ {F}rege Systems Are Not Automatable Unless $\mathbf{P} = \mathbf{NP}$},
  booktitle = {39th Computational Complexity Conference (CCC 2024)},
  pages = {22:1--22:17},
  year = {2024},
  doi = {10.4230/LIPIcs.CCC.2024.22},
  volume = {300},
  _bib2doi_selected = {dblp:/rec/conf/coco/Papamakarios24.bib},
  _bib2doi_confirmed = {true},
  _bib2doi_finished = {true},
}

@inproceedings{G20,
  title = {Failure of Feasible Disjunction Property for $k$-{DNF} Resolution and $\mathbf{NP}$-Hardness of Automating It},
  author = {Michal Garlík},
  year = {2024},
  doi = {10.4230/LIPIcs.CCC.2024.33},
  booktitle = {39th Computational Complexity Conference, {CCC} 2024, July 22-25, 2024, Ann Arbor, MI, {USA}},
  volume = {300},
  pages = {33:1--33:23},
  _bib2doi_selected = {dblp:/rec/conf/coco/Garlik24.bib},
  _bib2doi_confirmed = {true},
  _bib2doi_finished = {true},
}

@inproceedings{GKSMP20,
  title = {Automating cutting planes is $\mathbf{NP}$-hard},
  author = {G{\"o}{\"o}s, Mika and Koroth, Sajin and Mertz, Ian and Pitassi, Toniann},
  booktitle = {Proceedings of the 52nd Annual ACM SIGACT Symposium on Theory of Computing},
  pages = {68--77},
  year = {2020},
  timestamp = {Mon, 26 Jun 2023 01:00:00 +0200},
  biburl = {https://dblp.org/rec/conf/stoc/GoosKMP20.bib},
  bibsource = {dblp computer science bibliography, https://dblp.org},
  doi = {10.1145/3357713.3384248},
  _bib2doi_selected = {dblp:/rec/conf/stoc/GoosKMP20.bib},
  _bib2doi_confirmed = {true},
  _bib2doi_finished = {true},
}

@inproceedings{IR22,
  title = {Automating {OBDD} proofs is $\mathbf{NP}$-hard},
  author = {Itsykson, Dmitry and Riazanov, Artur},
  booktitle = {47th International Symposium on Mathematical Foundations of Computer Science (MFCS 2022)},
  year = {2022},
  doi = {10.4230/LIPIcs.MFCS.2022.59},
  pages = {59:1--59:15},
  volume = {241},
  _bib2doi_selected = {dblp:/rec/conf/mfcs/ItsyksonR22.bib},
  _bib2doi_confirmed = {true},
  _bib2doi_finished = {true},
}

@inproceedings{MPW19,
  title = {Short proofs are hard to find},
  author = {Mertz, Ian and Pitassi, Toniann and Wei, Yuanhao},
  booktitle = {46th International Colloquium on Automata, Languages, and Programming (ICALP 2019)},
  year = {2019},
  doi = {10.4230/LIPIcs.ICALP.2019.84},
  volume = {132},
  pages = {84:1--84:16},
  _bib2doi_selected = {dblp:/rec/conf/icalp/MertzPW19.bib},
  _bib2doi_confirmed = {true},
  _bib2doi_finished = {true},
}

@InProceedings{MP24,
  author =	{Mazor, Noam and Pass, Rafael},
  title =	{{Gap {MCSP} Is Not (Levin) $\mathbf{NP}$-Complete in Obfustopia}},
  booktitle =	{39th Computational Complexity Conference (CCC 2024)},
  pages =	{36:1--36:21},
  series =	{Leibniz International Proceedings in Informatics (LIPIcs)},
  year =	{2024},
  volume =	{300},
  doi =		{10.4230/LIPIcs.CCC.2024.36}
}

@article {BW01,
    AUTHOR = {Ben-Sasson, Eli and Wigderson, Avi},
     TITLE = {Short proofs are narrow---resolution made simple},
   JOURNAL = {J. ACM},
  FJOURNAL = {Journal of the ACM},
    VOLUME = {48},
      YEAR = {2001},
    NUMBER = {2},
     PAGES = {149--169},
   MRCLASS = {03F20 (03B35 68T15)},
  MRNUMBER = {1868713},
MRREVIEWER = {M. I. Dekhtyar},
       DOI = {10.1145/375827.375835},
}

@article {Thapen16,
    AUTHOR = {Thapen, Neil},
     TITLE = {A tradeoff between length and width in resolution},
   JOURNAL = {Theory Comput.},
  FJOURNAL = {Theory of Computing. An Open Access Journal},
    VOLUME = {12},
      YEAR = {2016},
     PAGES = {Paper No. 5, 14},
   MRCLASS = {03F20 (68Q17)},
  MRNUMBER = {3537405},
MRREVIEWER = {Cl\'{e}ment Aubert},
       DOI = {10.4086/toc.2016.v012a005},
}

@incollection {Garlik19,
    AUTHOR = {Garl\'{\i}k, Michal},
     TITLE = {Resolution lower bounds for refutation statements},
 BOOKTITLE = {44th {I}nternational {S}ymposium on {M}athematical
              {F}oundations of {C}omputer {S}cience},
     PAGES = {Art. No. 37, 13},
      YEAR = {2019},
doi = {10.4230/LIPIcs.MFCS.2019.37},
   MRCLASS = {03F20},
  MRNUMBER = {4008426},
}

@article{LLR24,
  title = {Metamathematics of {R}esolution Lower Bounds: A {TFNP} Perspective},
  author =  {Li, Jiawei and Li, Yuhao and Ren, Hanlin},
  journal = {Electronic Colloquium on Computational Complexity (ECCC)},
  number = {TR24-190},
  year = {2024},
  url = {https://eccc.weizmann.ac.il/report/2024/190/}
}

@inproceedings{ST21,
  title={Iterated lower bound formulas: a diagonalization-based approach to proof complexity},
  author={Santhanam, Rahul and Tzameret, Iddo},
  booktitle={Proceedings of the 53rd Annual ACM SIGACT Symposium on Theory of Computing},
  pages={234--247},
  year={2021},
  doi={https://dl.acm.org/doi/10.1145/3406325.3451010}
}

@inproceedings{KM00,
  title={Is the Standard Proof System for {SAT} P-Optimal?},
  author={K{\"o}bler, Johannes and Messner, Jochen},
  booktitle={International Conference on Foundations of Software Technology and Theoretical Computer Science},
  pages={361--372},
  year={2000},
  organization={Springer},
doi={10.1007/3-540-44450-5_29}
}

@article {FFNR03,
    AUTHOR = {Fenner, Stephen A. and Fortnow, Lance and Naik, Ashish V. and
              Rogers, John D.},
     TITLE = {Inverting onto functions},
   JOURNAL = {Inform. and Comput.},
  FJOURNAL = {Information and Computation},
    VOLUME = {186},
      YEAR = {2003},
    NUMBER = {1},
     PAGES = {90--103},
   MRCLASS = {68Q15 (94A60)},
  MRNUMBER = {2001741},
MRREVIEWER = {Frederic Green},
       DOI = {10.1016/S0890-5401(03)00119-6},
}

@article{BPT14,
author = {Beckmann, Arnold and Pudl\'{a}k, Pavel and Thapen, Neil},
title = {Parity Games and Propositional Proofs},
year = {2014},
issue_date = {April 2014},
publisher = {Association for Computing Machinery},
address = {New York, NY, USA},
volume = {15},
number = {2},
doi = {10.1145/2579822},
journal = {ACM Trans. Comput. Logic},
month = may,
articleno = {17},
numpages = {30},
}

@article{AT24,
  title = {Feasibly Constructive Proof of {S}chwartz-{Z}ippel Lemma and the Complexity of Finding Hitting Sets},
  author = {Atserias, Albert and Tzameret, Iddo},
  journal = {Electronic Colloquium on Computational Complexity (ECCC)},
  number = {TR24-174},
  year = {2024},
  url = {https://eccc.weizmann.ac.il/report/2024/174/}
}

@incollection {Razborov95,
    AUTHOR = {Razborov, Alexander A.},
     TITLE = {Bounded arithmetic and lower bounds in {B}oolean complexity},
 BOOKTITLE = {Feasible mathematics, {II} ({I}thaca, {NY}, 1992)},
    SERIES = {Progr. Comput. Sci. Appl. Logic},
    VOLUME = {13},
     PAGES = {344--386},
 PUBLISHER = {Birkh\"{a}user Boston, Boston, MA},
      YEAR = {1995},
   MRCLASS = {03D15 (03F30 06E30 68Q15)},
  MRNUMBER = {1322282},
MRREVIEWER = {Shih Ping Tung},
}

@phdthesis{jerabek:phd-thesis,
  author = "Emil Je{\v r}{\'a}bek",
  title = "Weak pigeonhole principle, and randomized computation",
  year = 2005,
  school = "Faculty of Mathematics and Physics, Charles University",
  address = "Prague"
}

@inproceedings {CEI96,
    AUTHOR = {Clegg, Matthew and Edmonds, Jeffery and Impagliazzo, Russell},
     TITLE = {Using the {G}roebner basis algorithm to find proofs of
              unsatisfiability},
 BOOKTITLE = {Proceedings of the {T}wenty-eighth {A}nnual {ACM} {S}ymposium
              on the {T}heory of {C}omputing ({P}hiladelphia, {PA}, 1996)},
     PAGES = {174--183},
 PUBLISHER = {ACM, New York},
      YEAR = {1996},
   MRCLASS = {68T15 (68Q35 68Q40)},
  MRNUMBER = {1427512},
       DOI = {10.1145/237814.237860},
}

@article {AD08,
    AUTHOR = {Atserias, Albert and Dalmau, V\'{\i}ctor},
     TITLE = {A combinatorial characterization of resolution width},
   JOURNAL = {J. Comput. System Sci.},
  FJOURNAL = {Journal of Computer and System Sciences},
    VOLUME = {74},
      YEAR = {2008},
    NUMBER = {3},
     PAGES = {323--334},
   MRCLASS = {03F20 (03C13 05C90 91A46)},
  MRNUMBER = {2384078},
MRREVIEWER = {Jan Kraj\'{\i}\v{c}ek},
       DOI = {10.1016/j.jcss.2007.06.025},
}

@article{ParkPham2024,
  author = {Jinyoung Park and Huy Tuan Pham},
  title = {A proof of the {K}ahn–{K}alai conjecture},
  journal = {J. Amer. Math. Soc.},
  volume = {37},
  pages = {235--243},
  year = {2024}, 
  doi = {10.1090/jams/1028}}

@article{AlweissLovettWuZhang2021,
  author = {Ryan Alweiss and Shachar Lovett and Kewen Wu and Jiapeng Zhang},
  title = {Improved bounds for the sunflower lemma},
  journal = {Annals of Mathematics}, 
  volume = {194},
  number = {3}, 
  pages = {795--815},
  year = {2021},
  doi = {10.4007/annals.2021.194.3.5}}

@incollection {AKPS24,
    AUTHOR = {Arteche, Noel and Khaniki, Erfan and Pich, J\'{a}n and Santhanam,
              Rahul},
     TITLE = {From proof complexity to circuit complexity via interactive
              protocols},
 BOOKTITLE = {51st {I}nternational {C}olloquium on {A}utomata, {L}anguages,
              and {P}rogramming},
     PAGES = {Art. No. 12, 20},
      YEAR = {2024},
   MRCLASS = {68Q06 (03F20)},
  MRNUMBER = {4774439},
       DOI = {10.4230/lipics.icalp.2024.12},
}

@misc{PS23,
      title={Towards $\mathbf{P} \neq \mathbf{NP}$ from {E}xtended {F}rege lower bounds}, 
  author={Pich, Jan and Santhanam, Rahul},
      year={2023},
      eprint={2312.08163},
      archivePrefix={arXiv},
}

@inproceedings{Kha24,
  title={Jump Operators, Interactive Proofs and Proof Complexity Generators},
  author={Khaniki, Erfan},
  booktitle={2024 IEEE 65th Annual Symposium on Foundations of Computer Science (FOCS)},
  pages={573--593},
  year={2024},
  organization={IEEE}
}

@article{Krajicek01,
author = {Jan Krajíček},
journal = {Fundamenta Mathematicae},
number = {1-2},
pages = {123-140},
title = {On the weak pigeonhole principle},
doi = {10.4064/fm170-1-8},
volume = {170},
year = {2001},
}

@book {KrajicekPC,
    AUTHOR = {Kraj\'{\i}\v{c}ek, Jan},
     TITLE = {Proof Complexity Generators},
 PUBLISHER = {Cambridge University Press, Cambridge},
      YEAR = {2025},
      doi = {https://doi.org/10.1017/9781009611664},
}

@article{KraInc, 
title={A proof complexity conjecture and the
Incompleteness theorem},
DOI={10.1017/jsl.2023.69}, 
journal={The Journal of Symbolic Logic}, 
author={Kraj\'{\i}\v{c}ek, Jan}, 
year={2023}, 
pages={1–5}}

@article{khainc,
	author = {Erfan Khaniki},
	doi = {10.1017/jsl.2021.99},
	journal = {Journal of Symbolic Logic},
	number = {3},
	pages = {912--937},
	title = {New Relations and Separations of Conjectures About Incompleteness in the Finite Domain},
	volume = {87},
	year = {2022}
}

@InProceedings{khanw,
  author =	{Khaniki, Erfan},
  title =	{{Nisan-Wigderson Generators in Proof Complexity: New Lower Bounds}},
  booktitle =	{37th Computational Complexity Conference (CCC 2022)},
  pages =	{17:1--17:15},
  series =	{Leibniz International Proceedings in Informatics (LIPIcs)},
  year =	{2022},
  volume =	{234},
  doi =		{10.4230/LIPIcs.CCC.2022.17},
}

@article{KraENS, 
title={Extended {N}ullstellensatz proof systems},
DOI={doi.org/10.1090/proc/16709}, 
journal={Proceedings of the American Mathematical Society}, 
author={Kraj\'{\i}\v{c}ek, Jan}, 
year={2024}, 
 volume =	{152},
pages={4881-4892}
}

@article{MoritcForcing,
title = {Typical forcings, $\mathbf{NP}$ search problems and an extension of a theorem of {R}iis},
journal = {Annals of Pure and Applied Logic},
volume = {172},
number = {4},
pages = {102930},
year = {2021},
doi = {doi.org/10.1016/j.apal.2020.102930},
author = {Moritz Müller},
}

@misc{Mykyta22,
      title={Models of Bounded Arithmetic and variants of Pigeonhole Principle}, 
      author={Mykyta Narusevych},
      year={2022},
      eprint={2208.14713},
      archivePrefix={arXiv},
}

@misc{Mykyta24,
      title={An independence of the {MIN} principle from the {PHP} principle}, 
      author={Mykyta Narusevych},
      year={2024},
      eprint={2406.14930},
      archivePrefix={arXiv},
}

@article{azza1,
    author = {Azza Gaysin},
    title = {{H}-coloring Dichotomy in Proof Complexity},
    journal = {Journal of Logic and Computation},
    volume = {31},
    number = {5},
    pages = {1206-1225},
    year = {2021},
    month = {04},
doi={https://doi.org/10.1093/logcom/exab028}
}

@misc{azza2,
      title={Proof complexity of {CSP}}, 
      author={Gaysin, Azza},
      year={2023},
      eprint={2201.00913},
      archivePrefix={arXiv},
}

@misc{azza3,
      title={Proof complexity of universal algebra in a {CSP} dichotomy proof}, 
      author={Gaysin, Azza},
      year={2024},
      eprint={2403.06704},
      archivePrefix={arXiv},
}

@article{UF96,
  title={Simplified lower bounds for propositional proofs},
  author={Fu, Xudong and Urquhart, Alasdair},
  journal={Notre Dame journal of formal logic},
  volume={37},
  number={4},
  pages={523--544},
  year={1996},
  doi ={10.1305/ndjfl/1040046140},
  publisher={Duke University Press}
}

\end{document}